\documentclass{article}
\usepackage{amsmath}
\usepackage{amsthm}
\usepackage{amssymb}
\usepackage{amsfonts}
\begin{document}
\theoremstyle{plain}
\newtheorem{thm}{Theorem}[section]
\newtheorem{lem}[thm]{Lemma}
\newtheorem{prop}[thm]{Proposition}
\newtheorem{cor}[thm]{Corollary}
\theoremstyle{definition}
\newtheorem{assum}[thm]{Assumption}
\newtheorem{notation}[thm]{Notation}
\newtheorem{defn}[thm]{Definition}
\newtheorem{clm}[thm]{Claim}
\newtheorem{ex}[thm]{Example}
\theoremstyle{remark}
\newtheorem{rem}[thm]{Remark}
\newcommand{\unit}{\mathbb I}
\newcommand{\ali}[1]{{\mathfrak A}_{[ #1 ,\infty)}}
\newcommand{\alm}[1]{{\mathfrak A}_{(-\infty, #1 ]}}
\newcommand{\nn}[1]{\lV #1 \rV}
\newcommand{\br}{{\mathbb R}}
\newcommand{\dm}{{\rm dom}\mu}
\newcommand{\Ad}{\mathop{\mathrm{Ad}}\nolimits}
\newcommand{\Proj}{\mathop{\mathrm{Proj}}\nolimits}
\newcommand{\RRe}{\mathop{\mathrm{Re}}\nolimits}
\newcommand{\RIm}{\mathop{\mathrm{Im}}\nolimits}
\newcommand{\Wo}{\mathop{\mathrm{Wo}}\nolimits}
\newcommand{\Prim}{\mathop{\mathrm{Prim}_1}\nolimits}
\newcommand{\Primz}{\mathop{\mathrm{Prim}}\nolimits}
\newcommand{\Class}{\mathop{\mathrm{Class}}\nolimits}
\def\qed{{\unskip\nobreak\hfil\penalty50
\hskip2em\hbox{}\nobreak\hfil$\square$
\parfillskip=0pt \finalhyphendemerits=0\par}\medskip}
\def\proof{\trivlist \item[\hskip \labelsep{\bf Proof.\ }]}
\def\endproof{\null\hfill\qed\endtrivlist\noindent}
\def\proofof[#1]{\trivlist \item[\hskip \labelsep{\bf Proof of #1.\ }]}
\def\endproofof{\null\hfill\qed\endtrivlist\noindent}
\newcommand{\caA}{{\mathcal A}}
\newcommand{\caB}{{\mathcal B}}
\newcommand{\caC}{{\mathcal C}}
\newcommand{\caD}{{\mathcal D}}
\newcommand{\caE}{{\mathcal E}}
\newcommand{\caF}{{\mathcal F}}
\newcommand{\caG}{{\mathcal G}}
\newcommand{\caH}{{\mathcal H}}
\newcommand{\caI}{{\mathcal I}}
\newcommand{\caJ}{{\mathcal J}}
\newcommand{\caK}{{\mathcal K}}
\newcommand{\caL}{{\mathcal L}}
\newcommand{\caM}{{\mathcal M}}
\newcommand{\caN}{{\mathcal N}}
\newcommand{\caO}{{\mathcal O}}
\newcommand{\caP}{{\mathcal P}}
\newcommand{\caQ}{{\mathcal Q}}
\newcommand{\caR}{{\mathcal R}}
\newcommand{\caS}{{\mathcal S}}
\newcommand{\caT}{{\mathcal T}}
\newcommand{\caU}{{\mathcal U}}
\newcommand{\caV}{{\mathcal V}}
\newcommand{\caW}{{\mathcal W}}
\newcommand{\caX}{{\mathcal X}}
\newcommand{\caY}{{\mathcal Y}}
\newcommand{\caZ}{{\mathcal Z}}
\newcommand{\bbA}{{\mathbb A}}
\newcommand{\bbB}{{\mathbb B}}
\newcommand{\bbC}{{\mathbb C}}
\newcommand{\bbD}{{\mathbb D}}
\newcommand{\bbE}{{\mathbb E}}
\newcommand{\bbF}{{\mathbb F}}
\newcommand{\bbG}{{\mathbb G}}
\newcommand{\bbH}{{\mathbb H}}
\newcommand{\bbI}{{\mathbb I}}
\newcommand{\bbJ}{{\mathbb J}}
\newcommand{\bbK}{{\mathbb K}}
\newcommand{\bbL}{{\mathbb L}}
\newcommand{\bbM}{{\mathbb M}}
\newcommand{\bbN}{{\mathbb N}}
\newcommand{\bbO}{{\mathbb O}}
\newcommand{\bbP}{{\mathbb P}}
\newcommand{\bbQ}{{\mathbb Q}}
\newcommand{\bbR}{{\mathbb R}}
\newcommand{\bbS}{{\mathbb S}}
\newcommand{\bbT}{{\mathbb T}}
\newcommand{\bbU}{{\mathbb U}}
\newcommand{\bbV}{{\mathbb V}}
\newcommand{\bbW}{{\mathbb W}}
\newcommand{\bbX}{{\mathbb X}}
\newcommand{\bbY}{{\mathbb Y}}
\newcommand{\bbZ}{{\mathbb Z}}
\newcommand{\str}{^*}
\newcommand{\lv}{\left \vert}
\newcommand{\rv}{\right \vert}
\newcommand{\lV}{\left \Vert}
\newcommand{\rV}{\right \Vert}
\newcommand{\la}{\left \langle}
\newcommand{\ra}{\right \rangle}
\newcommand{\ltm}{\left \{}
\newcommand{\rtm}{\right \}}
\newcommand{\lcm}{\left [}
\newcommand{\rcm}{\right ]}
\newcommand{\ket}[1]{\lv #1 \ra}
\newcommand{\bra}[1]{\la #1 \rv}
\newcommand{\lmk}{\left (}
\newcommand{\rmk}{\right )}
\newcommand{\al}{{\mathcal A}}
\newcommand{\md}{M_d({\mathbb C})}
\newcommand{\Tr}{\mathop{\mathrm{Tr}}\nolimits}
\newcommand{\Ran}{\mathop{\mathrm{Ran}}\nolimits}
\newcommand{\Ker}{\mathop{\mathrm{Ker}}\nolimits}
\newcommand{\spn}{\mathop{\mathrm{span}}\nolimits}
\newcommand{\Mat}{\mathop{\mathrm{M}}\nolimits}
\newcommand{\UT}{\mathop{\mathrm{UT}}\nolimits}
\newcommand{\DT}{\mathop{\mathrm{DT}}\nolimits}
\newcommand{\GL}{\mathop{\mathrm{GL}}\nolimits}
\newcommand{\spa}{\mathop{\mathrm{span}}\nolimits}
\newcommand{\supp}{\mathop{\mathrm{supp}}\nolimits}
\newcommand{\rank}{\mathop{\mathrm{rank}}\nolimits}
\newcommand{\idd}{\mathop{\mathrm{id}}\nolimits}
\newcommand{\ran}{\mathop{\mathrm{Ran}}\nolimits}
\newcommand{\dr}{ \mathop{\mathrm{d}_{{\mathbb R}^k}}\nolimits} 
\newcommand{\dc}{ \mathop{\mathrm{d}_{\cc}}\nolimits} \newcommand{\drr}{ \mathop{\mathrm{d}_{\rr}}\nolimits} 
\newcommand{\zin}{\mathbb{Z}}
\newcommand{\rr}{\mathbb{R}}
\newcommand{\cc}{\mathbb{C}}
\newcommand{\ww}{\mathbb{W}}
\newcommand{\nan}{\mathbb{N}}\newcommand{\bb}{\mathbb{B}}
\newcommand{\aaa}{\mathbb{A}}\newcommand{\ee}{\mathbb{E}}
\newcommand{\pp}{\mathbb{P}}
\newcommand{\wks}{\mathop{\mathrm{wk^*-}}\nolimits}
\newcommand{\he}{\hat {\mathbb E}}
\newcommand{\ikn}{{\caI}_{k,n}}
\newcommand{\mk}{{\Mat_{k+1}}}
\newcommand{\mnz}{\Mat_{n_0}}
\newcommand{\mn}{\Mat_{n}}
\newcommand{\mkk}{\Mat_{k_1+k_2+1}}
\newcommand{\mnzk}{\mnz\otimes \mkk}
\newcommand{\hbb}{H^{k,\bb}_{m,p,q}}
\newcommand{\gb}[1]{\Gamma_{#1,\bb}}
\newcommand{\cgv}[1]{\caG_{#1,\vv}}
\newcommand{\gv}[1]{\Gamma_{#1,\vv}}
\newcommand{\gvt}[1]{\Gamma_{#1,\vv(t)}}
\newcommand{\gbt}[1]{\Gamma_{#1,\bb(t)}}
\newcommand{\cgb}[1]{\caG_{#1,\bb}}
\newcommand{\cgbt}[1]{\caG_{#1,\bb(t)}}
\newcommand{\gvp}[1]{G_{#1,\vv}}
\newcommand{\gbp}[1]{G_{#1,\bb}}
\newcommand{\gbpt}[1]{G_{#1,\bb(t)}}
\newcommand{\Pbm}[1]{\Phi_{#1,\bb}}
\newcommand{\Pvm}[1]{\Phi_{#1,\bb}}
\newcommand{\mb}{m_{\bb}}
\newcommand{\E}[1]{\widehat{\mathbb{E}}^{(#1)}}
\newcommand{\lal}{{\boldsymbol\lambda}}
\newcommand{\Cl}[1]{\Lambda_{\lal}\lmk \unit+ Y\rmk^{#1}}
\newcommand{\cl}[1]{\check\Lambda^{#1}}
\newcommand{\oo}{{\boldsymbol\omega}}
\newcommand{\uu}{{\boldsymbol u}}
\newcommand{\vv}{{\boldsymbol v}}
\newcommand{\bbm}{{\boldsymbol m}}
\newcommand{\kl}[1]{{\mathcal K}_{#1}}
\newcommand{\wb}[1]{\widehat{B_{\mu^{(#1)}}}}
\newcommand{\ws}[1]{\widehat{\psi_{\mu^{(#1)}}}}
\newcommand{\wsn}[1]{\widehat{\psi_{\nu^{(#1)}}}}
\newcommand{\wv}[1]{\widehat{v_{\mu^{(#1)}}}}
\newcommand{\wbn}[1]{\widehat{B_{\nu^{(#1)}}}}
\newcommand{\wo}[1]{\widehat{\omega_{\mu^{(#1)}}}}
\newcommand{\dist}{\dc}
\newcommand{\hpu}{\hat P^{(n_0,k_R,k_L)}_R}
\newcommand{\hpd}{\hat P^{(n_0,k_R,k_L)}_L}
\newcommand{\pu}{ P^{(k_R,k_L)}_R}
\newcommand{\pd}{ P^{(k_R,k_L)}_L}
\newcommand{\opu}{\overline{ P^{(k_R,k_L)}_R}}
\newcommand{\opd}{\overline{ P^{(k_R,k_L)}_L}}
\newcommand{\puuz}{P_{R}^{(n_0-1,n_0-1)}\otimes P^{(k_R,k_L)}_R}
\newcommand{\pddz}{P_{L}^{(n_0-1,n_0-1)}\otimes P^{(k_R,k_L)}_L}
\newcommand{\puu}{\tilde P_U}
\newcommand{\pdd}{\tilde P_D}
\newcommand{\qu}[1]{ Q^{(k_R,k_L)}_{R, #1}}
\newcommand{\qd}[1]{ Q^{(k_R,k_L)}_{L,#1}}
\newcommand{\hqu}[1]{ \hat Q^{(n_0,k_R,k_L)}_{U, #1}}
\newcommand{\hqd}[1]{ \hat Q^{(n_0,k_R,k_L)}_{D,#1}}
\newcommand{\eijn}[1] {E^{(0,k)}_{#1}}
\newcommand{\eijr}[1] {E^{(k_R,0)}_{#1}}
\newcommand{\eijl}[1] {E^{(0,k_L)}_{#1}}
\newcommand{\fiin}[1] {f^{(0,k)}_{#1}}
\newcommand{\eij}[1] {E^{(k_R,k_L)}_{#1}}
\newcommand{\eijz}[1] {E^{(n_0-1,n_0-1)}_{#1}}
\newcommand{\ir}[1]{I_R^{(k_R,k_L)}\lmk #1\rmk}
\newcommand{\il}[1]{I_L^{(k_R,k_L)}\lmk #1\rmk}
\newcommand{\heij}[1] {\hat E^{(k_R,k_L)}_{#1}}
\newcommand{\cn}{\mathop{\mathrm{CN}(n_0,k_R,k_L)}\nolimits}
\newcommand{\ghd}[1]{\mathop{\mathrm{GHD}(#1,n_0,k_1,k_2,\bbG)}\nolimits}
\newcommand{\ghu}[1]{\mathop{\mathrm{GHU}(#1,n_0,k_1,k_2,\bbD)}\nolimits}
\newcommand{\ghdb}[1]{\mathop{\mathrm{GHD}(#1,n_0,k_1,k_2,\bbG)}\nolimits}
\newcommand{\ghub}[1]{\mathop{\mathrm{GHU}(#1,n_0,k_1,k_2,\bbD)}\nolimits}
\newcommand{\hfu}[1]{{\mathfrak H}_{#1}^U}
\newcommand{\hfd}[1]{{\mathfrak H}_{#1}^D}
\newcommand{\hfui}[1]{{\mathfrak H}_{#1,1}^U}
\newcommand{\hfdi}[1]{{\mathfrak H}_{#1,1}^D}
\newcommand{\hfuz}[1]{{\mathfrak H}_{#1,0}^U}
\newcommand{\hfdz}[1]{{\mathfrak H}_{#1,0}^D}
\newcommand{\tl}{\tilde\Lambda}
\newcommand{\tll}[1]{\tilde\Lambda^{#1}}
\newcommand{\hll}[1]{\hat\Lambda^{#1}}
\newcommand{\CN}{\overline{\hpd}\lmk\mnzk \rmk\overline{\hpu}}
\newcommand{\cnz}[1] {\chi_{#1}^{(n_0)}}
\newcommand{\eu}{\eta_{U}^{(k_1,k_2)}}
\newcommand{\ezu}{\eta_{U}^{(n_0-1,n_0-1)}}
\newcommand{\ed}{\eta_{D}^{(k_1,k_2)}}
\newcommand{\ezd}{\eta_{D}^{(n_0-1,n_0-1)}}
\newcommand{\fii}[1]{f_{#1}^{(k_R,k_L)}}
\newcommand{\fiiz}[1]{f_{#1}^{(0,k)}}
\newcommand{\fiz}[1]{f_{#1}^{(n_0-1,n_0-1)}}
\newcommand{\zeij}[1] {e_{#1}^{(n_0)}}
\newcommand{\CLn}{\Class_2(n,n_0,k_1,k_2)}
\newcommand{\ClassA}{\mathop{\mathrm{ClassA}}\nolimits}
\newcommand{\CL}{\ClassA}
\newcommand{\braket}[2]{\left\langle#1,#2\right\rangle}
\newcommand{\abs}[1]{\left\vert#1\right\vert}
\newtheorem{nota}{Notation}[section]
\def\qed{{\unskip\nobreak\hfil\penalty50
\hskip2em\hbox{}\nobreak\hfil$\square$
\parfillskip=0pt \finalhyphendemerits=0\par}\medskip}
\def\proof{\trivlist \item[\hskip \labelsep{\bf Proof.\ }]}
\def\endproof{\null\hfill\qed\endtrivlist\noindent}
\def\proofof[#1]{\trivlist \item[\hskip \labelsep{\bf Proof of #1.\ }]}
\def\endproofof{\null\hfill\qed\endtrivlist\noindent}
\title{A class of asymmetric gapped Hamiltonians on quantum spin chains and its characterization II}
\author{
{\sc Yoshiko Ogata}\footnote{Supported in part by
the Grants-in-Aid for
Scientific Research, JSPS.}\\
{\small Graduate School of Mathematical Sciences}\\
{\small The University of Tokyo, Komaba, Tokyo, 153-8914, Japan}
}

\maketitle{}
\begin{abstract}
We give a characterization of the class of gapped Hamiltonians introduced in PartI \cite{Ogata}.
The Hamiltonians in this class are given as MPS (Matrix product state) Hamiltonians.
In \cite{Ogata}, we list up properties of ground state structures of Hamiltonians in this class.
In this Part II, we show the converse.
Namely, if a (not necessarily MPS) Hamiltonian $H$ satisfies five of the listed properties,
there is a Hamiltonian $H'$ from the class
in \cite{Ogata},
satisfying the followings: The ground state spaces  of the two Hamiltonians on the infinite intervals coincide.
The spectral projections onto the ground state space of $H$ on each finite intervals
are approximated by that of $H'$ exponentially well, with respect to the interval size.
The latter property has an application to the classification problem with open boundary conditions. 
\end{abstract}

\section{Introduction}
In Part I \cite{Ogata}, we introduced a class of MPS Hamiltonians, which allows asymmetric ground state structures.
There, we gave a list of physical properties that the ground states of these Hamiltonians 
satisfy.
In this Part II, conversely, we show that these physical properties 
actually guarantee the ground state structure of the Hamiltonian to be 
captured by the class of Hamiltonians we introduced.
More precisely, we will show if a Hamiltonian satisfies five physical conditions, there is an MPS Hamiltonian from our class
satisfying the followings: {\it 1.} The ground state spaces  of the two Hamiltonians on the infinite intervals coincide.
{\it 2.} The spectral projections onto the ground state space of the original Hamiltonian on finite intervals
are well approximated by that of the MPS one. 
From the latter property we see that two Hamiltonians are in the same class in the classification problem of gapped Hamiltonians.

We use freely the notations and definitions  given in Part I, Subsection 1.1, 1.2, 1.3, and Appendix A.
In particular, recall the definition of $\ClassA$.
In Part I, we studied the properties of ground state structures of MPS Hamiltonians given by
elements in $\ClassA$.
We consider the quantum spin chain described in Subsection 1.1 of Part I.
Let $n\in\nan$ with $n\ge 2$.
\begin{assum}
Let $m\in\nan$ and $h$ a positive element in $\caA_{[0,m-1]}$,
and let $H$ be the Hamiltonian given by $h$
via the formula (1) and (2) of Part I.
We consider the following conditions.
\begin{enumerate}
\item[A1]There exist $N_1,d_1\in\nan$ such that $1\le \dim \ker\lmk H\rmk_{[0,N-1]}\le d_1$
for all $N_1\le N\in\nan$.
\item[A2]  Let $G_N$ be the orthogonal projection onto $\ker\lmk H\rmk_{[0,N-1]}$
acting on $\bigotimes_{i=0}^{N-1}\cc^n$.There exist $\gamma>0$ and $N_2\in\nan$ such that
\[
\gamma\lmk 1-G_N\rmk
\le \lmk H\rmk_{[0,N-1]},\text{ for all } N\ge N_2.
\] 
\item[A3]$\caS_{\bbZ}(H)$ consists of  a unique state $\omega_{\infty}$ on $\caA_{\bbZ}$,
\item[A4]There exist $0<C_1$, $0<s_1<1$, $N_3\in\nan$ and factor states
 $\omega_R\in \caS_{[0,\infty)}(H)$, 
$\omega_L\in \caS_{(-\infty,-1]}(H)$, such that
\begin{align}\label{eq:ltqo}
&\lv
\frac{\Tr_{[0,N-1]}\lmk G_{N} A\rmk}{\Tr_{[0,N-1]}\lmk G_{N}\rmk}-
\omega_R(A)
\rv\le C_1 s_1^{N-l}\lV A\rV,\notag\\
&\lv
\frac{\Tr_{[0,N-1]}\lmk G_{N}\tau_{N-l}\lmk A\rmk\rmk}{\Tr_{[0,N-1]}\lmk G_{N}\rmk}-
\omega_L\circ\tau_{-l}(A)
\rv\le C_1 s_1^{N-l}\lV A\rV,
\end{align}
for all $l\in\nan$, $A\in\caA_{[0,l-1]}$, and $N\ge \max\{l, N_3\}$,
and 
\begin{align}\label{eq:nzm}
\inf\left\{ \sigma\lmk D_{\omega_R\vert_{\caA_{[0,l-1]}}}\rmk\setminus \{0\}\mid l\in\nan\right\}>0,\notag\\
\inf\left\{ \sigma\lmk D_{\omega_L\vert_{\caA_{[-l,-1]}}}\rmk\setminus \{0\}\mid l\in\nan\right\}>0.
\end{align}
\item[A5]For any $\psi\in\caS_{[0,\infty)}(H)$ (resp. 
$\psi\in\caS_{(-\infty,-1]}(H)$), there exists an  $l_{\psi}\in \nan$
such that \[
\lV\psi-\psi\circ\tau_{l_\psi}\rV<2,\quad
(resp. \lV\psi-\psi\circ\tau_{-l_\psi}\rV<2).
\]
\end{enumerate}
\end{assum}
The main Theorem of Part II is the following.
\begin{thm}\label{thm:main}
Let $n\in\nan$ with $n\ge 2$.
Let $m\in\nan$, and $h$ a positive element in $\caA_{[0,m-1]}$. 
Let $H$ be the Hamiltonian given by this $h$.
Assume that [A1]-[A5] hold for this $h$.
Then there exist $\bb\in\ClassA$ and $m_1\in\nan$ such that 
\begin{align}
\caS_{(-\infty,-1]}(H)=\caS_{(-\infty,-1]}(H_{\Phi_{m_1,\bb}}),\quad
\caS_{[0,\infty)}(H)=\caS_{[0,\infty)}(H_{\Phi_{m_1,\bb}}),\quad
\caS_{\bbZ}(H)=\caS_{\bbZ}(H_{\Phi_{m_1,\bb}}).
\end{align}
Furthermore, there exist $C>0$, $0<s<1$, and $N_0\in\nan$ such that
\begin{align}\label{eq:pp}
\lV G_N-G_{N,\bb}\rV\le C s^N,\quad N\ge N_0.
\end{align}
If $\bb$ belongs to $\ClassA$
with respect to $(n_0,k_R,k_L,\lal, \bbD,\bbG,Y)$,
the above $m_1$ satisfies
\[
m_1\ge \max\left\{2l_{\bb}(n,n_0,k_R,k_L,\lal,\bbD,\bbG,Y),\frac{\log\lmk n_0^2(k_L+1)(k_R+1)+1\rmk}{\log n}\right\}.
\]
\end{thm}
The condition [A1] means that the Hamiltonian is frustration free, and 
its ground state dimensions on finite intervals are uniformly bounded.
This is inevitable if one hopes to describe the ground state structure in terms of matrices.
The second condition [A2] means that the Hamiltonian is gapped.  This is the condition under which
we would like to work.
The third one [A3] means that the bulk ground state is unique. We should be able to
extend the Theorem to the case that some discrete symmetry is broken in 
the bulk, but for the simplicity, we assume the uniqueness in the current paper.
The equations (\ref{eq:ltqo}) in [A4] can be called the edge versions of (a relaxed) local topological quantum order introduced in \cite{MP}. It says that the effect of the edge away from
the place the observation is made, decays exponentially fast with respect to the distance.
In \cite{MP}, local topological quantum order was assumed to guarantee the 
stability of the spectral gap. The equations (\ref{eq:nzm}) requires the non-existence
of the zero-mode.
The requirement that $\omega_R,\omega_L$ are factor states can be understood that
these states are in pure phase \cite{BR2}.
[A5]  is rather a technical condition. Recall that the maximal distance that two states can have 
is $2$. This condition [A5] says that for any edge ground state,
there exists a space translation of it whose distance from the original
one is smaller than this maximal number $2$.

From Theorem 1.18 of Part I, $H_{\Phi_{m_1,\bb}}$ with $\bb\in \Class A$
and $m_1$ large enough, satisfies [A1]-[A5].
In this sense, [A1]-[A5] can be understand as a qualitative characterization
of $\ClassA$.

We would like to emphasize that in this approximation (\ref{eq:pp}),
the size of the matrices $\bb$ is fixed, i.e.,
it does not grow with respect to the size of the interval $[0,N-1]$,
nor the precision of the approximation.
The existence of $n$-tuple of matrices 
which describes the bulk ground state of Hamiltonians
satisfying [A1]
is known \cite{Matsui1}, \cite {Matsui2}.
The question here is how to deal with the edge states.

Theorem \ref{thm:main} has an application to the classification problem of gapped
Hamiltonians.
Here, we would like to extend the notion of  $C^1$-classification 
of gapped Hamiltonians with open boundary conditions considered 
in \cite{bo} (which we would like to call $C^1$-classification I 
of gapped Hamiltonians with open boundary conditions, to distinguish from the following one). 
We use the notations and definitions in \cite{bo}.
\begin{defn}[$C^1$-classification II of gapped Hamiltonians with open boundary conditions]\label{def:phasec}
Let $H_0,H_1$ be gapped  Hamiltonians associated with interactions
$\Phi_{0},\Phi_{1}\in{\caJ}$.
We say that  $H_0,H_1$ are 
$C^1$-equivalent of type  II  if the following conditions are satisfied.
\begin{enumerate}
\item There exist $m\in\nan$ and a continuous and piecewise $C^1$-path $\Phi:[0,1]\to {\caJ}_m$ such that $\Phi(0)=\Phi_0$, 
$\Phi(1)= \Phi_1$. Let   $H(t)$ be the Hamiltonian associated with $\Phi(t)$ for each $t\in[0,1]$.
\item  
There are $\gamma>0$, $N_0\in\nan$, and finite intervals $I(t)=[a(t), b(t)]$, $t\in[0,1]$, satisfying the followings:
\begin{description}
\item[(i)] the endpoints $a(t), b(t)$ smoothly depends on $t\in[0,1]$, 
\item[(ii)]there exists a sequence $\{\varepsilon_N\}_{N\in\nan}$
of positive numbers with $\varepsilon_N\to 0$, for $N\to\infty$,
such that $\sigma\lmk H(t)_{[0,N-1]}\rmk\cap I(t)
=\sigma\lmk H(t)_{[0,N-1]}\rmk\cap 
[\lambda(t,N), \lambda(t,N)+\varepsilon_N]$,
and $\sigma\lmk H(t)_{[0,N-1]}\rmk\cap I(t)^c=
\sigma\lmk H(t)_{[0,N-1]}\rmk\cap [b(t)+\gamma,\infty)$
 for all $N\ge N_0$ and $t\in[0,1]$,
where $\lambda(t,N)$
is the smallest eigenvalue  of $H(t)_{[0,N-1]}$.
\end{description}
%
\end{enumerate}
\end{defn}
From Theorem \ref{thm:main}, we obtain the following:
\begin{cor}\label{cor:cl}
Let $n\in\nan$ with $n\ge 2$.
Let $m\in\nan$, and $h$ a positive element in $\caA_{[0,m-1]}$. Let $H$ be the Hamiltonian given by this $h$.
Assume that [A1]-[A5] hold for this $h$.
Let $\bb\in\ClassA$ and $m_1\in\nan$ given in Theorem \ref{thm:main} for this $h$.
Then there exist $\hat \gamma>0$, $0<c_1$, $0<s_1,s_2<1$, and $\hat N_0\in\nan$ such that
\[
\sigma\lmk
(1-t)H_{[0,N-1]}+t\lmk H_{\Phi_{m_1,\bb}}\rmk_{[0,N-1]}
\rmk\cap[\hat \gamma c_1 s_1^{N},\hat \gamma -s_2^{N}\hat \gamma]=\emptyset,\quad
t\in [0,1],\quad N\ge \hat N_0.
\]
In particular, $H$ and $H_{\Phi_{m_1,\bb}}$
are $C^1$-equivalent of type II.
\end{cor}
Let us call $\tilde \caH$ the set of all Hamiltonians satisfying the qualitative conditions
 [A1]-[A5]. The Corollary \ref{cor:cl} means that
for the $C^1$-classification of type II, each equivalence class
with an element from $\tilde \caH$
includes an element  $H_{\Phi_{m_1,\bb}}$
which is given by some $\bb\in\ClassA$.
Therefore, if we consider only $\tilde \caH$,
it suffices to consider MPS Hamiltonians given by $\bb\in\ClassA$.
As $\ClassA$ is given by quite concrete conditions, this is an advantage. 

This article is organized as follows. In Section\ref{sec:ni}, we give representations of $\caS_{(-\infty,-1]}(H)$,
and $\caS_{[0,\infty)}(H)$ by matrices. 
They are given by different $n$-tuples of matrices $\vv_L$, $\vv_R$. 
In our setting, from the theorem of Arveson \cite{arv}, we can find a representation of
the Cuntz algebra associated to the space translation \cite{BJ}. The argument of Matsui \cite{Matsui1}\cite{Matsui2} then
shows that from the representation, we can find an $n$-tuple of matrices which describes the bulk ground state.
In Section\ref{sec:ni}, we overview these materials.
We then start to investigate the structure of these matrices.
In Section \ref{sec:san}, we find that the upper triangular property emerges by the hierarchical structure of 
$\vv_\sigma^*$-invariant subspaces $\caK_{0\sigma}\subsetneq \caK_{1\sigma}\subsetneq\cdots\subsetneq \caK_{b\sigma}$, (Lemma \ref{lem:hie}).
We denote by $p_{a\sigma}$ the orthogonal projection onto $\caK_{a\sigma}$ and 
set $r_{a\sigma}:=p_{a\sigma}-p_{a-1,\sigma}$ and $\omega_{i,a,\sigma}=r_{a\sigma} v_{i\sigma} r_{a\sigma}$.
Each $T_{\oo_{a\sigma}}$ is nonzero due to [A5] (Lemma \ref{lem:oonz}), and
 is an irreducible CP map  (Lemma \ref{lem:hie}).
In Section \ref{sec:yon}, we show that this map (or $T_{\uu_{a\sigma}}$ which is similar to it) is primitive (Lemma \ref{lem:primu}).
This follows from the fact that this $\oo_{a\sigma}$ generates
an element of $\caS_\bbZ(H)$, i.e., from [A3], it generates
$\omega_\infty$ (Lemma \ref{lem:ug}).
Applying a theorem \label{thm:iso} from \cite{fnwpure}, we obtain the primitivity.
More precisely, ${\uu_{a\sigma}}$ are unitarily equivalent to the primitive  $n$-tuple of matrices which gives the minimal representation of $\omega_\infty$.
From this observation, we obtain the structure
$v_{\mu \sigma}\in \Mat_{n_0}\otimes \Mat_{k_\sigma+1}$.
Next, we investigate the $ \Mat_{k_\sigma+1}$ part.
The condition (\ref{eq:nzm}) in [A4] implies 
$\left.\Gamma_{l,\vv_\sigma}^{(\sigma)}\right\vert_{B(\caK_\sigma)r_{0\sigma}}$
is an injection onto $\Gamma_{l,\vv_\sigma}^{(\sigma)}\lmk B(\caK_\sigma)\rmk$,
for $l$ large enough (Lemma \ref{lem:bijecs}).
From this fact, we obtain a basis of 
$\caK_l(\vv^{(\sigma)})$,
$\{y_{a,\alpha,\beta}^{(l,\sigma)}\}$
satisfying the conditions in Lemma \ref{lem:ohy}.
It turns out that these conditions are so restrictive that
we obtain $\lal$, $\bbD$, $\bbG$, $Y$, $\{\hat x_{\mu,b}^{(L)}\}$
and $\{\hat x_{\mu,a}^{(R)}\}$ (Lemma \ref{lem:lrf}).
The key for this procedure is Lemma \ref{lem:key}.
Out of these objects we obtain by the end of Section \ref{sec:roku},
 we construct $\bbB\in\Class A$ in Section \ref{sec:nana}.
This $\bbB$ is obtained  by embedding $\Mat_{k_L+1}$ and $\Mat_{k_R+1}$ into
$\Mat_{k_L+k_R+1}$.
It satisfies $\caS_{[0,\infty)}(H_{\Phi_{m',\bb}})=\caS_{[0,\infty)}(H)$, 
$\caS_{(-\infty,-1]}(H_{\Phi_{m',\bb}})=\caS_{(-\infty,-1]}(H)$, and 
$\omega_\infty=\omega_{\bb,\infty}$.
In the final section, we show that $G_{N,\bbB}$s approximate
$G_N$s exponentially well, with respect to $N$.
The spectral gap and the assumption that the effect of the boundary disappears exponentially fast ([A4]) show (\ref{eq:pp}).

\begin{rem}
In addition to the notations given in Subsection 1.1, 1.2, 1.3, and Appendix A of Part I, we use the notations given in Appendix \ref{sec:nota}.
\end{rem}

\section{Representation of $\caS_{(-\infty,-1]}(H)$, $\caS_{[0,\infty)}(H)$ by matrices}
\label{sec:ni}
In this section, we give a representation of elements in $\caS_{(-\infty,-1]}(H)$, $\caS_{[0,\infty)}(H)$ by matrices.
Readers should be aware that at this point, the matrices representing  $\caS_{(-\infty,-1]}(H)$ and $\caS_{[0,\infty)}(H)$ are different.
Most of the following Lemmas are well known.(See \cite{BJ}\cite{Matsui1}\cite{Matsui2})
\begin{lem}\label{lem:vr}
Let $\sigma=L,R$.
Assume [A1]. Then the followings hold.
\begin{enumerate}
\item For a state $\varrho_\sigma$ on ${\mathcal A}_\sigma$, 
$\varrho_\sigma$ belongs to $\caS_{\sigma}(H)$
if and only if
$\varrho_\sigma(\tau_x(h))=0$ for all $x\in\bbZ^{(\sigma)}$.
For a state $\varrho$ on ${\mathcal A}_{\bbZ}$, 
$\varrho$ belongs to $\caS_{\bbZ}(H)$
if and only if
$\varrho(\tau_x(h))=0$ for all $x\in\bbZ$.
\item  $\caS_\sigma(H)$ is a $wk*$-compact convex face in the set of all states on $\caA_\sigma$.
\item There exists a pure state $\varphi_\sigma$ in $\caS_{\sigma}(H)$.
\end{enumerate}
\end{lem}
\begin{proof}
The proof of {\it 1.} is the same as the proof of Lemma 3.10 in PartI \cite{Ogata}. {\it 2} is clear from {\it 1}. {\it 3} follows from {\it 2} and the Klein-Milman theorem.
\end{proof}
\begin{rem}\label{rem:zp}
From Lemma \ref{lem:vr},  we choose and fix one pure state $\varphi_\sigma$ in $\caS_{\sigma}(H)$ from now on.
\end{rem}

\begin{lem}\label{lem:a1a4}Let $\sigma=L,R$.
Assume [A1] and [A4].
Then we have the followings.
\begin{enumerate}
\item For $d_1\in\nan$ in [A1] and the state $\omega_\sigma$ given in [A4], 
we have $\varrho_\sigma\le d_1\cdot\omega_{\sigma}$ for any $\varrho_\sigma\in\caS_{\sigma}(H)$.
\item Any elements in $\caS_{\sigma}(H)$ are mutually quasi-equivalent.
\end{enumerate}\end{lem}
\begin{proof}
To prove {\it 1.}, let $D_{\varrho_R\vert_{\caA_{[0,N-1]}}}$ be the reduced density matrix of  $\varrho_R\in\caS_{R}(H)$, 
on $\caA_{[0,N-1]}$.
Then, by [A1] and {\it 1.} of Lemma \ref{lem:vr}, we have
\[
0\le\varrho_R\lmk A\rmk= \Tr\lmk D_{\varrho_R\vert_{\caA_{[0,N-1]}}} A\rmk
\le \Tr G_N A\le d_1\frac{\Tr G_NA}{Tr G_N},
\]
for any $l\le N$ and $A\in\caA_{[0,l-1],+}$.
Taking $N\to\infty$ limit, and from [A4], we obtain 
\[
0\le\varrho_R\lmk A\rmk\le d_1\omega_R(A),
\]
for any $l\in\nan$ and $A\in\caA_{[0,l-1],+}$, proving {\it 1.} for $\sigma=R$.
The same argument proves {\it 1.} for $\sigma=L$.
From {\it 1.}, any $\varrho_\sigma\in\caS_{\sigma}(H)$ is $\omega_{\sigma}$-normal.
As $\omega_{\sigma}$ is a factor state, this means that $\varrho_\sigma$ and $\omega_\sigma$ are quasi-equivalent, proving {\it 2.}
\end{proof}
The following  (except for {\it 6}) is the list of Lemmas proven in \cite{Matsui1} and \cite{Matsui2}.
\begin{lem}\label{lem:fs}Let $\sigma=L,R$.
Assume [A1] and [A4]. 
Let $\varphi_\sigma\in\caS_{\sigma}(H)$ be the pure state fixed in Remark \ref{rem:zp}, and $(\caH_{\sigma},\pi_{\sigma},\Omega_{\sigma})$
its GNS triple.
Define the subspace $\caK_{\sigma}$ by
\[
\caK_{\sigma}:=\bigcap_{x\in\bbZ^{(\sigma)}}\ker\pi_{\sigma}\lmk\tau_x(h)\rmk\subset\caH_{\sigma}
\]
Then the followings hold.
\begin{enumerate}
\item $1\le \dim\caK_{\sigma}\le d_1$.
\item $\varphi_\sigma$ and $\varphi_\sigma\circ\tau^{(\sigma)}_{x}$ are quasi-equivalent for all $x\in\nan$.
\item 
There exists $S_{i,\sigma}\in B(\caH_{\sigma})$, $i=1,\ldots,n$ such that
\begin{align*}
S_{i,\sigma}^*S_{j,\sigma}=\delta_{ij},\quad
\sum_{j=1}^n S_{j,\sigma} \pi_{\sigma}\lmk A\rmk {S_{j,\sigma}}^*=\pi_\sigma\circ\tau^{(\sigma)}_{1}\lmk A\rmk,\quad A\in \caA_{\sigma}.
\end{align*}
with
\begin{align*}
\pi_{R}\lmk\bigotimes_{k=0}^{l-1}e_{i_k,j_k}^{(n)}\rmk
=S_{(i_0,R)}\cdots S_{(i_{l-1},R)}S_{(j_{l-1},R)}^*\cdots S_{(j_0,R)}^*,\quad \text{if } \sigma=R,\\
\pi_{L}\lmk\bigotimes_{k=-l}^{-1}e_{i_k,j_k}^{(n)}\rmk
=S_{(i_{-1},L)}\cdots S_{(i_{-l},L)}S_{(j_{-l},L)}^*\cdots S_{(j_{-1},L)}^*,
\quad \text{if } \sigma=L,
\end{align*}for all $l\in\nan$, $i_k,j_k\in\{1,\ldots,n\}$.
\item For $\{S_{i,\sigma}\}$ in {\it 3}, we have $S_{i,\sigma}^*\caK_{\sigma}\subset \caK_{\sigma}$, $i=1,\ldots,n$.
\item 
There exists one to one correspondence between 
 $\psi\in\caS_{\sigma}(H)$ and  a density matrix $\rho_\psi$ in $\caH_\sigma$ with support in $\caK_{\sigma}$
via
\[
\psi\lmk A\rmk=\Tr\lmk\rho_{\psi}\pi_{\sigma}(A)\rmk,\quad A\in\caA_{\sigma}.
\]
In this correspondence, $\psi$ is pure if and only if $\rho_{\psi}$ is rank one.
\item
For $\omega_\sigma$ in [A4],
$\rho_{\omega_\sigma}$ is a strictly positive element of
$B(\caK_{\sigma})$. 
\end{enumerate}
\end{lem}
\begin{proof}
Proof of {\it 1} is the same as that of Proposition 5.1 of \cite{Matsui2}.
{\it 2.} is due to the fact that  $\varphi_\sigma\circ\tau^{(\sigma)}_{x}\in \caS_{\sigma}(H)$ and Lemma \ref{lem:a1a4}.
From {\it 2}, we may apply Lemma \ref{lem:srep} to $\varphi_{\sigma}$ and obtain {\it 3,4},
as in the proof of Theorem 1.2 \cite{Matsui1}.
By Lemma \ref{lem:a1a4} {\it 2}. any $\psi\in\caS_{\sigma}(H)$ is quasi-equivalent to
$\varphi_{\sigma}$. Therefore, it can be represented by a density matrix $\rho_{\psi}$ on $\caH_{\sigma}$.
However, as $\psi(\tau_x(h))=0$ for all $x\in\bbZ^{(\sigma)}$,
the support of $\rho_{\psi}$ is in $\caK_{\sigma}$.
Conversely, if $\rho$ is a density matrix in $\caH_\sigma$ with support in $\caK_{\sigma}$, then the state given by $\Tr\rho\pi_{\sigma}(\cdot)$
belongs to $\caS_{\sigma}(H)$.
As $\varphi_{\sigma}$ is pure, we have $\pi_{\sigma}(\caA_\sigma)^{''}=B(\caH_{\sigma})$.
Therefore, the correspondence above is one to one, and the statement about the purity holds.
To prove {\it 6}, note that if $\rho_{\omega_\sigma}$ is not
strictly positive in $B(\caK_\sigma)$, then there exists a unit vector $\xi\in\caK_\sigma$ which is orthogonal to
$s(\rho_{\omega_\sigma})$.
By {\it 5}, this  $\xi$ defines a state
$\omega_{\xi}=\braket{\xi}{\pi_\sigma\lmk\cdot\rmk\xi}\in\caS_\sigma(H)$.
Let $p$ be the orthogonal projection onto $\xi$.
As $\varphi_\sigma$ is pure, we have $\pi_\sigma(\caA_\sigma)^{''}=B(\caK_\sigma)$.
Therefore, by Kaplansky's density Theorem, there exists
a net $\{x_\alpha\}_{\alpha}$ in the unit ball of $\caA_{\sigma,+}$,
such that $\pi_{\sigma}\lmk x_\alpha\rmk\to p$ in the $\sigma w$-topology.
For this net, we have $\lim_{\alpha}\omega_\sigma(x_\alpha)=0$
and $\lim_\alpha\omega_{\xi}(x_\alpha)=1$.
This contradicts $\omega_{\xi}\le d_1\cdot\omega_{\sigma}$
given in Lemma \ref{lem:a1a4}.
\end{proof}
\begin{notation}\label{nota:vv}Assume [A1] and [A4]. 
Let $\sigma=L,R$ and $\varphi_\sigma$ be the pure state fixed in Remark \ref{rem:zp}.
Let $\caK_\sigma$ be the finite dimensional Hilbert space and $\{S_{i,\sigma}\}_{i=1}^n\subset B(\caH_{\sigma})$ given in
Lemma \ref{lem:fs}.
By {\it 4} of Lemma \ref{lem:fs}, we can define $v_{i,\sigma}\in B(\caK_{\sigma})$
by $v_{i,\sigma}^*:=S_{i,\sigma}^*\vert_{\caK_{\sigma}}$, for $i=1,\ldots,n$.
We also set $m_\sigma:=\dim\caK_\sigma$, and define $P_{\caK_\sigma}$ to be the orthogonal projection
onto $\caK_\sigma$ on $\caH_\sigma$.
We constantly identify $B(\caK_\sigma)$, $P_{\caK_\sigma}B(\caH_\sigma)P_{\caK_\sigma}$,
and $\Mat_{m_\sigma}$.
\end{notation}
With this notation, {\it 3,5} of Lemma \ref{lem:fs} can be rephrased as follows.:
\begin{lem}\label{lem:fv}
Assume [A1] and [A4]. Let $\sigma=L,R$.
Then there exist $m_\sigma\in\nan$ and $n$-tuple of matrices $\vv_\sigma=(v_{\mu,\sigma})_{\mu=1}^n \in\Mat_{m_\sigma}^{\times n}$ satisfying the followings.
\begin{enumerate}
\item 
\begin{align*}
\sum_{j=1}^n v_{j\sigma} {v_{j\sigma}}^*=\unit_{\Mat_{m_\sigma}}.
\end{align*}
\item 
There exists one to one correspondence between 
 $\psi\in\caS_{\sigma}(H)$ and  a density matrix $\rho_\psi\in\Mat_{m_\sigma}$ 
via
\begin{align*}
&\psi\lmk\bigotimes_{k=0}^{l-1}e_{i_k,j_k}^{(n)}\rmk
=\Tr\lmk\rho_{\psi}\lmk v_{(i_0,R)}\cdots v_{(i_{l-1},R)}v_{(j_{l-1},R)}^*\cdots v_{(j_0,R)}^*\rmk\rmk,\quad \text{if } \sigma=R,\\
&\psi\lmk\bigotimes_{k=-l}^{-1}e_{i_k,j_k}^{(n)}\rmk
=\Tr\lmk\rho_{\psi}\lmk v_{(i_{-1},L)}\cdots v_{(i_{-l},L)}v_{(j_{-l},L)}^*\cdots v_{(j_{-1},L)}^*\rmk\rmk,\quad \text{if } \sigma=L,
\end{align*}
for all $l\in\nan$, $i_k,j_k\in\{1,\ldots,n\}$.
In this correspondence, $\psi$ is pure if and only if $\rho_{\psi}$ is rank one.
\end{enumerate}
\end{lem}
Hence we obtained the $n$-tuples of matrices $\vv_{R},\vv_{L}$ which have all the information of 
$\caS_R(H)$ and $\caS_L(H)$ respectively. These tuples $\vv_{R},\vv_{L}$  are not equal in general.
The question is how to connect these informations in a way to approximate $G_N$ simultaneously.
This requires further investigations on the properties of  $\vv_{R},\vv_{L}$, which will be carried out in the next three sections.
\section{A sequence of $v_{i\sigma}^*$-invariant subspaces}\label{sec:san}
We start from the following observation.
\begin{lem}\label{lem:li}
Assume [A1] and [A4]. For $\sigma=L,R$, let 
$\vv_\sigma=(v_{1,\sigma},\ldots, v_{n,\sigma})$ be the matrices
given in Notation \ref{nota:vv}.
Define for each $N\in\nan$,
\begin{align}
\caL_{N,\sigma}:=
\spa\left\{
v_{(i_1,\sigma)}\cdots v_{(i_{N},\sigma)}v_{(j_{N},\sigma)}^*\cdots v_{(j_1,\sigma)}^*\mid
i_k,j_k\in\{1,\ldots,n\},\;k=1,\ldots,N
\right\}.
\end{align}
Then we have
\begin{align}\label{eq:li}
\caL_{1,\sigma}\subset \caL_{2,\sigma}\subset\cdots\subset  \caL_{N,\sigma}\subset \caL_{N+1,\sigma}\subset\cdots\subset B(\caK_\sigma), 
\end{align}
and there exists $N_\sigma\in\nan$ such that 
$\caL_{N_\sigma,\sigma}=B(\caK_\sigma)$.
In particular, we have $P_{\caK_\sigma}\pi_\sigma(\caA_\sigma)P_{\caK_\sigma}=B(\caK_\sigma)$.
\end{lem}
\begin{proof}
By Lemma \ref{lem:fv} {\it 1}, we have 
\[
v_{(i_1,\sigma)}\cdots v_{(i_{N},\sigma)}v_{(j_{N},\sigma)}^*\cdots v_{(j_1,\sigma)}^*
=\sum_iv_{(i_1,\sigma)}\cdots v_{(i_{N},\sigma)}v_{(i\sigma)}v_{(i\sigma)}^*v_{(j_{N},\sigma)}^*\cdots v_{(j_1,\sigma)}^*\in  \caL_{N+1,\sigma}.
\]
This proves the inclusion $\caL_{N,\sigma}\subset \caL_{N+1,\sigma}$.

Note that by the definition of $v_{i\sigma}$ and Lemma \ref{lem:fs},
we have
\[
\bigcup_{N=1}^{\infty}\caL_{N,\sigma}=P_{\caK_\sigma}\pi_\sigma\lmk \caA_\sigma\cap {\mathcal A}^{\rm loc}_{\bbZ}\rmk P_{\caK_\sigma}\subset B(\caK_\sigma).
\]
Note that this is a subspace of a finite dimensional vector space $B(\caK_\sigma)$. Therefore, it is 
$\sigma w$-closed.
As $\varphi_\sigma$ is pure, we have $\pi_\sigma( \caA_\sigma\cap {\mathcal A}^{\rm loc}_{\bbZ})^{''}=B(\caH_\sigma)$.
Therefore, $\bigcup_{N=1}^{\infty}\caL_{N,\sigma}$ is $\sigma w$-dense in $B(\caK_\sigma)$.
Hence we obtain $\bigcup_{N=1}^{\infty}\caL_{N,\sigma}=B(\caK_\sigma)$.
Combining this with (\ref{eq:li}), we conclude that
$\caL_{N,\sigma}=B(\caK_\sigma)$ for $N$ large enough.
\end{proof}
\begin{lem}\label{lem:vucp}
Assume [A1], [A3], and [A4]. For $\sigma=L,R$, let 
$\vv_\sigma=(v_{1,\sigma},\ldots, v_{n,\sigma})$ be the matrices given in Notation \ref{nota:vv}.
Then  followings hold.
\begin{enumerate}
\item $T_{\vv_{\sigma}}$ is a unital CP map such that
$\sigma\lmk T_{\vv_\sigma}\rmk\cap\bbT=\{1\}$.
\item 
There exists a $T_{\vv_\sigma}$-invariant state $\rho_\sigma$ on $\Mat_{m_\sigma}$ such that
\begin{align*}
&P_{\{1\}}^{T_{\vv_\sigma}}(\cdot)=\rho_\sigma(\cdot)\unit_{\caK_\sigma},\\
&P_{\{1\}}^{T_{\vv_\sigma}}\lmk P_{\caK_\sigma}\pi_\sigma\lmk A\rmk P_{\caK_\sigma}\rmk
=\omega_{\infty}\lmk A\rmk\unit_{\caK_\sigma},\quad A\in\caA_{\sigma}.
\end{align*}
\end{enumerate}
\end{lem}
\begin{proof}
The map  $T_{\vv_{\sigma}}$ is a unital CP map because of the definition and Lemma \ref{lem:fv}.
In particular, $T_{\vv_{\sigma}}$ is a contraction. Therefore, 
the spectrum of $T_{\vv_{\sigma}}$ is in the closed unit ball of $\cc$ and every Jordan cell of $T_{\vv_{\sigma}}$ corresponding to an
eigenvalue in $\bbT$ has dimension $1$.
From Lemma \ref{lem:fs}, for any unit vector $\xi\in\caK_\sigma$, $\omega_{\xi}:=\braket{\xi}{\pi_\sigma\lmk\cdot\rmk\xi}$
defines an element in $\caS_\sigma(H)$.
Choose any state $\psi$ on $\caA_{(\Gamma_\sigma)^c}$. Define for each $N\in\nan$, a state $\psi_{N,\sigma}$
on $\caA_\bbZ$ by 
$\psi_{N,R}:=\lmk \psi\otimes \omega_\xi\rmk\circ\tau^{(R)}_N$ if $\sigma=R$, under the identification 
$\caA_{\bbZ}\simeq \caA_{(\Gamma_R)^c}\otimes \caA_R$
and
$\psi_{N,L}:=\lmk  \omega_\xi \otimes\psi\rmk\circ\tau^{(L)}_N$ if $\sigma=L$, under the identification 
$\caA_{\bbZ}\simeq  \caA_L\otimes \caA_{(\Gamma_L)^c}$.
For any $x\in\bbZ$, we have $\psi_{N,\sigma}\circ\tau_x(h)=0$ eventually as $N\to\infty$.
Therefore, any $\wks$accumulation point of $\psi_{N,\sigma}$ belongs to $\caS_{\bbZ}(H)$. By [A3], this means
that  $\psi_{N,\sigma}$ converges to $\omega_\infty$ in $\wks$topology as $N\to\infty$.

For any $A\in \caA_\sigma$, by Lemma \ref{lem:fs}, we have
\[
\psi_{N,\sigma}(A)=\omega_\xi\circ\tau^{(\sigma)}_N(A)
=\braket{\xi}{\pi_\sigma\circ\tau^{(\sigma)}_N\lmk A\rmk\xi}
=\braket{\xi}{T_{\vv_\sigma}^N\lmk P_{\caK_\sigma}\pi_\sigma \lmk A\rmk P_{\caK_\sigma}\rmk\xi}.
\]
On the other hand, by the argument above, we have $\lim_N\psi_{N,\sigma}(A)=\omega_\infty(A)$.
Therefore, we have
\begin{align}
\lim_NT_{\vv_\sigma}^N\lmk P_{\caK_\sigma}\pi_\sigma \lmk A\rmk P_{\caK_\sigma}\rmk=\omega_\infty(A)\unit_{\caK_\sigma},\quad
A\in\caA_\sigma.
\end{align}
By Lemma \ref{lem:li}, this means that $\sigma\lmk T_{\vv_\sigma}\rmk\cap\bbT=\{1\}$,
$P_{\{1\}}^{T_{\vv_\sigma}}(B(\caK_\sigma))=\cc\unit_{\caK_\sigma}$, and
\[
P_{\{1\}}^{T_{\vv_\sigma}}\lmk P_{\caK_\sigma}\pi_\sigma\lmk A\rmk P_{\caK_\sigma}\rmk
=\omega_{\infty}\lmk A\rmk\unit_{\caK_\sigma},\quad A\in\caA_{\sigma}.
\]
Furthermore, by the positivity and the unitarity of $T_{\vv_\sigma}$, there 
 exists a $T_{\vv_\sigma}$-invariant state $\rho_\sigma$ on $\Mat_{m_\sigma}\simeq B(\caK_\sigma)$ such that
\begin{align*}
&P_{\{1\}}^{T_{\vv_\sigma}}(\cdot)=\rho_\sigma(\cdot)\unit_{\caK_\sigma}.
\end{align*}
\end{proof}
Next we consider the restriction $(\vv_\sigma)_{s(\rho_\sigma)}$. (Recall the definitions in Subsection 1.2 of Part I.)
\begin{lem}\label{lem:srs}
Assume [A1], [A3], and [A4]. For $\sigma=L,R$, let 
$\vv_\sigma=(v_{1,\sigma},\ldots, v_{n,\sigma})$ be the matirces given in Notation \ref{nota:vv}.
Then  for $\rho_\sigma$ given in Lemma \ref{lem:vucp}, we have the followings.
\begin{enumerate}
\item \begin{align}\label{eq:svsvs}
s(\rho_\sigma)v_{\mu\sigma}=s(\rho_\sigma)v_{\mu\sigma}s(\rho_\sigma),\quad\mu=1,\ldots,n.
\end{align}
\item $T_{(\vv_{\sigma})_{s(\rho_\sigma)}}$ is a unital primitive CP map
on $B(s(\rho_\sigma)\caK_\sigma)$. 
\item Define a linear map
$\bbE^{(\sigma)}:\Mat_n\otimes B(s(\rho_\sigma)\caK_\sigma)\to B(s(\rho_\sigma)\caK_\sigma)$ by
\[
\bbE^{(\sigma)}\lmk e_{\mu\nu}^{(n)}\otimes X\rmk:=\lmk v_{\mu\sigma}\rmk_{s(\rho_\sigma)}X\lmk \lmk v_{\nu\sigma}\rmk_{s(\rho_\sigma)}\rmk^*,\quad X\in B(s(\rho_\sigma)\caK_\sigma).
\] 
Then $(B(s(\rho_\sigma)\caK_\sigma)),\bbE^{(\sigma)},\rho_\sigma\vert_{B(s(\rho_\sigma)\caK_\sigma)})$
is a minimal standard triple $\sigma$-generating $\omega_\infty$.
\end{enumerate}
(See Appendix \ref{app:c} for the definitions.)
\end{lem}
\begin{proof}
{\it 1} : As $\rho_\sigma$ is a $T_{\vv_\sigma}$-invariant state, we have
$
\sum_{\mu=1}^n v_{\mu\sigma}^*\tilde \rho_\sigma v_{\mu\sigma}=\tilde\rho_\sigma
$,
for the density matrix $\tilde\rho_\sigma$ of $\rho_\sigma$.
This implies (\ref{eq:svsvs}).
\\
{\it 2.} : The map $T_{\lmk\vv_\sigma\rmk_{s(\rho_\sigma)}}$ is clearly CP on $B(s(\rho_\sigma)\caK_\sigma)$ but it is also unital because
\[
T_{\lmk\vv_\sigma\rmk_{s(\rho_\sigma)}}(s(\rho_\sigma))
=s(\rho_\sigma)T_{\vv_\sigma}\lmk s(\rho_\sigma)\rmk s(\rho_\sigma)
=s(\rho_\sigma)T_{\vv_\sigma}\lmk \unit_{\caK_\sigma}\rmk s(\rho_\sigma)= s(\rho_\sigma).
\]
Here we used (\ref{eq:svsvs}) for the second equality.
For any $\lambda\in\sigma\lmk T_{(\vv_{\sigma})_{s(\rho_\sigma)}}\rmk\cap\bbT$, there exists
a nonzero $X\in B(s(\rho_\sigma)\caK_\sigma)$ such that $T_{(\vv_{\sigma})_{s(\rho_\sigma)}}(X)=\lambda X$.
Using (\ref{eq:svsvs}) and Lemma \ref{lem:vucp} again, we obtain
\[
\lambda^NX=T_{(\vv_{\sigma})_{s(\rho_\sigma)}}^N(X)=s(\rho_\sigma)T_{\vv_\sigma}^N\lmk X\rmk s(\rho_\sigma)\to\rho_\sigma(X)s(\rho_\sigma),\quad N\to\infty.
\]
From this, we conclude $\lambda=1$ and $X\in\cc s(\rho_\sigma)$.
In other words, we have $\sigma\lmk T_{(\vv_{\sigma})_{s(\rho_\sigma)}}\rmk\cap\bbT=\{1\}$ and
$P_{\{1\}}^{T_{(\vv_{\sigma})_{s(\rho_\sigma)}}}\lmk B(s(\rho_\sigma)\caK_\sigma)\rmk=\cc s(\rho_\sigma)$.
The restriction $\rho_\sigma\vert_{B(s(\rho_\sigma)\caK_\sigma)}$ is $T_{\lmk\vv_\sigma\rmk_{s(\rho_\sigma)}}$-invariant
faithful state on $B(s(\rho_\sigma)\caK_\sigma)$.
Therefore,  $T_{(\vv_{\sigma})_{s(\rho_\sigma)}}$ is primitive from Lemma C.4 of \cite{bo}.
\\
{\it 3} : It is clear from {\it 2}. that  $(B(s(\rho_\sigma)\caK_\sigma),\bbE^{(\sigma)},\rho_\sigma\vert_{B(s(\rho_\sigma)\caK_\sigma)})$
is a standard triple. It is minimal because $T_{(\vv_{\sigma})_{s(\rho_\sigma)}}$ is primitive.
For any $a_1,a_2\in\bbZ$ with $a_1\le a_2$ and $A\in \caA_{[a_1,a_2]}$,
choose $l\in\nan$ so that $\tau_l^{(\sigma)}(A)\in\caA_\sigma$.
We have by (\ref {eq:svsvs}) and Lemma \ref{lem:fs}, \ref{lem:vucp}
\begin{align*}
&\omega_{\infty}\lmk A\rmk\unit_{\caK_\sigma}
=\omega_{\infty}\lmk \tau_l^{(\sigma)}(A)\rmk\unit_{\caK_\sigma}=P_{\{1\}}^{T_{\vv_\sigma}}\lmk P_{\caK_\sigma}\pi_\sigma\lmk \tau_l^{(\sigma)}(A)\rmk P_{\caK_\sigma}\rmk\\
&=\sum_{\mu^{(a_2-a_1+1)},\nu^{(a_2-a_1+1)}}
\braket{\widehat\psi_{\mu^{(a_2-a_1+1)}}}{\tau_{-a_1}\lmk A\rmk\widehat\psi_{\nu^{(a_2-a_1+1)}}}
\rho_\sigma\lmk
\widehat v_{s(\rho_\sigma)),\mu^{(a_2-a_1+1,\sigma)}}
\lmk
\widehat v_{s(\rho_\sigma)),\nu^{(a_2-a_1+1,\sigma)}}
\rmk^*
\rmk\unit_{\caK_\sigma}.
\end{align*}
This means that $(B(s(\rho_\sigma)\caK_\sigma)),\bbE^{(\sigma)},\rho_\sigma\vert_{B(s(\rho_\sigma)\caK_\sigma)})$
 $\sigma$-generates $\omega_\infty$.
\end{proof}
\begin{lem}\label{lem:hie}
Let $\caK$ be a finite dimensional Hilbert space, $n\in\nan$, and $\{v_i\}_{i=1}^n$ a set of elements in $B(\caK)$.
We say a subspace $W$ of $\caK$ is  $\{v_i^*\}_{i=1}^n$-invariant if
$v_i^*W\subset W$ for all $i=1,\ldots,n$.
Let $\caK_0'$ be a  $\{v_i^*\}_{i=1}^n$-invariant subspace of $\caK$. 
Suppose that $\caK_0'$ does not have any proper nonzero subspace which is $\{v_i^*\}_{i=1}^n$-invariant.
Then there exists $k\in\nan\cup\{0\}$ and a finite sequence of subspaces $\{\caK_a\}_{a=0}^k$ of $\caK$ satisfying the followings:
\begin{enumerate}
\item[(i)]
$\caK_0=\caK_0'$.
\item[(ii)]
$\caK_0\subsetneq \caK_1\subsetneq\cdots\subsetneq \caK_k=\caK$.
\item[(iii)]
$v_i^*\caK_a\subset\caK_a$, for any $i=1,\ldots,n$, and $a=0,\ldots,k$.
\item[(iv)]
For any $a=0,\ldots,k-1$,
there is no proper intermediate subspace between $\caK_a$ and $\caK_{a+1}$ which is $\{v_i^*\}_{i=1}^n$-invariant.
. 
\end{enumerate}
Furthermore, set $\caK_{-1}:=\{0\}$ and let $p_a$ be the orthogonal projection onto $\caK_a$,
for $a=-1,\ldots,k$.
Set $r_a:=p_a-p_{a-1}$ and $\oo_a:=(\omega_{i,a})_{i=1}^n$, $\omega_{i,a}=r_a v_i r_a$
for all $a=0,\ldots,k$.
Then for each $a=0,\ldots,k$, $T_{\oo_a}$ is an irreducible CP map on $B(r_a\caK)$ and
\[
\omega_{i_1a}\omega_{i_2a}\cdots\omega_{i_la}=p_av_{i_1}v_{i_2}\cdots v_{i_l}\overline{p_{a-1}},\quad
l\in\nan,\quad i_1,\ldots,i_l\in\{1,\ldots,n\}.
\]
\end{lem}
\begin{rem}
Here, a proper intermediate space between $W_1$ and $W_2$ means a space $W$ such that  $W_1\subsetneq W\subsetneq W_2$.
\end{rem}
\begin{proof}
We consider the following proposition for $b\in\nan\cup\{0\}$:\\\\
$(P_b)$:
There exists a finite sequence of subspaces $\{\caK_a\}_{a=0}^b$ of $\caK$ satisfying the followings	:
\begin{center}
\begin{enumerate}
\item[(i)]
$\caK_0=\caK_0'$.
\item[(ii)]
$\caK_0\subsetneq \caK_1\subsetneq\cdots\subsetneq \caK_b$.
\item[(iii)]
$v_i^*\caK_a\subset\caK_a$, for any $i=1,\ldots,n$, and $a=0,\ldots,b$.
\item[(iv)]
For any $0\le a\le b-1$,
there is no proper intermediate subspace between $\caK_a$ and $\caK_{a+1}$
which is $\{v_i^*\}_{i=1}^n$-invariant.
. 
\end{enumerate}
\end{center}

\noindent 
Clearly $(P_0)$ holds.
Suppose that $(P_b)$ holds and $\caK_b\neq \caK$.
We claim $(P_{b+1})$ holds.
If there is no intermediate subspace which is $\{v_i^*\}_{i=1}^n$-invariant,
between $\caK_b$ and $\caK$, then $\caK_{b+1}:=\caK$ satisfies the condition $(P_{b+1})$.
If there is a proper intermediate subspace $\caH_1$ which is $\{v_i^*\}_{i=1}^n$-invariant,
between $\caK_b$ and $\caK$, then we have
$\dim\caK_b<\dim \caH_1< \dim\caK$ and $\caK_b\subsetneq\caH_1\subsetneq\caK$.
Set $\caK_{b+1}:=\caH_1$ if there is no proper intermediate subspace which is $\{v_i^*\}_{i=1}^n$-invariant,
between $\caK_b$ and $\caH_1$, and $(P_{b+1})$ holds.
Otherwise, there exists a proper intermediate subspace $\caH_2$ which is $\{v_i^*\}_{i=1}^n$-invariant,
between $\caK_b$ and $\caH_1$, and
$\dim \caK_b<\dim\caH_2<\dim\caH_1<\dim\caK$.
This procedure terminates at some point because $\dim\caK$ is finite and 
at each step, the dimension of $\caH_i$ decreases at least $1$.
Suppose that this iteration terminates at $l$-th step, and we obtain a subspace $\caH_l$.
Setting $\caK_{b+1}=\caH_l$, we obtain $(P_{b+1})$.
Note that if $(P_b)$ holds for some $\{\caK_a\}_{a=0}^b$, then
we have $b\le\dim\caK_b\le\dim\caK$.
Therefore, there exists $k\in\nan\cup\{0\}$ and subspaces 
$\{\caK_a\}_{a=0}^k$ satisfying $(P_k)$ and $\caK_k=\caK$.
This proves the first part of the Lemma.

In order to show that 
$T_{\oo_a}$ is an irreducible CP map on $B(r_a\caK)$ for $a=0,\ldots,k$,
it suffices to show that if a projection $p$ in $B(r_a\caK)$ satisfies
$T_{\oo_a}\lmk pB(r_a\caK)p\rmk\subset p B(r_a\caK)p$, then
$p=r_a$ or $p=0$. See Lemma C.2 \cite{Ogata}.
To do so, we assume that there exists a projection
$p$ in $B(r_a\caK)$ such that
$T_{\oo_a}\lmk pB(r_a\caK)p\rmk\subset p B(r_a\caK)p$ and
$p\neq 0, r_a$, and show a contradiction. 
By $T_{\oo_a}\lmk pB(r_a\caK)p\rmk\subset p B(r_a\caK)p$,
we have
\[
0=(r_a-p)T_{\oo_a}\lmk p\rmk(r_a-p)=\sum_{i=1}^n(r_a-p)w_{ia}p\omega_{ia}^*(r_a-p).
\]
Hence we obtain
\begin{align}\label{eq:pvp}
pv_i^*(r_a-p)=p\omega_{ia}^*(r_a-p)
=0,\quad i=1,\ldots n.
\end{align}
Set $\caK':=\caK_{a-1}+(r_a-p)\caK_a$. Note that $\caK_{a-1}$ and
$(r_a-p)\caK_a$ are orthogonal to each other and  $\caK'$ is a subspace such that
$\caK_{a-1}\subsetneq \caK'\subsetneq \caK_a$.
We claim that $\caK'$ is $\{v_i^*\}_{i=1}^n$-invariant:
for any $\xi\in\caK'$, we have an orthogonal decomposition
$\xi=p_{a-1}\xi+(r_a-p)\xi$.
For the first term, we have $v_i^*p_{a-1}\xi\in v_i^*\caK_{a-1}\subset \caK_{a-1}\subset \caK'$
by the  $v_i^*$-invariance of $\caK_{a-1}$.
For the second term, we have
\[
v_i^*(r_a-p)\xi=p_av_i^*(r_a-p)\xi=
p_{a-1}v_i^*(r_a-p)\xi+(r_a-p)v_i^*(r_a-p)\xi+pv_i^*(r_a-p)\xi,
\]
by the  $v_i^*$-invariance of $\caK_{a}$.
The first term on the right hand side is clearly in $\caK_{a-1}\subset \caK'$
and the second term is in $(r_a-p)\caK_a\subset\caK'$.
The last term is $0$ because of (\ref{eq:pvp}).
Hence we prove the claim.\\
This means that 
$\caK'$ is a proper intermediate subspace which is $\{v_i^*\}_{i=1}^n$-invariant,
between $\caK_a$ and $\caK_{a+1}$. 
This is the contradiction.

To show the last equality, note that the  $\{v_i^*\}$-invariance of $\caK_{a}$ implies
that $p_a v_i =p_a v_ip_a$. From the analogous property of $p_{a-1}$ ,  we also have $\overline{p_{a-1}}v_i\overline{p_{a-1}}=v_i\overline{p_{a-1}}$.
From this we obtain the last equality.
\end{proof}
\begin{notation}\label{nota:oa}We assume [A1],[A3], and [A4].We denote the density matrix of $\rho_\sigma$
by $\tilde \rho_\sigma$.
Set ${n_{0}^{(\sigma)}}=\rank s(\rho_\sigma)$ for $\sigma=R,L$.
Let us consider  $\caK_\sigma$, $v_{i\sigma}$
and $s(\rho_\sigma)\caK_\sigma$, given in Lemma \ref{lem:fs},
Notation \ref{nota:vv}
and Lemma \ref{lem:vucp}
for $\sigma=L,R$.
We set $\widetilde v_{\mu\sigma}:=s(\rho_\sigma)v_{\mu\sigma}s(\rho_\sigma)$,
for $\mu=1,\ldots,n$.
We would like to apply Lemma \ref{lem:hie}
to $\caK=\caK_\sigma$, $v_i=v_{i\sigma}$
and $\caK_0'=s(\rho_\sigma)\caK_\sigma$.
Note that by (\ref{eq:svsvs}), $\caK_0'=s(\rho_\sigma)\caK_\sigma$
is  $\{v_\mu^*\}_{\mu=1}^n$-invariant.
Furthermore, as $T_{(\vv_{\sigma})_{s(\rho_\sigma)}}$
is primitive (Lemma \ref{lem:srs}, {\it 2}), for $l\in\nan$
large enough, we have $\caK_l\lmk(\vv_{\sigma})_{s(\rho_\sigma)}\rmk
=B(s(\rho_\sigma)\caK_\sigma)$.
From this, $\caK_0'=s(\rho_\sigma)\caK_\sigma$
does not have any nonzero proper subspace which is  $\{v_\mu^*\}_{\mu=1}^n$-invariant.
Therefore, we may apply Lemma \ref{lem:hie}.
We then obtain $k_\sigma\in\nan\cup\{0\}$ and 
subspaces $\caK_{a,\sigma}$, $a=0,\ldots, k_\sigma$ satisfying the properties (i)-(iv)
given in Lemma \ref{lem:hie}.
Furthermore, set $\caK_{-1\sigma}:=\{0\}$ and let $p_{a\sigma}$ be the orthogonal projection onto $\caK_{a\sigma}$ on $\caK_\sigma$,
for $a=-1,\ldots,k_\sigma$.
Set $r_{a\sigma}:=p_{a\sigma}-p_{a-1,\sigma}$ for $a=0,\ldots,k_\sigma$ and $\oo_{a,\sigma}:=(\omega_{i,a,\sigma})_i$, $\omega_{i,a,\sigma}=r_{a\sigma} v_{i\sigma} r_{a\sigma}$,
$i=1,\ldots,n$, $a=0,\ldots,k_\sigma$.
Then for each $a=0,\ldots,k_\sigma$, $T_{\oo_{a\sigma}}$ is an irreducible CP map on $B(r_{a\sigma}\caK_\sigma)$ and
\begin{align}\label{eq:rvr}
\omega_{i_1a\sigma}\omega_{i_2a\sigma}\cdots\omega_{i_la\sigma}=p_{a\sigma}v_{i_1\sigma}v_{i_2\sigma}\cdots v_{i_l\sigma}\overline{p_{a-1\sigma}},\quad
l\in\nan,\quad i_1,\ldots,i_l\in\{1,\ldots,n\}.
\end{align}

In particular, either $T_{\oo_{a,\sigma}}=0$ and $r_{a\sigma}$ is rank one,
or $T_{\oo_{a,\sigma}}$ is a nonzero irreducible operator.
For the latter case, 
$r_{T_{\oo_{a,\sigma}}}>0$ and there exists a strictly positive element
$t_{a\sigma}$ in
$B(r_{a\sigma}\caK_{\sigma})$ such that 
$T_{\oo_{a,\sigma}}(t_{a\sigma})=r_{T_{\oo_{a,\sigma}}}\cdot t_{a\sigma}$
and any eigenvector of $T_{\oo_{a,\sigma}}$ corresponding to $r_{T_{\oo_{a,\sigma}}}$ belongs to 
$\cc t_{a\sigma}$.
We define $\uu_{a\sigma}=(u_{\mu a\sigma})_{\mu}$ by
\[
u_{\mu a\sigma}:=r_{T_{\oo_{a,\sigma}}}^{-\frac 12}t_{a\sigma}^{-\frac 12}\omega_{\mu a\sigma}t_{a\sigma}^{\frac 12},\quad
\mu=1,\ldots,n.
\]
By this definition, $T_{\uu_{a\sigma}}$ is a unital irreducible CP map.
Therefore, applying Lemma \ref{lem:irr}, we have 
\begin{enumerate}
\item There exists a $b_{{a\sigma}}\in\nan$ such that
$\sigma(T_{\uu_{a\sigma}})\cap\bbT=\left\{ \exp\lmk \frac{2\pi i}{b_{{a\sigma}}}k\rmk\mid
k=0,\ldots,b_{{a\sigma}}-1\right\}$.
\item
For any $\lambda\in \sigma(T_{\uu_{a\sigma}})\cap\bbT$, $\lambda$ is a nondegenerate
eigenvalue of $T_{\uu_{a\sigma}}$.
\item
There exists a unitary matrix  $U_{a\sigma}\in B(r_{a\sigma}\caK_\sigma)$ such that 
\[
T_{\uu_{a\sigma}}\lmk U_{{a\sigma}}^k\rmk=\exp\lmk \frac{2\pi i}{b_{{a\sigma}}}k\rmk U_{{a\sigma}}^k,\quad k=0,\ldots,b_{{a\sigma}}-1.
\]
\item The unitary matrix $U_{{a\sigma}}$ in {\it 3} has a spectral decomposition
\[
U_{{a\sigma}}=\sum_{k=0}^{b_{{a\sigma}}-1}\exp\lmk \frac{2\pi i}{b_{{a\sigma}}}k\rmk Q_k,
\]
with spectral projections satisfying
\begin{align*}
T_{\uu_{a\sigma}}\lmk Q_k\rmk=Q_{k-1},\quad k\;\mod b_{{a\sigma}}.
\end{align*}
\item
The restriction $T_{\uu_{a\sigma}}^{b_{a\sigma}}\vert_{Q_k  B(r_{a\sigma}\caK_{\sigma})Q_k}$ of $T_{\uu_{a\sigma}}^{b_{a\sigma}}$ on 
$Q_k B(r_{a\sigma}\caK_{\sigma}) Q_k$ defines a primitive unital CP map on
$Q_k B(r_{a\sigma}\caK_{\sigma}) Q_k$.
\item
There exists a faithful $T_{\uu_{a\sigma}}$-invariant state
$\varphi_{a\sigma}$.
\end{enumerate}
\end{notation}

\begin{lem}\label{lem:rta}
For any $a\ge 1$, we have $r_{T_{\oo_{a,\sigma}}}<1$.
\end{lem}
\begin{proof}
Using (\ref{eq:rvr}), we have for $a\ge 1$,
\begin{align*}
&r_{T_{\oo_{a,\sigma}}}
=\lim_{N\to\infty}\lV T_{\oo_{a\sigma}}^N\lmk r_{a\sigma}\rmk\rV^{\frac 1N}
=\lim_{N\to\infty}\lV r_{a\sigma}T_{\vv_\sigma}^N
\lmk r_{a\sigma}\rmk r_{a\sigma}\rV^{\frac 1N}
\le\limsup_{N\to\infty} \lmk
\lV\rho_{\sigma}\lmk r_{a\sigma}\rmk r_{a\sigma}\rV
+\lV T_{\vv_\sigma}^N\lmk\unit-P_{\{1\}}^{T_{\vv_\sigma}}\rmk\rV
\rmk^{\frac1N}\\
&=\limsup_{N\to\infty} \lmk
\lV T_{\vv_\sigma}^N\lmk\unit-P_{\{1\}}^{T_{\vv_\sigma}}\rmk\rV
\rmk^{\frac1N}<1.
\end{align*}
\end{proof}
As a result of [A5], we eliminate the possibility $T_{\oo_{a,\sigma}}=0$.
\begin{lem}\label{lem:oonz}
Assume [A1],[A3],[A4], and [A5], and consider the setting in Notation \ref{nota:oa}.
Then for any $a=1,\ldots,k_\sigma$,
we have $T_{\oo_{a,\sigma}}\neq 0$.
\end{lem}
\begin{proof}
If the claim does not hold,
there exists an $a_0=1,\ldots,k_\sigma$ such that
$T_{\oo_{a_0,\sigma}}=0$. As stated in Notation \ref{nota:oa}, in this case,
$r_{a_0\sigma}$ is rank one.
Let $\xi$ be a unit vector in the one dimensional space
$r_{a_0\sigma}\caK_{\sigma}$.
Then $\psi=\braket{\xi}{\pi_\sigma\lmk\cdot \rmk\xi}$
is an element of $\caS_{\sigma}(H)$ by Lemma \ref{lem:fs}.

Because of the choice of $a_0$, we know that 
$\overline{p_{a_0-1,\sigma}}v_{\mu^{(l)},\sigma}^*p_{a_0,\sigma}=0$, using (\ref{eq:rvr}).
Therefore, we have $v_{\mu^{(l)},\sigma}^*\caK_{a_0,\sigma}\subset 
\caK_{a_0-1,\sigma}$, for any $\mu^{(l)}\in\{1,\ldots,n\}^{\times l}$,
and $l\in\nan$.
In particular, $\xi$ and  $v_{\mu^{(l)},\sigma}^*\xi$ are orthogonal,
for any $\mu^{(l)}\in \{1,\ldots,n\}^{\times l}$ and $l\in\nan$.
From this, we obtain
\begin{align*}
\psi\circ\tau_l^{(\sigma)}=
\sum_{\mu^{(l)}}\braket{\xi}{\widehat{v_{\mu^{(l)},\sigma}}
\pi_\sigma\lmk\cdot\rmk {\widehat{v_{\mu^{(l)},\sigma}}}^*\xi}
=\sum_{\mu^{(l)}}\braket{\widehat{v_{\mu^{(l)},\sigma}}^*\xi}
{\pi_\sigma\lmk\cdot\rmk {\widehat{v_{\mu^{(l)},\sigma}}}^*\xi}.
\end{align*}

Let $p$ be the orthogonal projection onto $\cc\xi$.
As $\xi$ and  $v_{\mu^{(l)},\sigma}^*\xi$ are orthogonal, we have
$pv_{\mu^{(l)}}^*\xi=0$.
Since $\varphi_{\sigma}$ is pure, we have $\pi_\sigma(\caA_\sigma)^{''}=B(\caH_\sigma)$.
Therefore, by Kaplansky's density Theorem, there exists
a net $\{x_\alpha\}$ in the unit ball of $\caA_{\sigma+}$
such that 
$
\sigma w-\lim_{\alpha}\pi_{\sigma}\lmk x_\alpha\rmk
=p
$.
For this net, we have
$
\lim_{\alpha}\psi\circ\tau_l^{(\sigma)}\lmk x_\alpha\rmk
=0$, and
$\lim_{\alpha}\psi(x_\alpha)=1$,
for any $l\in\nan$.
Note also that
$-1\le 2 x_\alpha-1\le 1$, hence
$\lV 2x_\alpha-1\rV\le 1$.

This implies
\[
\lV\psi-\psi\circ\tau_l^{(\sigma)}\rV\ge
\lim_\alpha\lv \lmk\psi-\psi\circ\tau_l^{(\sigma)}\rmk\lmk 2 x_\alpha-1\rmk\rv=2,\quad
l\in\nan.
\]
This contradicts [A5].
\end{proof}
\section{Primitivity of $T_{\uu_{a\sigma}}$}\label{sec:yon}
In this section, we show that the number $b_{a\sigma}$, given in
\ref{nota:oa} is $1$. In other words, $T_{\uu_{a\sigma}}$ is primitive.
More precisely, we show the following Lemma.
\begin{lem}\label{lem:primu}
We assume [A1],[A3],[A4], and [A5] and use Notation \ref{nota:vv} and Notation \ref{nota:oa}.
Then for any $a=0,\ldots, k_\sigma$,
there exist a unitary $V_{a\sigma}:\cc^{n_0^{(\sigma)}}\to r_{a\sigma}\caK_{\sigma}$
and $c_{a\sigma}\in \bbT$ such that
\[
u_{\mu a\sigma}=c_{a\sigma}V_{a\sigma}\widetilde{v_{\mu\sigma}}V_{a\sigma}^*,\quad
\mu=1,\ldots,n.
\]
In particular, $T_{\uu_{a\sigma}}$ is primitive and $\rank r_{a\sigma}={n_{0}^{(\sigma)}}$.
\end{lem}
This is because of the uniqueness of the bulk ground state [A3],
which implies the following Lemma.
\begin{lem}\label{lem:ug}
We assume [A1], [A3], [A4] [A5] ,and use Notation \ref{nota:vv},
 Notation \ref{nota:oa}.
Then $(B(r_{a\sigma}\caK_\sigma), \uu_{a\sigma},\varphi_{a\sigma})$
$\sigma$-generates
$\omega_{\infty}$.
\end{lem}
\begin{proof}
We denote $\vv=\vv_{\sigma}$, $\oo=\oo_{a\sigma}$, and
$\uu=\uu_{a\sigma}$,
in this proof, for simplicity.
Define 
$\bbE^{(a\sigma)}:\Mat_n\otimes B(r_{a\sigma}\caK_\sigma)\to 
B(r_{a\sigma}\caK_\sigma)$ by
\begin{align}\label{eq:as}
\bbE^{(a\sigma)}\lmk e_{\mu\nu}^{(n)}\otimes X\rmk
:=u_{\mu } X u_{\nu }^*,\quad X\in B(r_{a\sigma}\caK_\sigma).
\end{align}
Then $(B(r_{a\sigma}\caK_\sigma), \bbE^{(a\sigma)},\varphi_{a\sigma})$
is a standard triple, and $\sigma$-generates a state
$\widetilde \omega_{a,\sigma}$ on $\caA_{\bbZ}$.
We claim that $\widetilde \omega_{a,\sigma}=\omega_{\infty}$.
To see this, it suffices to show that 
$\widetilde \omega_{a,\sigma}(\tau_x(h))=0$ for all $x\in\bbZ$ because of [A3]
and Lemma \ref{lem:vr}, {\it 1}.

By the definition, we have for any $x\in\bbZ$,
\begin{align*}
&0\le\widetilde \omega_{a,\sigma}\lmk
\tau_x(h)
\rmk
=\sum_{\mu^{(m)}\nu^{(m)}}
\braket{\ws{m}}{ h\wsn{m}}
\varphi_{a\sigma}\lmk \widehat u_{\mu^{(m,\sigma)}} 
\lmk \widehat u_{\nu^{(m,\sigma)}}\rmk^*\rmk\\
&=r_{T_{\oo}}^{- m}\sum_{\mu^{(m)}\nu^{(m)}}\braket{\ws{m}}{ h\wsn{m}}
\varphi_{a\sigma}
\lmk
t_{a\sigma}^{-\frac 12}
\widehat{v_{\mu^{(m,\sigma)}} }t_{a\sigma}\lmk \widehat {v_{\nu^{(m,\sigma)}}}\rmk^*
t_{a\sigma}^{-\frac 12}
\rmk\\
&=r_{T_{\oo}}^{- m}\sum_{\lambda^{(m)}}
\varphi_{a\sigma}
\lmk
t_{a\sigma}^{-\frac 12}
\lmk
\sum_{\mu^{(m)}}\braket{\ws{m}}{ h^{\frac12}\widehat{\psi_{\lambda^{(m)}}}}
\widehat{v_{\mu^{(m,\sigma)}}}\rmk t_{a\sigma}
\lmk
\sum_{\nu^{(m)}}\braket{\wsn{m}}{ h^{\frac12}\widehat{\psi_{\lambda^{(m)}}}}
\widehat{ v_{\nu^{(m,\sigma)}}} \rmk^*t_{a\sigma}^{-\frac 12}\rmk\\
&\le
\lV t_{a\sigma}\rV
r_{T_{\oo}}^{- m}\sum_{\lambda^{(m)}}
\varphi_{a\sigma}
\lmk
t_{a\sigma}^{-\frac 12}
\lmk
\sum_{\mu^{(m)}}\braket{\ws{m}}{ h^{\frac12}\widehat{\psi_{\lambda^{(m)}}}}
\widehat{v_{\mu^{(m,\sigma)}}}\rmk
\lmk
\sum_{\nu^{(m)}}\braket{\wsn{m}}{ h^{\frac12}\widehat{\psi_{\lambda^{(m)}}}}
\widehat{v_{\nu^{(m,\sigma)}}} \rmk^*t_{a\sigma}^{-\frac 12}\rmk\\
&=\lV t_{a\sigma}\rV
r_{T_{\oo}}^{- m}
\sum_{\mu^{(m)}}\sum_{\nu^{(m)}}\braket{\ws{m}}{ h\widehat{\psi_{\nu^{(m)}}}}
\varphi_{a\sigma}
\lmk
t_{a\sigma}^{-\frac 12}
\lmk
\widehat{v_{\mu^{(m,\sigma)}}}\rmk
\lmk
\widehat{v_{\nu^{(m,\sigma)} }}\rmk^*t_{a\sigma}^{-\frac 12}\rmk\\
&=
\lV t_{a\sigma}\rV
r_{T_{\oo}}^{- m}
\varphi_{a\sigma}
\lmk
t_{a\sigma}^{-\frac 12}P_{\caK_\sigma}
\pi_{\sigma}
\lmk\tau_{y}\lmk
h\rmk
\rmk
P_{\caK_\sigma}
t_{a\sigma}^{-\frac 12}\rmk=0,
\end{align*}
where $y=0$ if $\sigma=R$ and $y=-m$ if $\sigma=L$.
\end{proof}
\begin{rem}\label{rem:iso}
Let
${\mathfrak B}_{a\sigma}$ be the minimal $C^*$-subalgebra of 
$B(r_{a\sigma}\caK_\sigma)$ which contains $\unit$ and is
$\bbE_A$-invariant for any $A\in\Mat_n$.
Then for $\bbE^{(a\sigma)}$ given by (\ref{eq:as}),
 $({\mathfrak B}_{a\sigma},\bbE^{(a\sigma)}\vert_{\Mat_n\otimes {\mathfrak B}_{a\sigma}},\varphi_{a\sigma}\vert_{{\mathfrak B}_{a\sigma}})$
is a minimal standard triple $\sigma$-generating $\omega_\infty$, from Lemma \ref{lem:ug}.
The eigenspace of $1$ for $T_{\uu_{a\sigma}}=\bbE^{(a\sigma)}_{\unit}$
is $\bbC\unit$.
Recall that $(B(s(\rho_\sigma)\caK_\sigma)),\bbE^{(\sigma)},\rho_\sigma\vert_{B(s(\rho_\sigma)\caK_\sigma)})$
is a minimal standard triple $\sigma$-generating $\omega_\infty$.
The eigenspace of $1$ for $T_{(\vv_\sigma)_{s(\rho_\sigma)}}=\bbE^{(\sigma)}_\unit$
is $\bbC\unit$ (Lemma \ref{lem:srs}).
For each $N\in \nan$, let $D_N$ be the density matrix of $\omega_\infty\vert_{\caA_{[0,N-1]}}$.
We have $\sup_N\rank D_N<\infty$ because of [A1] and Lemma \ref{lem:vr}.
By Theorem \ref{thm:iso}, this implies the existence of
a $*$-isomorphism $\Theta_{a\sigma}:B\lmk s(\rho_\sigma)\caK_\sigma\rmk\to {\mathfrak B}_{a\sigma}$ satisfying
\[
\bbE^{(a\sigma)}\circ\lmk id_{\Mat_n}\otimes\Theta_{a\sigma}\rmk
=\Theta_{a\sigma}\circ\bbE^{(\sigma)}.
\]
This condition can be written
\begin{align*}
\Theta_{a\sigma}\lmk
\lmk v_{\mu \sigma}\rmk_{s(\rho_\sigma)}X
\lmk\lmk v_{\nu \sigma}\rmk_{s(\rho_\sigma)}\rmk^*
\rmk
=u_{\mu a\sigma}\Theta_{a\sigma}\lmk
X\rmk\lmk u_{\nu a\sigma}\rmk^*,\quad
\mu,\nu=1,\ldots,n,\quad X\in B\lmk s(\rho_\sigma)\caK_\sigma\rmk.
\end{align*}
\end{rem}
We apply the following Lemma to this situation.
\begin{lem}\label{lem:irpri}
Let $n,d_1,d_2\in\nan$. Let $\vv^{(i)}=(v_{\mu}^{(i)})_{\mu=1}^n\in
\Mat_{d_i}^{\times n}$, for $i=1,2$.
Assume that $T_{\vv^{(1)}}$ is a primitive unital CP map and
that $T_{\vv^{(2)}}$ is an irreducible unital CP map.
Furthermore, assume that 
there exists an injective unital $*$-homomorphism $\Theta:\Mat_{d_1}\to\Mat_{d_2}$
such that
\begin{align}\label{eq:th}
\Theta\lmk
v_{\mu}^{(1)}X\lmk v_{\nu}^{(1)}\rmk^*
\rmk
=v_{\mu}^{(2)}\Theta\lmk
X\rmk\lmk v_{\nu}^{(2)}\rmk^*,\quad
\mu,\nu=1,\ldots,n,\quad X\in \Mat_{d_1}.
\end{align}
Then $d_1=d_2$, and $T_{\vv^{(2)}}$ is primitive.
There exists a unitary $W:\cc^{d_2}\to \cc^{d_1}$ and 
a complex number $c\in\bbT$ such that
\[
Wv^{(2)}_{\mu}W^*=cv_{\mu}^{(1)},\quad
\mu=1,\ldots,n.
\]
\end{lem}
\begin{proof}
Applying Lemma \ref{lem:irr} to $T_{\vv^{(2)}}$,
we obtain $b, U, P_k$ satisfying {\it 1.-6.} of Lemma \ref{lem:irr}.
Note from {\it 4.} of Lemma \ref{lem:irr}, we have
\begin{align}\label{eq:qv}
v_{\mu}^{(2)}P_k=P_{k-1}v_{\mu}^{(2)},\quad
k=0,\ldots,b-1\;\;\mod b,\;\; \mu=1,\ldots,n.
\end{align}

First we claim that
\begin{align}\label{eq:tdcom}
\Theta(\Mat_{d_1})\subset \bigoplus_{k=0}^{b-1}P_k\Mat_{d_2}P_k.
\end{align}
To see this, recall that $T_{\vv^{(1)}}$ is primitive.
Therefore, there exists an $l_0\in\nan$ such that
$\caK_{lb}(\vv^{(1)})=\Mat_{d_1}$, for all $l\ge l_0$.
From this and (\ref{eq:th}), (\ref{eq:qv}),  we have
\begin{align}\label{eq:tmd1}
&\Theta\lmk\Mat_{d_1}\rmk
=\Theta\lmk\caK_{lb}\lmk \vv^{(1)}\rmk\lmk 
\caK_{lb}\lmk \vv^{(1)}\rmk\rmk^*\rmk
=\lmk
\caK_{lb}\lmk
\vv^{(2)}\rmk\rmk\Theta(1)\lmk \caK_{lb}\lmk \vv^{(2))}\rmk\rmk^*\notag\\
&=\sum_{k=0}^{b-1}\lmk
\caK_{lb}\lmk \vv^{(2)}\rmk\rmk P_{k}\lmk \caK_{lb}\lmk \vv^{(2)}\rmk\rmk^*
=\sum_{k=0}^{b-1}P_{k}\lmk
\caK_{lb}\lmk \vv^{(2)}\rmk\rmk 
\lmk \caK_{lb}\lmk \vv^{(2)}\rmk\rmk^*P_k,
\end{align}
for any $l\ge l_0$,
proving the claim.

Next we claim that
for each $k=0,\ldots,b-1$, 
there exists a unitary
$V_k:\cc^{d_1}\to P_k\cc^{d_2}$
such that
\[
\Theta(X)P_k= V_k X V_k^*,\quad
X\in\Mat_{d_1}.
\]
To see this, first note that $\Theta_k:\Mat_{d_1}\to P_k\Mat_{d_2}P_k$
given by $\Theta_k(X)=\Theta(X)P_k$, $k=0,\ldots,b$
is a $*$-homomorphism
because of the first observation.\\
The map $T_{\vv^{(2)}}^b\vert_{P_k\Mat_{d_2}P_k}$ is primitive by {\it 5.}
of Lemma \ref{lem:irr}.
This fact, combined with (\ref{eq:qv}),  implies the existence of an $l_1\in\nan$
such that $P_k \caK_{lb}(\vv^{(2)})=P_k\Mat_{d_2}P_k$
for all $l\ge l_1$.
Therefore, using (\ref{eq:tmd1}) and (\ref{eq:tdcom}), we have
for $l\ge \max\{l_1,l_0\}$,
\begin{align*}
P_k\Mat_{d_2}P_k=\lmk P_k\Mat_{d_2}P_k\rmk\lmk P_k\Mat_{d_2}P_k\rmk^*=
P_k\lmk \caK_{lb}(\vv^{(2)})\rmk 
\lmk \caK_{lb}(\vv^{(2)})\rmk^*P_k
=P_k
\Theta\lmk \Mat_{d_1}\rmk P_k
=\Theta_k\lmk \Mat_{d_1}\rmk.
\end{align*}
Therefore, $\Theta_k$ is a  $*$-homomorphism from
$\Mat_{d_1}$ onto $P_k\Mat_{d_2}P_k$.
As $\Mat_{d_1}$ is simple, it is also injective. Hence,
$\Theta_k$ is a $*$-isomorphism between
$\Mat_{d_1}$ and $P_k\Mat_{d_2}P_k$.
By Wigner's Theorem, this implies the existence of $V_k$ as we claimed.

Define a linear map $W:\cc^{d_2}\to\cc^{d_1}\otimes \cc^b$ by 
\[
W\xi:=\sum_{k=0}^{b-1}V_k^*P_k\xi\otimes \chi_{k+1}^{(b)},\quad
\xi\in \cc^{d_2}.
\]
It is easy to check that $W$ is unitary and 
\begin{align}\label{eq:in}
\lmk Wv_{\mu}^{(2)}W^*\rmk
 \lmk X\otimes \unit \rmk \lmk W{ v_{\nu}^{(2)}}W^*\rmk^*
= v_{\mu}^{(1)}X{v_{\nu}^{(1)}}^*\otimes\unit,\quad
X\in \Mat_{d_1},\quad \mu,\nu=1,\ldots,n,
\end{align}
from (\ref{eq:th}).
Substituting $X=1$ and $\mu=\nu$ in (\ref{eq:in}), we obtain
\[
\lmk Wv_{\mu}^{(2)}W^*\rmk
\lmk W{ v_{\mu}^{(2)}}W^*\rmk^*
= v_{\mu}^{(1)}{v_{\mu}^{(1)}}^*\otimes\unit,\quad \mu=1,\ldots,n.
\]
By the polar decomposition,
This means that there exist unitary operators $\caW_\mu$, $\mu=1,\ldots,n$
in $\Mat_{d_1}\otimes \Mat_b$ such that
\begin{align}\label{eq:wvv}
Wv_{\mu}^{(2)}W^*=
\lmk
v_{\mu}^{(1)}\otimes\unit
\rmk
\caW_{\mu}.
\end{align}
We claim for each  $\mu=1,\ldots,n$ that 
$\caW_{\mu}$ has a decomposition
\begin{align}\label{eq:wdec}
\caW_{\mu}=w_{\mu}^{(1)}+w_{\mu}^{(2)},
\end{align}
where
$w_{\mu}^{(1)}$ is a unitary in 
$\lmk s_r(v_{\mu}^{(1)})\Mat_{d_1}s_r(v_{\mu}^{(1)})\rmk\otimes \Mat_b$
and 
$w_{\mu}^{(2)}$ is a unitary in 
$\lmk \overline{s_r(v_{\mu}^{(1)})}\Mat_{d_1}\overline{s_r(v_{\mu}^{(1)})}
\rmk\otimes \Mat_b$.
We substitute $X= \overline{s_r(v_{\mu}^{(1)})}$ in (\ref{eq:in})
with $\mu=\nu$,
and obtain
\begin{align*}
&0=v_{\mu}^{(1)}\overline{s_r(v_{\mu}^{(1)})}{v_{\mu}^{(1)}}^*\otimes\unit
=\lmk Wv_{\mu}^{(2)}W^*\rmk\lmk\overline{s_r(v_{\mu}^{(1)})}\otimes \unit
\rmk
\lmk W{ v_{\mu}^{(2)}}W^*\rmk^*
=
\lmk v_{\mu}^{(1)}\otimes \unit \rmk \caW_\mu
\lmk\overline{s_r(v_{\mu}^{(1)})}\otimes \unit
\rmk\caW_{\mu}^*
\lmk v_{\mu}^{(1)}\otimes \unit \rmk^*.
\end{align*}
This means $\lmk s_r(v_{\mu}^{(1)})\otimes \unit \rmk \caW_\mu
\lmk\overline{s_r(v_{\mu}^{(1)})}\otimes \unit
\rmk=0$ which implies the claim.

Assume that $v_{\mu}^{(1)}\neq 0$.
By (\ref{eq:in}), (\ref{eq:wvv}), and (\ref{eq:wdec}), we have
\begin{align*}
\lmk
v_{\mu}^{(1)}\otimes\unit
\rmk
w_{\mu}^{(1)} \lmk X\otimes \unit\rmk {w_{\mu}^{(1)}}^*
\lmk
v_{\mu}^{(1)}\otimes\unit
\rmk^*
= v_{\mu}^{(1)}X{v_{\mu}^{(1)}}^*\otimes\unit,\quad
X\in  \Mat_{d_1}.
\end{align*}
From this we get
\[
w_{\mu}^{(1)} \lmk X\otimes \unit\rmk {w_{\mu}^{(1)}}^*
= X\otimes\unit,\quad
X\in  s_r(v_{\mu}^{(1)})\Mat_{d_1} s_r(v_{\mu}^{(1)}).
\]
This means $w_{\mu}^{(1)}\in  s_r(v_{\mu}^{(1)})\otimes\Mat_b$.
Therefore, there exists a unitary $\tilde w_{\mu}$
such that
\[
w_{\mu}^{(1)} = s_r(v_{\mu}^{(1)})\otimes \tilde w_{\mu}.
\]
We have
\begin{align}\label{eq:vtw}
v_{\mu}^{(1)}\otimes\tilde w_{\mu}
=\lmk v_{\mu}^{(1)}\otimes\unit\rmk
\lmk s_r(v_{\mu}^{(1)})\otimes \tilde w_{\mu}\rmk
=\lmk v_{\mu}^{(1)}\otimes\unit\rmk\caW_{\mu}
=Wv_{\mu}^{(2)}W^*,
\end{align}
for all $\mu=1,\ldots,n$ with $v_{\mu}^{(1)}\neq 0$.

The unitary matrices $\tilde w_{\mu}$ can be taken independent of $\mu$,
$v_{\mu}^{(1)}\neq 0$.
To see this, substitute (\ref{eq:vtw}) to (\ref{eq:in}) and obtain
\begin{align*}
  v_{\mu}^{(1)}X{v_{\nu}^{(1)}}^*
\otimes \tilde w_{\mu}{\tilde w_{\nu}}^*
= v_{\mu}^{(1)}X{v_{\nu}^{(1)}}^*\otimes\unit,\quad
X\in \Mat_{d_1},\quad \mu,\nu=1,\ldots,n,\quad
v_{\mu}^{(1)}\neq0,\;\; v_{\nu}^{(1)}\neq0.
\end{align*}
If $v_{\mu}^{(1)}$ and $v_{\nu}^{(1)}$ are not zero, this equality means
that $\tilde w_{\mu}{\tilde w_{\nu}}^*=1$, i.e., 
$\tilde w_{\mu}={\tilde w_{\nu}}$.
We deote this common $\tilde w_{\mu}$ by $w$.
(Note that there exists at least one $v_{\mu}^{(1)}\neq 0$ because
$T_{\vv^{(1)}}$ is unital.)
Hence we obtain 
\begin{align}\label{eq:univ}
Wv_{\mu}^{(2)}W^*
=
v_{\mu}^{(1)}\otimes w,\quad \mu=1,\ldots,n.
\end{align}
Note that  this also holds for $\mu$ with  $v_{\mu}^{(1)}= 0$
because of (\ref{eq:wvv}).

Now we prove $b=1$, i.e., $d_1=d_2$.
Assume that $b\neq1$. Then because $w$ is unitary,
there exists an $x\in\Mat_b$ such that $w x w^*=x$ and
$x\notin \cc \unit_b$.
Set $X:=W^*\lmk \unit\otimes x\rmk W\in \Mat_{d_2}$.
We have $X\notin \cc\unit_{d_2}$ and
\[
T_{\vv^{(2)}}(X)
=\sum_{\mu=1}^n W^*Wv_{\mu}^{(2)}W^*
\lmk \unit\otimes x\rmk W {v_{\mu}^{(2)}}^* W^* W
=\sum_{\mu=1}^n W^*
\lmk v_\mu^{(1)}{v_\mu^{(1)}}^*
\otimes wxw^*\rmk 
 W
=W^*\lmk
\unit \otimes x
\rmk W=X.
\]
By {\it 2.} of Lemma \ref{lem:irr},
$1$ is a nondegenerate eigenvalue of $T_{\vv^{(2)}}$.
This is a contradiction.
Therefore, we conclude $b=1$ and 
$d_1=d_2$.
In this case, (\ref{eq:univ}) implies the existence of 
unitary $W:\cc^{d_2}\to \cc^{d_1}$  and $c\in\bbT$ satisfying 
\begin{align*}
Wv_{\mu}^{(2)}W^*
=
c
v_{\mu}^{(1)}\quad \mu=1,\ldots,n.
\end{align*}
Clearly, this implies the primitivity of $T_{\vv^{(2)}}$.
\end{proof}
\begin{proofof}[Lemma \ref{lem:primu}]
Recall Remark \ref{rem:iso}.
Applying Lemma \ref{lem:irpri} to $\Theta_{a\sigma}$,
we 
obtain a unitary $W_{a\sigma}:r_{a\sigma}\caK_\sigma\to s(\rho_\sigma)\caK_{\sigma}$,
and $c_{a\sigma}\in\bbT$
such that
$
W_{a\sigma}u_{\mu a\sigma}W_{a\sigma}^*
=c_{a\sigma}\widetilde{v_{\mu\sigma}}$ for
$\mu=1,\ldots,n$.
Set $V_{a\sigma}:=W_{a\sigma}^*$ and under the identification
$\cc^{n_0^{(\sigma)}}\simeq s(\rho_\sigma)\caK_\sigma$,
we complete the proof.
\end{proofof}
\section{The bijectivity  of $\left.\Gamma_{l,\vv}^{(\sigma)}\right\vert_{B(\caK_\sigma)s(\rho_\sigma)}$ }\label{sec:go}
In this section, we prove the following Lemma.
\begin{lem}\label{lem:bijecs}
Assume [A1], [A3], [A4], and [A5]. Let $\sigma=L,R$.
Let $\vv_\sigma$ be the $n$-tuple of
elements in  $B(\caK_\sigma)$ given in Notation \ref{nota:vv} and $\rho_\sigma$
the state given in Lemma \ref{lem:vucp}.
Then there exists an $l_\sigma'\in\nan$ such that
\[
\left.\Gamma_{l,\vv_\sigma}^{(\sigma)}\right\vert_{B(\caK_\sigma)s(\rho_\sigma)}: B(\caK_\sigma)s(\rho_\sigma)\to \Gamma_{l,\vv_\sigma}^{(\sigma)}\lmk B(\caK_\sigma)\rmk=
\tau_{y_{\sigma}}\lmk
s\lmk \left. \omega_\sigma\right\vert _{\caA_{\sigma,l}}\rmk
\rmk\bigotimes_{i=0}^{l-1}\cc^n
\]
is a bijection for any $l\ge l_\sigma'$.
Here, $y_R=0$ and $y_L=l$.
\end{lem}
We start from the following simple observation.
\begin{lem}\label{lem:Xxp}
Assume [A1] and [A4].
 Let $\vv_\sigma$ be the $n$-tuple of
elements in $B(\caK_\sigma)$ given in Notation \ref{nota:vv}.
For $l\in\nan$,
a unit vector $\xi\in\caK_\sigma$ and a projection $p$ in $B(\caK_\sigma)$,
define $X_{l,\xi,p}^{(\sigma)}\in \caA_{[0,l-1]}$ by
\[
X_{l,\xi,p}^{(\sigma)}:=
\sum_{\mu^{(l)},\nu^{(l)}}\braket{\xi}{{\widehat{v_{\mu^{(l,\sigma)}}}}p {\widehat{v_{\nu^{(l,\sigma)}}}}^*\xi}
\ket{\wsn{l}}\bra{\ws{l}}.
\]
Let $\omega_\xi$ be a state given by $\omega_\xi=\braket{\xi}{\pi_\sigma\lmk\cdot\rmk \xi}$.
Then $X_{l,\xi,p}^{(\sigma)}$ is positive and 
\begin{align*}
s\lmk X_{l,\xi,p}^{(\sigma)}\rmk\le
\tau_{y_\sigma}\lmk
s\lmk\omega_\xi\vert_{\caA_{\sigma,l}}\rmk\rmk,\quad\quad
\omega_{\xi}\circ \tau_{-y_\sigma}\lmk A\rmk
=\Tr\lmk X_{l,\xi,\unit}^{(\sigma)}A
\rmk,\quad A\in \caA_{[0,l-1]},
\end{align*}
where $y_R=0$ and $y_L=l$.
Furthermore, for a unit vector $\eta\in\caK_\sigma$,
\[
X_{l,\xi,\ket{\eta}\bra{\eta}}^{(\sigma)}
=\ket{\Gamma_{l\vv_\sigma}^{(\sigma)}\lmk\ket{\xi}\bra{\eta}\rmk}\bra{\Gamma_{l\vv_\sigma}^{(\sigma)}\lmk\ket{\xi}\bra{\eta}\rmk}.
\]
\end{lem}
\begin{proof}
For any $\zeta\in\bigotimes_{i=0}^{l-1}\cc^n$,
we have
\[
\braket{\zeta}{X_{l,\xi,p}^{(\sigma)}\zeta}
=\lV p\lmk\sum_{\nu^{(l)}}\braket{\wsn{l}}{\zeta}
{\widehat{v_{\nu^{(l,\sigma)},\sigma}}}\rmk^*\xi\rV^2\ge 0,
\]
and
\[
\omega_{\xi}\lmk \tau_{-y_\sigma}\lmk \ket{\zeta}\bra{\zeta}\rmk\rmk
=\lV \lmk \sum_{\nu^{(l)}}\braket{\wsn{l}}{\zeta}
{\widehat{v_{\nu^{(l,\sigma)},\sigma}}}\rmk^*\xi\rV^2,
\]
where $y_R=0$ and $y_L=l$.
The claim of the Lemma can be checked from these equations.
\end{proof}
\begin{lem}\label{lem:Xxp2}
Assume [A1] and [A4].
 Let $\omega_\sigma$ be the state in [A4] and $\vv_\sigma$ the $n$-tuple of
elements in $B(\caK_\sigma)$ given in Notation \ref{nota:vv}.
Then for any $l\in\nan$,
$\tau_{y_\sigma}\lmk s(\omega_{\sigma}\vert_{\caA_{\sigma,l}})\rmk$ is equal to
the orthogonal projection onto $\Gamma_{l,\vv_\sigma}^{(\sigma)}\lmk B(\caK_\sigma)\rmk$,
where $y_R=0$ and $y_L=l$.
\end{lem}
\begin{proof}
Let $\rho_{\omega_\sigma}$ be the density matrix given by
Lemma \ref{lem:fs}.
As it is strictly positive, it can be decomposed as $\rho_{\omega_\sigma}=\sum_i \lambda_i\ket{\eta_i}\bra{\eta_i}$
with numbers $\lambda_i>0$ and 
CONS $\{\eta_i\}_i$ of $\caK_\sigma$.
By Lemma \ref{lem:Xxp}, we get
\begin{align*}
&\omega_\sigma\circ \tau_{-y_\sigma}\lmk A\rmk
=\sum_i\lambda_i \omega_{\eta_i}\circ \tau_{-y_\sigma}\lmk A\rmk
=\sum_{ij}\lambda_i\Tr\lmk X_{l,\eta_i,\ket{\eta_j}\bra{\eta_j}}^{(\sigma)}A\rmk\\
&=\sum_{ij}\lambda_i
\braket{\Gamma_{l\vv_\sigma}^{(\sigma)}\lmk\ket{\eta_i}\bra{\eta_j}\rmk}{A \Gamma_{l\vv_\sigma}^{(\sigma)}\lmk\ket{\eta_i}\bra{\eta_j}\rmk},\quad
A\in \caA_{[0,l-1]}, \;l\in\nan.
\end{align*}
As $\{\Gamma_{l\vv_\sigma}^{(\sigma)}\lmk\ket{\eta_i}\bra{\eta_j}\rmk\}_{ij}$ spans  $\Gamma_{l,\vv_\sigma}^{(\sigma)}\lmk B(\caK_\sigma)\rmk$,
this proves the Lemma.
\end{proof}

\begin{lem}
Assume [A1], [A3], [A4], and [A5]. Let $\vv_\sigma$ be the $n$-tuple of
elements in $B(\caK_\sigma)$ given in Notation \ref{nota:vv} and $\rho_\sigma$
the state given in Lemma \ref{lem:vucp}.
Then there exists
an $\tilde l_\sigma\in\nan$ such that
\[
\Gamma_{l\vv_\sigma}^{(\sigma)}\lmk B(\caK_\sigma)s(\rho_\sigma)\rmk
=\Gamma_{l\vv_\sigma}^{(\sigma)}\lmk B(\caK_\sigma)\rmk,\quad
l\ge \tilde l_\sigma. 
\]
\end{lem}
\begin{proof}
Set 
\[
C_\sigma:=\inf\left\{ \sigma\lmk D_{\omega_\sigma\vert_{\caA_{\sigma,l}}}\rmk
\setminus \{0\}\mid l\in\nan\right\}>0,
\]
and 
\[
\tilde l_\sigma
:=\min\left\{
l\in\nan\mid
\sup_{l':l\le l'}\left\{
\lV T_{\vv_\sigma}^{l'}\lmk\unit-P_{\{1\}}^{T_{\vv_\sigma}}\rmk\rV\right\}
<\frac 12 C_\sigma
\right\}\in\nan.
\]
(Recall [A4] and Lemma \ref{lem:vucp} to see $C_\sigma>0$ and $\tilde l_\sigma\in\nan$.)
We assume that there exists an $l\ge\tilde l_\sigma$
such that $\Gamma_{l\vv_\sigma}^{(\sigma)}\lmk B(\caK_\sigma)s(\rho_\sigma)\rmk
\neq\Gamma_{l\vv_\sigma}^{(\sigma)}\lmk B(\caK_\sigma)\rmk$
and show a contradiction.

If $\Gamma_{l\vv_\sigma}^{(\sigma)}\lmk B(\caK_\sigma)s(\rho_\sigma)\rmk
\neq\Gamma_{l\vv_\sigma}^{(\sigma)}\lmk B(\caK_\sigma)\rmk$, then there exists an $X\in B(\caK_\sigma)$ such that 
$\Gamma_{l\vv_\sigma}^{(\sigma)}(X)$ and 
$\Gamma_{l\vv_\sigma}^{(\sigma)}\lmk B(\caK_\sigma s(\rho_\sigma))\rmk$
are orthogonal and $\lV \Gamma_{l\vv_\sigma}^{(\sigma)}(X)\rV=1$.

First we show 
\begin{align}\label{eq:ranX}
\ket{\Gamma_{l\vv_\sigma}^{(\sigma)}(X)}\bra{\Gamma_{l\vv_\sigma}^{(\sigma)}(X)}
\le \tau_{y_\sigma}\lmk s\lmk D_{\omega_\sigma\vert_{\caA_{\sigma,l}}}\rmk\rmk,
\end{align}
where $y_R=0$ and $y_L=l$.
We may represent $X=\sum_{i=1}^{m_\sigma}c_i\ket{\xi_i}\bra{\eta_i}$, with 
$\lV \xi_i\rV=\lV \eta_i\rV=1$.
From Lemma \ref{lem:Xxp}, we have
\[
\ket{\Gamma_{l\vv_\sigma}^{(\sigma)}\lmk \ket{\xi_i}\bra{\eta_i}\rmk}
\bra{\Gamma_{l\vv_\sigma}^{(\sigma)}\lmk \ket{\xi_i}\bra{\eta_i}\rmk}
=X_{l,\xi_i,\ket{\eta_i}\bra{\eta_i}}^{(\sigma)},
\]
with notation in Lemma \ref{lem:Xxp}.
By Lemma \ref{lem:Xxp}, and Lemma \ref{lem:a1a4} we have
\[
s\lmk\ket{\Gamma_{l\vv_\sigma}^{(\sigma)}\lmk \ket{\xi_i}\bra{\eta_i}\rmk}
\bra{\Gamma_{l\vv_\sigma}^{(\sigma)}\lmk \ket{\xi_i}\bra{\eta_i}\rmk}\rmk=
s\lmk X_{l,\xi_i,\ket{\eta_i}\bra{\eta_i}}^{(\sigma)}\rmk\le \tau_{y_\sigma}\lmk
s\lmk\omega_{\xi_i}\vert_{\caA_{\sigma,l}}\rmk\rmk
\le
\tau_{y_\sigma}\lmk s\lmk\omega_\sigma\vert_{\caA_{\sigma,l}}\rmk\rmk.
\]
where $y_R=0$ and $y_L=l$.
Hence each term in the decomposition $\Gamma_{l\vv_\sigma}^{(\sigma)}(X)=\sum_{i=1}^{m_\sigma}c_i\Gamma_{l\vv_\sigma}^{(\sigma)}\lmk \ket{\xi_i}\bra{\eta_i}\rmk$
is in $\tau_{y_\sigma}\lmk s\lmk\omega_\sigma\vert_{\caA_{\sigma,l}}\rmk\rmk\lmk \bigotimes_{i=0}^{l-1}\cc^n\rmk$.
From this, we obtain (\ref{eq:ranX}).

Next we show 
\begin{align}\label{eq:Xsnd}
0\le X_{l,\xi,(1-s(\rho_\sigma))}^{(\sigma)}\le \lV T_{\vv_\sigma}^{l}\lmk\unit-P_{\{1\}}^{T_{\vv_\sigma}}\rmk\rV,\quad
\text{for all } \; \xi\in\caK_\sigma,\quad\text{with}\quad \lV\xi\rV=1.
\end{align}
The first inequality is already proven in Lemma \ref{lem:Xxp2}.
The second one follows from the following calculation for any $\zeta\in  \bigotimes_{i=0}^{l-1}\cc^n$:
\begin{align*}
&\braket{\zeta}{ X_{l,\xi,(1-s(\rho_\sigma))}^{(\sigma)}\zeta}=
\lV \lmk 1-s(\rho_\sigma)\rmk
\lmk \sum_{\nu^{(l)}}\braket{\wsn{l}}{\zeta}
{\widehat{v_{\nu^{(l,\sigma)},\sigma}}}\rmk^*\xi\rV^2
\le\lmk
\sum_{\nu^{(l)}}\lv \braket{\wsn{l}}{\zeta}\rv
\lV \lmk 1-s(\rho_\sigma)\rmk
{\widehat{v_{\nu^{(l,\sigma)},\sigma}}}^*\xi\rV
\rmk^2\\
&\le\lmk \sum_{\nu^{(l)}}\lv \braket{\wsn{l}}{\zeta}\rv^2\rmk
\lmk\sum_{\nu^{(l)}}
\lV \lmk 1-s(\rho_\sigma)\rmk
{\widehat{v_{\nu^{(l,\sigma)},\sigma}}}^*\xi\rV^2
\rmk
=\lV\zeta\rV^2 
\braket{\xi}{\lmk
 T_{\vv_\sigma}^{l}\lmk 1-s(\rho_\sigma)\rmk
\rmk
\xi}\\
&=\lV\zeta\rV^2 
\braket{\xi}{\lmk
 T_{\vv_\sigma}^{l}\circ\lmk\unit-P_{\{1\}}^{T_{\vv_\sigma}}\rmk\lmk 1-s(\rho_\sigma)\rmk
+T_{\vv_\sigma}^{l}\circ P_{\{1\}}^{T_{\vv_\sigma}}\lmk 1-s(\rho_\sigma)\rmk
\rmk
\xi}\\
&=\lV\zeta\rV^2 
\braket{\xi}{\lmk
 T_{\vv_\sigma}^{l}\circ\lmk\unit-P_{\{1\}}^{T_{\vv_\sigma}}\rmk\lmk 1-s(\rho_\sigma)\rmk
\rmk
\xi}
\le
\lV
 T_{\vv_\sigma}^{l}\circ\lmk\unit-P_{\{1\}}^{T_{\vv_\sigma}}\rmk\rV
\lV\zeta\rV^2\lV \xi\rV^2.
\end{align*}
We used Lemma \ref{lem:vucp} in the last equality.

Finally, we claim 
\begin{align}\label{eq:Xthird}
\ran\lmk X_{l,\xi,s(\rho_\sigma)}^{(\sigma)}\rmk\subset \Gamma_{l,\vv_\sigma}^{(\sigma)}\lmk B\lmk\caK_\sigma\rmk s(\rho_\sigma)\rmk,\quad
\text{for all} \quad \xi\in\caK_\sigma,\quad \text{with}\quad \lV \xi\rV=1.
\end{align}
To see this, decompose $s(\rho_\sigma)$ as $s(\rho_\sigma)=\sum_i\ket{x_i}\bra{x_i}$
with CONS $\{x_i\}$ of $s(\rho_\sigma)\caK_\sigma$.
We then obtain
\begin{align*}
&X_{l,\xi,s(\rho_\sigma)}^{(\sigma)}:=
\sum_{\mu^{(l)},\nu^{(l)}}\braket{\xi}{{\widehat{v_{\mu^{(l,\sigma)},\sigma}}}s(\rho_\sigma) {\widehat{v_{\nu^{(l,\sigma)},\sigma}}}^*\xi}
\ket{\wsn{l}}\bra{\ws{l}}
=\sum_i\sum_{\mu^{(l)},\nu^{(l)}}\braket{\xi}{{\widehat{v_{\mu^{(l,\sigma)},\sigma}}}\lmk \ket{x_i}\bra{x_i}\rmk {\widehat{v_{\nu^{(l,\sigma)},\sigma}}}^*\xi}
\ket{\wsn{l}}\bra{\ws{l}}\\
&=\sum_i \ket{\Gamma_{l,\vv_\sigma}^{(\sigma)}\lmk \ket{\xi}\bra{x_i}\rmk}\bra{\Gamma_{l,\vv_\sigma}^{(\sigma)}\lmk \ket{\xi}\bra{x_i}\rmk},
\end{align*}
for any $\xi\in\caK_\sigma$, with $\lV \xi\rV=1$.
As each $\Gamma_{l,\vv_\sigma}^{(\sigma)}\lmk \ket{\xi}\bra{x_i}\rmk$ is in $\Gamma_{l,\vv_\sigma}^{(\sigma)}\lmk B\lmk\caK_\sigma\rmk s(\rho_\sigma)\rmk$,
this proves the claim.

Now we derive the claim of the Lemma, combining the above statements.
Let $\rho_{\omega_\sigma}$ be the density matrix given by
Lemma \ref{lem:fs}.
It can be decomposed as 
$\rho_{\omega_\sigma}=\sum_i \lambda_i\ket{\eta_i}\bra{\eta_i}$,
with numbers $\lambda_i>0$, $\sum_i\lambda_i=1$, and
CONS $\{\eta_i\}$ of $\caK_\sigma$.
Then we have
\begin{align*}
&C_\sigma\le \omega_\sigma\lmk
\tau_{-y_\sigma}\lmk\ket{\Gamma_{l\vv_\sigma}^{(\sigma)}(X)}\bra{\Gamma_{l\vv_\sigma}^{(\sigma)}(X)}
\rmk
\rmk\\
&=\sum_i\lambda_i
\lmk
\braket{\Gamma_{l,\vv_\sigma}^{(\sigma)}\lmk X\rmk}{X_{l,\eta_i,s(\rho_\sigma)}^{(\sigma)}\Gamma_{l,\vv_\sigma}^{(\sigma)}\lmk X\rmk}
+\braket{\Gamma_{l,\vv_\sigma}^{(\sigma)}\lmk X\rmk}{X_{l,\eta_i,\lmk 1-s(\rho_\sigma)\rmk}^{(\sigma)}\Gamma_{l,\vv_\sigma}^{(\sigma)}\lmk X\rmk}
\rmk.
\end{align*}
In the first inequality, we used (\ref{eq:ranX}) and the definition of $C_\sigma$.
The equality follows from Lemma \ref{lem:fs}, and the definition of $X_{l,\xi,p}^{(\sigma)}$.
By the third claim (\ref{eq:Xthird}) and the orthogonality of
$\Gamma_{l\vv_\sigma}^{(\sigma)}(X)$ and 
$\Gamma_{l\vv_\sigma}^{(\sigma)}\lmk B(\caK_\sigma s(\rho_\sigma))\rmk$,
the first term on the right hand side is $0$.
The second term can be bounded using (\ref{eq:Xsnd}) and we have
\begin{align}&C_\sigma\le 
\lV T_{\vv_\sigma}^{l}\lmk\unit-P_{\{1\}}^{T_{\vv_\sigma}}\rmk\rV
\le \frac 12C_\sigma.
\end{align}
This is a contradiction. Hence we proved the Lemma.
\end{proof}
\begin{proofof}[Lemma \ref{lem:bijecs}]
Set $c_\sigma:=\inf\{\sigma\lmk\tilde \rho_\sigma\rmk\setminus \{0\}\}>0$.
By the routine calculation from Part I, we obtain
\[
\lv
\lV\Gamma_{l,\vv_\sigma}^{(\sigma)}\lmk X\rmk\rV^2-\rho_\sigma\lmk X^* X\rmk
\rv
\le
\lV T_{\vv_\sigma}^{l}\lmk\unit-P_{\{1\}}^{T_{\vv_\sigma}}\rmk\rV m_\sigma^2 c_\sigma^{-1}\rho_\sigma\lmk X^*X\rmk,\quad X\in B(\caK_\sigma)s(\rho_\sigma).
\]
We set
\[
 l_\sigma'
:=\min\left\{
\tilde l_\sigma \le l\in\nan\mid
\sup_{l':l\le l'}\left\{
\lV T_{\vv_\sigma}^{l'}\lmk\unit-P_{\{1\}}^{T_{\vv_\sigma}}\rmk\rV m_\sigma^2 c_\sigma^{-1}\right\}
<\frac 12
\right\}
\]
with $\tilde l_\sigma$ given in the previous Lemma.
Then for $l_\sigma'\le l$, we have
\[
 \frac{c_\sigma}{2}\Tr\lmk X^* X\rmk\le \frac 12 \rho_\sigma\lmk X^*X\rmk\le \lV \Gamma_{l,\vv_\sigma}^{(\sigma)}\lmk X\rmk\rV^2,\quad X\in B(\caK_\sigma)s(\rho_\sigma).
\]
This means $\left.\Gamma_{l,\vv_\sigma}^{(\sigma)}\right\vert_{B(\caK_\sigma)s(\rho_\sigma)}$ is injective for $l_\sigma'\le l$.
As we have chosen $l_\sigma'$ so that $\tilde l_\sigma\le l_\sigma'$, it is also onto $\Gamma_{l,\vv_\sigma}^{(\sigma)}\lmk B(\caK_\sigma)\rmk$.
That $\Gamma_{l,\vv_\sigma}^{(\sigma)}\lmk B(\caK_\sigma)\rmk=
\tau_{y_{\sigma}}\lmk
s\lmk \left. \omega_\sigma\right\vert _{\caA_{\sigma,l}}\rmk
\rmk\bigotimes_{i=0}^{l-1}\cc^n$ is proven in Lemma \ref{lem:Xxp2}.
\end{proofof}
As a result, we obtain the following Lemma.
\begin{lem}\label{lem:yaab}
Assume [A1], [A3], [A4], and [A5]. Let $\sigma=L,R$. Let $\vv_\sigma$ be the $n$-tuple of
elements in  $B(\caK_\sigma)$ given in Notation \ref{nota:vv},
$\rho_\sigma$
the state given in Lemma \ref{lem:vucp},
and  $V_{a\sigma}$the unitary given in Lemma \ref{lem:primu}.
Let $l_\sigma'\in\nan$ be the number given in Lemma \ref{lem:bijecs}.
For each $a=0,\ldots, k_\sigma$,
let  $\{g_{\alpha}^{(a)}\}_{\alpha=1}^{{n_{0}^{(\sigma)}}}$ be a CONS
of $r_{a\sigma}\caK_\sigma\simeq \cc^{n_0^{(\sigma)}}$ given by
$g_\alpha^{(a)}:=V_{a\sigma}\chi_\alpha^{(n_0^{(\sigma)})}$.
Then
there exist a
 $y_{a,\alpha,\beta,\sigma}^{(l)}\in B(\caK_\sigma)$,
for each $a=0,\ldots,k_\sigma$, $\alpha,\beta=1,\ldots,{n_{0}^{(\sigma)}}$, and 
$l\ge l_\sigma'$,
satisfying the followings.
\begin{enumerate}
\item[(1)] For each $l\ge l_\sigma'$, the set
 $\{y_{a,\alpha,\beta,\sigma}^{(l)}\}_{a=0,\ldots,k_\sigma,
\alpha,\beta=1,\ldots,{n_{0}^{(\sigma)}}}$ is a basis of $\caK_l(\vv_\sigma)$.
\item[(2)]For any $a_1=0,\ldots,k_\sigma$, $\alpha_1,\alpha_2,\beta_1,\beta_2=1,\ldots,n_0^{(\sigma)}$, and $l_1,l_2\ge l_\sigma'$, we have
\[
y_{a_1,\alpha_1,\beta_1,\sigma}^{(l_1)}y_{0,\alpha_2,\beta_2,\sigma}^{(l_2)}
=\delta_{\beta_1\alpha_2} y_{a_1,\alpha_1,\beta_2,\sigma}^{(l_1+l_2)}.
\]
\item[(3)]
For any $a=0,\ldots,k_\sigma$, $\alpha,\beta=1,\ldots,{n_{0}^{(\sigma)}}$, and 
$l\ge l_\sigma'$,
\[
y_{a,\alpha,\beta,\sigma}^{(l)}s(\rho_\sigma)
=\ket{g_\alpha^{(a)}}\bra{g_{\beta}^{(0)}}.
\]
\item[(4)]If $X\in \caK_l(\vv_\sigma)$, $l\ge l_\sigma'$,
satisfies $Xs(\rho_\sigma)=0$, then $X=0$.
\end{enumerate}
\end{lem}
\begin{proof}
By Lemma \ref{lem:bijecs}, for $l\ge l_\sigma'$,
the set $\left
\{\Gamma_{l,\vv_\sigma}^{(\sigma)}
\lmk\ket{g_\alpha^{(a)}}\bra{g_{\beta}^{(0)}}\rmk\right\}_{a=0,\ldots,k_\sigma,\alpha,\beta=1,\ldots,{n_{0}^{(\sigma)}}}$ is linearly independent.
Therefore, there exist $\xi_{a,\alpha,\beta}^{(l)}\in \bigotimes_{i=0}^{l-1}\cc^n$,
$a=0,\ldots,k_\sigma,\alpha,\beta=1,\ldots,{n_{0}^{(\sigma)}}$ such that
\[
\braket{\xi_{a,\alpha,\beta}^{(l)}}{\Gamma_{l,\vv_\sigma}^{(\sigma)}
\lmk\ket{g_{\alpha'}^{(a')}}\bra{g_{\beta'}^{(0)}}\rmk}
=\delta_{aa'}\delta_{\alpha\alpha'}\delta_{\beta\beta'}.
\]
Set
\[
y_{a,\alpha,\beta,\sigma}^{(l)}
:=\sum_{\mu^{(l)}}\braket{\ws{l}}{\xi_{a,\alpha,\beta}^{(l)}}\widehat{v_{\mu^{(l,\sigma)},\sigma}}\in \caK_{l}(\vv_\sigma),\quad
a=0,\ldots,k_\sigma,\quad \alpha,\beta=1,\ldots, n_0^{(\sigma)},\quad l\ge l_{\sigma}'.
\]
By a straightforward calculation, we have
\[
\delta_{aa'}\delta_{\alpha\alpha'}\delta_{\beta\beta'}
=\braket{\xi_{a,\alpha,\beta}^{(l)}}{\Gamma_{l,\vv_\sigma}^{(\sigma)}
\lmk\ket{g_{\alpha'}^{(a')}}\bra{g_{\beta'}^{(0)}}\rmk}
=\overline{\braket{g_{\alpha'}^{(a')}}{y_{a,\alpha,\beta,\sigma}^{(l)} g_{\beta'}^{(0)}}}.
\]
This means
\begin{align}\label{eq:yc}
y_{a,\alpha,\beta,\sigma}^{(l)}s(\rho_\sigma)
=\ket{g_\alpha^{(a)}}\bra{g_{\beta}^{(0)}},\quad
a=0,\ldots,k_\sigma,\quad \alpha,\beta=1,\ldots, n_0^{(\sigma)},\quad l\ge l_{\sigma}',
\end{align}
corresponding to (3) in the claim.
This means  $\left\{y_{a,\alpha,\beta,\sigma}^{(l)}
\right\}_{a=0,\ldots,k_\sigma,\alpha,\beta=1,\ldots,{n_{0}^{(\sigma)}}}$
is linearly independent, spanning $({n_{0}^{(\sigma)}})^2(k_\sigma+1)$-dimensional
subspace of $\caK_l(\vv_\sigma)$.
However, from Lemma \ref{lem:bijecs} ,
the dimension of $\caK_l(\vv_\sigma)$ is  $({n_{0}^{(\sigma)}})^2(k_\sigma+1)$.
Hence, $\{y_{a,\alpha,\beta,\sigma}^{(l)}\}_{a=0,\ldots,k_\sigma,
\alpha,\beta=1,\ldots,{n_{0}^{(\sigma)}}}$ is a basis of $\caK_l(\vv_\sigma)$,
proving (1).

Let us prove (4). Let $X\in \caK_l(\vv_\sigma)$, $l\ge l_\sigma'$,
such that $Xs(\rho_\sigma)=0$.
Then $X$ can be written as a linear combination $X=\sum_{a\alpha\beta}C_{a\alpha\beta}y_{a,\alpha,\beta,\sigma}^{(l)}$,
and we obtain
\[
0=Xs(\rho_\sigma)
=\sum_{a\alpha\beta}C_{a\alpha\beta}\ket{g_\alpha^{(a)}}\bra{g_{\beta}^{(0)}}.
\]
This implies $C_{a\alpha\beta}=0$ and we conclude $X=0$.

 To prove (2), note that
\[X=
y_{a_1,\alpha_1,\beta_1,\sigma}^{(l_1)}y_{0,\alpha_2,\beta_2,\sigma}^{(l_2)}
-\delta_{\beta_1\alpha_2} y_{a_1,\alpha_1,\beta_2,\sigma}^{(l_1+l_2)}
\in\caK_{l_1+l_2}(\vv_\sigma)
\]
and 
\[
Xs(\rho_\sigma)
=y_{a_1,\alpha_1,\beta_1,\sigma}^{(l_1)}
\ket{g_{\alpha_2}^{(0)}}\bra{g_{\beta_2}^{(0)}}-\delta_{\beta_1\alpha_2}\ket{g_{\alpha_1}^{(a_1)}}\bra{g_{\beta_2}^{(0)}}
=\ket{g_{\alpha_1}^{(a_1)}}\bra{g_{\beta_1}^{(0)}}\cdot\ket{g_{\alpha_2}^{(0)}}\bra{g_{\beta_2}^{(0)}}-\delta_{\beta_1\alpha_2}\ket{g_{\alpha_1}^{(a_1)}}\bra{g_{\beta_2}^{(0)}}
=0.
\]
Applying the above argument, we conclude $X=0$, proving (2).
\end{proof}

\section{Deformation of $\vv_\sigma$}\label{sec:roku}
By Lemma \ref{lem:primu}, we have 
${n_{0}^{(\sigma)}}=\rank r_{a\sigma}=\rank s(\rho_\sigma)$, for all $a=0,\ldots, k_\sigma$.
This means that we can identify $B(\caK_\sigma)$ with $\Mat_{{n_{0}^{(\sigma)}}}\otimes\Mat_{k_\sigma+1}$.
We introduce two conditions.
\begin{defn}
Let $n, n_0\in\nan$ and  $k\in\nan\cup\{0\}$. Let
$\oo=(\omega_\mu)_{\mu=1}^n\in\Mat_{n_0}^{\times n}$,
$\vv=(v_\mu)_{\mu=1}^n\in\lmk\Mat_{n_0}\otimes\Mat_{k+1}\rmk^{\times n}$, and
 $\lal=(\lambda_a)_{a=0}^k\in\cc^{k+1}$.
Let $l_0\in\nan$ and $y_{a,\alpha,\beta}^{(l)}\in \Mat_{n_{0}}\otimes\Mat_{k+1}$,
for $a=0,\ldots,k$, $\alpha,\beta=1,\ldots,n_0$, and $l\ge l_0$.
We say that the septuplet $(n_0,k,\oo,\vv,\lal,l_0,\{y_{a,\alpha,\beta}^{(l)}\})$
satisfies {\it Condition 5} if the following holds.
\begin{enumerate}
\item[(i)]
$\lambda_0=1$ and $0<\lv\lambda_a\rv<1$ for all $a\ge 1$.
\item[(ii)]$v_\mu\in \Mat_{n_0}\otimes\DT_{k+1}$, $\mu=1,\ldots,n$.
\item[(iii)] $(\unit\otimes E_{aa}^{(0,k)})v_\mu(\unit\otimes E_{aa}^{(0,k)})
=\lambda_a \omega_\mu\otimes  E_{aa}^{(0,k)}$,
for all $a=0,\ldots,k$, and $\mu=1,\ldots,n$.
\item[(iv)]
\begin{enumerate}
\item[(1)] For each $l\ge l_0$, the set
 $\{y_{a,\alpha,\beta}^{(l)}\}_{a=0,\ldots,k,
\alpha,\beta=1,\ldots,n_0}$ is a basis of $\caK_l(\vv)$.
\item[(2)]For any $a_1=0,\ldots,k$, $\alpha_1,\alpha_2,\beta_1,\beta_2=1,\ldots,n_0$, and $l_1,l_2\ge l_0$, we have
\[
y_{a_1,\alpha_1,\beta_1}^{(l_1)}y_{0,\alpha_2,\beta_2}^{(l_2)}
=\delta_{\beta_1\alpha_2} y_{a_1,\alpha_1,\beta_2}^{(l_1+l_2)}
\]
\item[(3)]
For any $\alpha,\beta=1,\ldots,n_0$, and $l\ge l_0$,
\[
y_{0,\alpha,\beta}^{(l)}\lmk\unit\otimes E_{00}^{(0,k)}\rmk
=\lmk\unit\otimes E_{00}^{(0,k)}\rmk y_{0,\alpha,\beta}^{(l)}\lmk\unit\otimes E_{00}^{(0,k)}\rmk.
\]
\end{enumerate}
\end{enumerate}
\end{defn}
\begin{defn}
Let $n, n_0\in\nan$ and  $k\in\nan\cup\{0\}$. Let
$\oo=(\omega_\mu)_{\mu=1}^n \in\Mat_{n_0}^{\times n}$,
$\vv=(v_\mu)_{\mu=1}^n\in\lmk\Mat_{n_0}\otimes\Mat_{k+1}\rmk^{\times n}$, and
 $\lal=(\lambda_a)_{a=0}^k\in\cc^{k+1}$.
Let $l_0\in\nan$ and $y_{a,\alpha,\beta}^{(l)}\in \Mat_{n_{0}}\otimes\Mat_{k+1}$,
for $a=0,\ldots,k$, $\alpha,\beta=1,\ldots,n_0$, and $l\ge l_0$.
Let $i\in\{0,\ldots,k\}$. We say that the septuplet $(n_0,k,\oo,\vv,\lal,l_0,\{y_{a,\alpha,\beta}^{(l)}\})$
satisfies {\it Condition 6-i} if the followings hold.
\begin{enumerate}
\item[(i)]The septuplet $(n_0,k,\oo,\vv,\lal,l_0,\{y_{a,\alpha,\beta}^{(l)}\})$
satisfies {\it Condition 5}.
\item[(ii)] There exists a $Y\in \DT_{0, k+1}$ such that $[\Lambda_\lal,Y]=0$.
\item[(iii)] For any $\alpha,\beta=1,\ldots,n_0$, $l\ge l_0$, and $Y$ in (ii), we have
\[
y_{0,\alpha,\beta}^{(l)}-\zeij{\alpha\beta}\otimes\Lambda_\lal^l (1+Y)^l
\in \mnz\otimes\sum_{a,a':a-a'\ge i+1} E_{aa}^{(0,k)}\Mat_{k+1}E_{a'a'}^{(0,k)}.
\]
\end{enumerate}
\end{defn}
\begin{rem}
When we would like to specify $Y$, we say
the septuplet $(n_0,k,\oo,\vv,\lal,l_0,\{y_{a,\alpha,\beta}^{(l)}\})$
satisfies {\it Condition 6-i}
with respect to $Y$.
We may set $Y=0$ for  {\it Condition 6-0}.
\end{rem}

Out of our $\vv_\sigma$, we can construct a
septuplet satisfying {\it Condition 6-0}.
\begin{lem}\label{lem:ohy}
Assume [A1],[A3],[A4], and [A5]. We use Notation \ref{nota:vv} and Notation \ref{nota:oa}.
Then for each $\sigma=R,L$, there exist 
$\oo^{(\sigma)}\in \Primz_u(n,n_0^{(\sigma)})$,
$\vv^{(\sigma)}\in \lmk \Mat_{{n_{0}^{(\sigma)}}}\otimes\Mat_{k_\sigma+1}\rmk^{\times n}$,
 $\lal^{(\sigma)}=(\lambda^{(\sigma)}_a)_{a=0,\ldots,k_\sigma}\in \cc^{k_\sigma+1}$,
${l_{0}^{(\sigma)}}\in\nan$,
and 
 $\{y_{a,\alpha,\beta}^{(l,\sigma)}\}_{a=0,\ldots,k_\sigma,
\alpha,\beta=1,\ldots,{n_{0}^{(\sigma)}},l\ge l_0^{(\sigma)}}\subset  \Mat_{{n_{0}^{(\sigma)}}}\otimes\Mat_{k_\sigma+1}$
satisfying the followings.
\begin{enumerate}
\item The septuplet $({n_{0}^{(\sigma)}},k_\sigma,\oo^{(\sigma)},\vv^{(\sigma)},\lal^{(\sigma)},l_0^{(\sigma)},\{y_{a,\alpha,\beta}^{(l,\sigma)}\})$
satisfies {\it Condition 6-0}.
\item For  the state $\omega_\sigma$ in  [A4] and $ l\in\nan$,
 $\tau_{y_\sigma}\lmk s(\omega_{\sigma}\vert_{\caA_{\sigma,l}})\rmk$ is equal to
the orthogonal projection onto $\Gamma_{l,\vv^{(\sigma)}}^{(\sigma)}\lmk \Mat_{{n_{0}^{(\sigma)}}}\otimes\Mat_{k_\sigma+1}\rmk$,
where $y_R=0$ and $y_L=l$.
\item The triple $(\Mat_{{n_{0}^{(\sigma)}}}, \oo^{(\sigma)}, \rho_\sigma\vert _{\Mat_{{n_{0}^{(\sigma)}}}})$ $\sigma$-generates $\omega_\infty$.
\item There exist strictly positive elements 
$h_{a\sigma}$ $a=0,\ldots, k$ in $\Mat_{n_0^{(\sigma)}}$ with $h_{0\sigma}=\unit$ such that for all $ l\ge l_0^{(\sigma)}$, $a=0,\ldots, k_\sigma$, and $\alpha,\beta=1,\ldots,n_0^{(\sigma)}$,
\[
y_{a,\alpha,\beta}^{(l,\sigma)}\lmk \unit\otimes E_{00}^{(0,k_\sigma)}\rmk
=h_{a\sigma}^{\frac 12}e_{\alpha\beta}^{(n_0^{(\sigma)})}\otimes E_{a0}^{(0,k_\sigma)}.
\]
\end{enumerate}
\end{lem}
\begin{proof}
We use Notation \ref{nota:vv} and Notation \ref{nota:oa}.
Define $\oo^{(\sigma)}:=\widetilde{\vv_{\sigma}}\in \Mat_{{n_{0}^{(\sigma)}}}^{\times n}$ under the identification $\Mat_{{n_{0}^{(\sigma)}}}\simeq B(s(\rho_\sigma)\caK_\sigma)$.
By Lemma \ref{lem:srs}, we have $\oo^{(\sigma)}\in\Primz_u(n, n_0)$,
and $(\Mat_{{n_{0}^{(\sigma)}}}, \oo^{(\sigma)}, \rho_\sigma\vert _{\Mat_{{n_{0}^{(\sigma)}}}})$ $\sigma$-generates $\omega_\infty$.
(This proves {\it 3} of the Lemma.)

From  Notation \ref{nota:oa} and
Lemma \ref{lem:primu}, for each $a=0,\ldots, k_\sigma$,
there is a unitary $V_{a\sigma}:\cc^{n_0^{(\sigma)}}\to r_{a\sigma}\caK_{\sigma}$
and $c_{a\sigma}\in \bbT$ such that
\begin{align}\label{eq:ntlm}
r_{T_{\oo_{a,\sigma}}}^{-\frac 12}t_{a\sigma}^{-\frac 12}r_{a\sigma}
v_{\mu\sigma}r_{a\sigma}t_{a\sigma}^{\frac 12}=u_{\mu a\sigma}=c_{a\sigma}V_{a\sigma}\widetilde{v_{\mu\sigma}}V_{a\sigma}^*
=c_{a\sigma}V_{a\sigma}\omega_\mu^{(\sigma)}V_{a\sigma}^*,\quad
\mu=1,\ldots,n.
\end{align}
Note that $r_{T_{\oo_{0,\sigma}}}=1$, and we may choose $c_{0\sigma}=1$, $t_{0\sigma}=r_{0\sigma}$, and $V_{0\sigma}$ is the identity map.

Define an invertible element
$R^{(\sigma)}:=\bigoplus_{a=0}^{k_\sigma} t_{a\sigma}^{\frac 12}$ in $B(\caK_{\sigma})$,
and  $\lal^{(\sigma)}=(\lambda^{(\sigma)}_a)_{a=0,\ldots,k_\sigma}\in \cc^{k_\sigma+1}$
by $\lambda^{(\sigma)}_a=r_{T_{\oo_{a\sigma}}}^{\frac 12}c_{a\sigma}$.
Furthermore, we set $h_{a\sigma}^{\frac 12}:= V_{a\sigma}^*t_{a\sigma}^{-\frac 12} V_{a\sigma}$.
Note that $h_{0\sigma}=1$.
Using unitaries $V_{a\sigma}$ above, we define a linear map
$V^{\sigma}:\caK_\sigma\to \cc^{n_0^{(\sigma)}}\otimes\cc^{k_\sigma+1}$
by
\[
V^{(\sigma)}\xi
:=\sum_{a=0}^{k_\sigma}
\lmk V_{a\sigma}^* r_{a\sigma}\xi\rmk
\otimes f_{a}^{(0,k_\sigma)},\quad \xi\in\caK_\sigma.
\]
By definition, this is unitary.

Let
 $ l_\sigma'$, and
 $\{y_{a,\alpha,\beta,\sigma}^{(l)}\}_{a=0,\ldots,k_\sigma,
\alpha,\beta=1,\ldots,{n_{0}^{(\sigma)}},l\ge l_\sigma'}\subset B(\caK_\sigma)$
given in Lemma \ref{lem:yaab}.

We define $\vv^{(\sigma)}=(v_{\mu}^{(\sigma)})_{\mu=1}^{n}\in 
\lmk \Mat_{n_{0}^{(\sigma)}}\otimes\Mat_{k_\sigma+1}\rmk^{\times n}$ by
\begin{align*}
v_{\mu}^{(\sigma)}:=
V^{(\sigma)}(R^{(\sigma)})^{-1}v_{\mu\sigma}R^{(\sigma)}{V^{(\sigma)}}^*,\quad
\mu=1,\ldots,n.
\end{align*}
We also set
\begin{align*}
y_{a,\alpha,\beta}^{(l,\sigma)}
=V^{(\sigma)}(R^{(\sigma)})^{-1}y_{a,\alpha,\beta,\sigma}^{(l)}
R^{(\sigma)}{V^{(\sigma)}}^*,\quad
a=0,\ldots,k_\sigma,\;
\alpha,\beta=1,\ldots,{n_{0}^{(\sigma)}},\; l\ge l_\sigma'.
\end{align*}
Furthermore, we set $l_0^{(\sigma)}:=l_\sigma'$.
As $v_\mu^{(\sigma)}$ is similar to $v_{\mu\sigma}$ with common
invertible operator $R^{(\sigma)}{V^{(\sigma)}}^*$,
Lemma \ref{lem:Xxp2} implies {\it 2.} of the current Lemma.

Next we show that the septuplet $({n_{0}^{(\sigma)}},k_\sigma,\oo^{(\sigma)},\vv^{(\sigma)},\lal^{(\sigma)},l_0^{(\sigma)},\{y_{a,\alpha,\beta}^{(l,\sigma)}\})$
satisfies {\it Condition 5}.
(i) follows from Lemma \ref{lem:rta} and the above remark on $c_{0\sigma}$ etc.
Recall the definition of $\caK_{a,\sigma}$, given as an
$v_{\nu}^*$-invariant subspace.
This property is translated to (ii) of {\it Condition 5}
for $\vv^{(\sigma)}$.
The equality (\ref{eq:ntlm}) implies (iii) of  {\it Condition 5}.
(1), (2) of (iv) follows from the fact that 
$y_{a,\alpha,\beta}^{(l,\sigma)}$ (resp. $\vv^{(\sigma)}$) and
$y_{a,\alpha,\beta,\sigma}^{(l)}$ (resp. $\vv_{\sigma}$) are similar to each other with common 
invertible operator $R^{(\sigma)}{V^{(\sigma)}}^*$.
(3) of (iv) follows from (3) of Lemma \ref{lem:yaab} and the fact that
$R^{(\sigma)}{V^{(\sigma)}}^*\lmk\unit\otimes E_{00}^{(0,k_\sigma)}\rmk
=s(\rho_\sigma)R^{(\sigma)}{V^{(\sigma)}}^*$.

It is left to prove (ii), (iii) of {\it Condition 6-0}.
We set $Y=0$.
Then clearly (ii) holds.
From (ii), (iii) of {\it Condition 5} which we have already proved,
we see that
\[
(\unit\otimes E_{aa}^{(0,k)})v_{\mu^{(l)}}^{(\sigma)}(\unit\otimes E_{aa}^{(0,k)})
=\lmk \lambda_a^{(\sigma)}\rmk^l  \omega_{\mu^{(l)}}^{(\sigma)}\otimes  E_{aa}^{(0,k)},
\quad \mu^{(l)}\in\{1,\ldots,n\}^{\times l}, \quad l\in\nan,\quad
a=0,\ldots,k.
\]
From this, we  obtain
\[
(\unit\otimes E_{aa}^{(0,k)})y_{0,\alpha,\beta}^{(l,\sigma)}
(\unit\otimes E_{aa}^{(0,k)})
=\lmk \lambda_a^{(\sigma)}\rmk^l\zeij{\alpha\beta}\otimes  E_{aa}^{(0,k)}.
\]
As $y_{a,\alpha,\beta}^{(l,\sigma)}$ is a lower triangular matrix
(because each  $v_\mu^{(\sigma)}$ is),
this implies (iii) of {\it Condition 6-0}.

Lastly, we prove {\it 4}.
As we observed, we have
\begin{align*}
R^{(\sigma)}{V^{(\sigma)}}^*\lmk\unit\otimes E_{00}^{(0,k_\sigma)}\rmk
=s(\rho_\sigma)R^{(\sigma)}{V^{(\sigma)}}^*,\quad
\lmk\unit\otimes E_{00}^{(0,k_\sigma)}\rmk{V^{(\sigma)}}R^{(\sigma)}
={V^{(\sigma)}}R^{(\sigma)}s(\rho_\sigma).
\end{align*}
Therefore, we have
\begin{align*}
&y_{a,\alpha,\beta}^{(l,\sigma)}\lmk\unit\otimes E_{00}^{(0,k_\sigma)}\rmk
=V^{(\sigma)}(R^{(\sigma)})^{-1}y_{a,\alpha,\beta,\sigma}^{(l)}
R^{(\sigma)}{V^{(\sigma)}}^*\lmk\unit\otimes E_{00}^{(0,k_\sigma)}\rmk
=V^{(\sigma)}(R^{(\sigma)})^{-1}y_{a,\alpha,\beta,\sigma}^{(l)}
s(\rho_\sigma)R^{(\sigma)}{V^{(\sigma)}}^*\\
&=V^{(\sigma)} t_{a\sigma}^{-\frac 12}\ket{g_{\alpha}^{(a)}}\bra{g_{\beta}^{(0)}}R^{(\sigma)}{V^{(\sigma)}}^*
=\lmk h_{a\sigma}^{\frac 12}\otimes E_{aa}^{(0,k_\sigma)}\rmk 
V^{(\sigma)} \ket{g_{\alpha}^{(a)}}\bra{g_{\beta}^{(0)}}R^{(\sigma)}{V^{(\sigma)}}^*
=h_{a\sigma}^{\frac12}\zeij{\alpha\beta}\otimes E_{a0}^{(0,k_\sigma)}.
\end{align*}

\end{proof}
We prove that {\it Condition 6-0} implies {\it Condition 6-$k_\sigma$},
inductively.
\begin{lem}\label{lem:c6in}
Let $0 \le i \le k-1$.
Assume that the septuplet $(n_0,k,\oo,\vv,\lal,l_0,\{y_{a,\alpha,\beta}^{(l)}\})$
satisfies {\it Condition 6-i} with respect to $Y$.
Then there exist 
$\{J_j\}_{j={i+1}}^{k}\subset \mnz$
and
$\{c_j\}_{j={i+1}}^{k}\subset \cc$
satisfying the followings.
\begin{enumerate}
\item
Set 
$R:=\unit-\sum_{j=i+1}^{k} J_j\otimes E_{j,j-(i+1)}^{(0,k)}$.
Then 
$R$ is invertible, $R-\unit\in\mnz\otimes \DT_{0,k+1}$
and
$R\lmk \unit\otimes E_{00}^{(0,k)}\rmk
=R^{-1}\lmk \unit\otimes E_{00}^{(0,k)}\rmk
= \unit\otimes E_{00}^{(0,k)}$.
\item
Set
$Y':=\sum_{i+1\le j\le k:\lambda_j=\lambda_{j-(i+1)}} c_j E_{j,j-(i+1)}^{(0,k)}$.
The septuplet $(n_0,k,\oo,R\vv R^{-1},\lal,l_0,\{R y_{a,\alpha,\beta}^{(l)}R^{-1}\})$
satisfies {\it Condition 6-(i+1)} with respect to $Y+Y'$.
\end{enumerate}
\end{lem}
For the proof, we need the following Lemma.
\begin{lem}\label{lem:key}
Let $n_0,l_0\in\nan$, $\lambda\in\cc\setminus \{0\}$
and suppose that
 $x_{\alpha,\beta}^{(l)}\in \mnz$,
$\alpha,\beta=1,\ldots,n_0$, are given for
all $l\ge l_0$.
Assume that
they satisfy the following condition:
For
any $\alpha_1,\beta_1,\alpha_2,\beta_2\in\{1,\ldots,n_0\}$
and $l_1,\l_2\ge l_0$,
\begin{align}\label{eq:oi}
x_{\alpha_1,\beta_1}^{(l_1)}\zeij{\alpha_2\beta_2}
+\lambda^{-l_1}\zeij{\alpha_1\beta_1}x_{\alpha_2\beta_2}^{(l_2)}
=\delta_{\beta_1,\alpha_2}x_{\alpha_1,\beta_2}^{(l_1+l_2)}.
\end{align}
Then we have the followings.
\begin{enumerate}
\item[(1)]
If $\lambda\neq 1$, 
there exists $J\in\mnz$
such that
\[
x_{\alpha\beta}^{(l)}
=J\zeij{\alpha\beta}-\lambda^{-l}\zeij{\alpha\beta}J,\quad
l\ge l_0,\quad  \alpha,\beta\in\{1,\ldots,n_0\}.
\]
\item[(2)]
If $\lambda= 1$, 
there exist $J\in\mnz$ and $c\in\cc$
such that
\[
x_{\alpha\beta}^{(l)}
=J\zeij{\alpha\beta}-\zeij{\alpha\beta}J+c\cdot l\cdot \zeij{\alpha\beta}
,\quad
l\ge l_0,\quad  \alpha,\beta\in\{1,\ldots,n_0\}.
\]
\end{enumerate}
\end{lem}
\begin{proof}
We claim 
there exists ${\tilde J}\in\mnz$
such that
\begin{align}\label{eq:claimxi}
x_{\alpha\beta}^{(l)}(1-\zeij{\beta\beta})
=-\lambda^{-l}\zeij{\alpha\beta}{\tilde J},\quad
(1-\zeij{\alpha\alpha})x_{\alpha\beta}^{(l)}
= {\tilde J}\zeij{\alpha\beta},\quad
l\ge l_0,\quad  \alpha,\beta\in\{1,\ldots,n_0\}.
\end{align}
To see this, 
set $F_\beta:=\sum_{\alpha:\alpha\neq\beta}
x_{\alpha\alpha}^{(l_0)}$.
Then by (\ref{eq:oi}) with $\beta_1\neq \alpha_2$, we have for $l\ge l_0$ and $\alpha,\beta\in\{1,\ldots,n_0\}$,
\begin{align}\label{eq:jx}
x_{\alpha\beta}^{(l)}(1-\zeij{\beta\beta})
+\lambda^{-l}\zeij{\alpha\beta}F_\beta=
x_{\alpha\beta}^{(l)}\sum_{\alpha_2:\alpha_2\neq\beta}\zeij{\alpha_2\alpha_2}
+\lambda^{-l}\zeij{\alpha\beta}
\sum_{\alpha_2:\alpha_2\neq\beta}x_{\alpha_2\alpha_2}^{(l_0)}
=0.
\end{align}
We also have for $l\ge l_0$ and $\alpha,\beta\in\{1,\ldots,n_0\}$,
\begin{align}\label{eq:xj}
F_{\alpha}\zeij{\alpha\beta}+
\lambda^{-l_0}\lmk 1-\zeij{\alpha\alpha}\rmk x_{\alpha\beta}^{(l)}
=
\sum_{\alpha_1:\alpha_1\neq\alpha}
x_{\alpha_1\alpha_1}^{(l_0)}\zeij{\alpha\beta}
+\sum_{\alpha_1:\alpha_1\neq\alpha}\lambda^{-l_0}\zeij{\alpha_1\alpha_1}x_{\alpha\beta}^{(l)}=0
\end{align}
from (\ref{eq:oi}).

Set ${\tilde J}:=\sum_{\beta=1}^{n_0} \zeij{\beta\beta} F_{\beta}(1-\zeij{\beta\beta})$
and ${\tilde J}':=-\lambda^{l_0}
\sum_{\alpha=1}^{n_0} (1-\zeij{\alpha\alpha})
F_{\alpha}\zeij{\alpha\alpha} $.
By (\ref{eq:jx}) and (\ref{eq:xj}), we have for $l\ge l_0$ and $\alpha,\beta\in\{1,\ldots,n_0\}$,
\begin{align}\label{eq:jj}
x_{\alpha\beta}^{(l)}(1-\zeij{\beta\beta})
=-\lambda^{-l}\zeij{\alpha\beta}F_\beta
=-\lambda^{-l}\zeij{\alpha\beta}{\tilde J},\quad
(1-\zeij{\alpha\alpha})x_{\alpha\beta}^{(l)}
= -\lambda^{l_0}F_{\alpha}\zeij{\alpha\beta}
=
{\tilde J}'\zeij{\alpha\beta}.
\end{align}
 
To complete the proof of the claim, we show ${\tilde J}={\tilde J}'$.
For any $\alpha_1,\beta_1,\alpha_2,\beta_2\in\{1,\ldots,n_0\}$ with $\beta_1\neq\alpha_2$ and $l_1,l_2\ge l_0$, 
we have from (\ref{eq:jj}) and (\ref{eq:oi}),
\begin{align*}
-\lambda^{-l_1}\zeij{\alpha_1\beta_1}{\tilde J}\zeij{\alpha_2\beta_2}
+\lambda^{-l_1}\zeij{\alpha_1\beta_1}{\tilde J}'\zeij{\alpha_2\beta_2}
=x_{\alpha_1\beta_1}^{(l_1)}\zeij{\alpha_2\beta_2}
+\lambda^{-l_1}\zeij{\alpha_1\beta_1}x_{\alpha_2\beta_2}^{(l_2)}
=0.
\end{align*}
This means the off diagonal elements of ${\tilde J}$ and ${\tilde J}'$ coincide.
As the diagonal elements of ${\tilde J}$ and ${\tilde J}'$
are zero, we obtain ${\tilde J}={\tilde J}'$, proving the claim.

To proceed, we first consider the $\lambda\neq 1$ case.
We fix some $l_0'\ge l_0$ such that $\lambda^{l_0'}\neq 1$.
Set 
\[
C_{\alpha\beta}^{(l)}
:=\braket{\cnz{\alpha}}{x_{\alpha\beta}^{(l)}\cnz{\beta}},\quad l\ge l_0,\quad
\alpha,\beta=1,\ldots,n_0,
\]
and
$d:=(1-\lambda^{-{l_0'}})^{-1}\lambda^{-{l_0'}}C_{11}^{({l_0'})}$.
We claim
\begin{align}\label{eq:claimcd}
C_{\alpha\beta}^{(l)}=
C_{\alpha 1}^{({l_0'})}+d-\lambda^{-l}\lmk C_{\beta 1}^{({l_0'})}+d\rmk
,\quad l\ge {l_0},\quad
\alpha,\beta=1,\ldots,n_0.
\end{align}
For any $\alpha_1,\beta_2,\beta\in\{1,\ldots,n_0\}$ and $l_1,\l_2\ge {l_0}$, 
substituting (\ref{eq:claimxi}), we have
\begin{align*}
&x_{\alpha_1,\beta}^{(l_1)}\zeij{\beta,\beta_2}
+\lambda^{-l_1}\zeij{\alpha_1\beta}x_{\beta\beta_2}^{(l_2)}\\
&=\zeij{\alpha_1\alpha_1}
x_{\alpha_1,\beta}^{(l_1)}\zeij{\beta,\beta_2}
+\lmk 1- \zeij{\alpha_1\alpha_1}\rmk 
x_{\alpha_1,\beta}^{(l_1)}\zeij{\beta,\beta_2}
+\lambda^{-l_1}\zeij{\alpha_1\beta}x_{\beta\beta_2}^{(l_2)}
\zeij{\beta_2\beta_2}
+\lambda^{-l_1}\zeij{\alpha_1\beta}x_{\beta\beta_2}^{(l_2)}
(1-\zeij{\beta_2\beta_2})\\
&=\lmk
C_{\alpha_1\beta}^{(l_1)}+\lambda^{-l_1}C_{\beta\beta_2}^{(l_2)}
\rmk
\zeij{\alpha_1\beta_2}
+\tilde J\zeij{\alpha_1\beta_2}-\lambda^{-l_1-l_2}\zeij{\alpha_1\beta_2}\tilde J.
\end{align*}
Note that the left hand side is $\beta$-independent because of
(\ref{eq:oi}), and the second and the third term on the right hand side is also $\beta$-independent.
Therefore,
for any $\alpha_1,\beta_2\in\{1,\ldots,n_0\}$, and $l_1,\l_2\ge {l_0}$, 
$C_{\alpha_1\beta}^{(l_1)}+\lambda^{-l_1}C_{\beta\beta_2}^{(l_2)}$
is $\beta$-independent.
Hence, for $\alpha,\beta\in\{1,\ldots,n_0\}$, and 
$l\ge {l_0}$,
\begin{align}\label{cab}
C_{\alpha\beta}^{(l)}=-\lambda^{-l}C_{\beta1}^{({l_0'})}
+C_{\alpha 1}^{(l)}
+\lambda^{-l} C_{11}^{({l_0'})}.
\end{align}

Substituting this to $C_{\alpha_1\beta}^{(l_1)}+\lambda^{-l_1}C_{\beta\beta_2}^{(l_2)}$, we see that
for $\alpha_1,\beta_2,\beta=1,\ldots,n_0$
 and $l_1,l_2\ge {l_0}$,
\begin{align*}
C_{\alpha_1\beta}^{(l_1)}+\lambda^{-l_1}C_{\beta\beta_2}^{(l_2)}
=-\lambda^{-l_1}C_{\beta 1}^{({l_0'})}+C_{\alpha_1 1}^{(l_1)}
+\lambda^{-l_1}C_{11}^{({l_0'})}
+\lambda^{-l_1}
\lmk
-\lambda^{-l_2}C_{\beta_2 1}^{({l_0'})}+C_{\beta 1}^{(l_2)}
+\lambda^{-l_2} C_{11}^{({l_0'})}
\rmk.
\end{align*}
Recall that the left hand side is $\beta$-independent.
This means, $W_l:=C_{\beta 1}^{(l)}-C_{\beta 1}^{({l_0'})}$
is $\beta$-independent for $l\ge l_0$.
Note that $W_{{l_0'}}=0$.

Using $W_l$, (\ref{cab}) can be written as
\begin{align}\label{eq:alw}
C_{\alpha\beta}^{(l)}=\lambda^{-l}\lmk C_{11}^{({l_0'})}-C_{\beta 1}^{({l_0'})}\rmk
+W_l+C_{\alpha 1}^{({l_0'})},\quad
l\ge {l_0},\quad \alpha,\beta=1,\ldots,n_0.
\end{align}

Considering $(\alpha_1,\beta_2)$-matrix element of (\ref{eq:oi}),
with $\beta_1=\alpha_2=\beta$,
$l_1=l\ge {l_0}$ and $l_2={l_0}$, we have
$
C_{\alpha_1\beta}^{(l)}+\lambda^{-l}C_{\beta\beta_2}^{({l_0'})}
=C_{\alpha_1\beta_2}^{({l+l_0'})}
$.
Same consideration with 
$l_1= {l_0'}$ and $l_2=l\ge {l_0}$, 
implies
$
C_{\alpha_1\beta}^{(l_0')}+\lambda^{-{l_0'}}C_{\beta\beta_2}^{({l})}
=C_{\alpha_1\beta_2}^{({l+l_0'})}
$.
From these, we obtain
$
C_{\alpha_1\beta}^{(l)}+\lambda^{-l}C_{\beta\beta_2}^{({l_0'})}
=C_{\alpha_1\beta}^{({l_0'})}+\lambda^{-{l_0'}}C_{\beta\beta_2}^{(l)}
$.
Substituting (\ref{eq:alw}) and $W_{{l_0'}}=0$ to this,
we obtain
\[
W_l=(\lambda^{-{l_0'}}-\lambda^{-l})
\lmk 1-\lambda^{-{l_0'}}\rmk^{-1}C_{1,1}^{({l_0'})},\quad l\ge l_0.
\]
Substituting this to (\ref{eq:alw}), we
obtain
(\ref{eq:claimcd}).

Set $J:=\tilde J+\sum_{\alpha}\lmk C_{\alpha 1}^{({l_0'})}+d \rmk \zeij{\alpha\alpha}\in\mnz$.
Then for $l\ge {l_0}$ and $\alpha,\beta\in\{1,\dots,n_0\}$,
using (\ref{eq:claimxi}),
\begin{align*}
&x_{\alpha\beta}^{(l)}
=\lmk 1-\zeij{\alpha\alpha}\rmk
x_{\alpha\beta}^{(l)}
+\zeij{\alpha\alpha}x_{\alpha\beta}^{(l)}\lmk 1-\zeij{\beta\beta}\rmk
+\zeij{\alpha\alpha}x_{\alpha\beta}^{(l)}\zeij{\beta\beta}
= \tilde J \zeij{\alpha\beta}
-\lambda^{-l}\zeij{\alpha\beta}\tilde J
+C_{\alpha\beta}^{(l)}\zeij{\alpha\beta} \\
&=\tilde J \zeij{\alpha\beta}
-\lambda^{-l}\zeij{\alpha\beta}\tilde J
+\lmk
C_{\alpha 1}^{({l_0'})}+d-\lambda^{-l}\lmk C_{\beta 1}^{({l_0'})}+d\rmk
\rmk
\zeij{\alpha\beta} 
=J\zeij{\alpha\beta}- \lambda^{-l}\zeij{\alpha\beta}J.
\end{align*}

Now we turn to the case $\lambda=1$.
Set
\[
C_{\alpha\beta}^{(l)}
:=\braket{\cnz{\alpha}}{x_{\alpha\beta}^{(l)}\cnz{\beta}},\quad l\ge l_0,\quad
\alpha,\beta=1,\ldots,n_0.
\]
By (\ref{eq:oi}) with $\beta_1=\alpha_2=\beta$ we have 
\begin{align}\label{eq:ccc}
C_{\alpha_1\beta}^{(l_1)}+C_{\beta,\beta_2}^{(l_2)}
=C_{\alpha_1,\beta_2}^{(l_1+l_2)},\quad
\alpha_1,\beta_2,\beta\in\{1,\ldots,n_0\},
\quad l_1,\l_2\ge l_0.
\end{align}
Therefore, for any $l\ge l_0$, we have
\[
C_{11}^{(l+2)}-C_{11}^{(l+1)}
=C_{11}^{(l+1)}-C_{11}^{(l)},
\]
and obtain
\begin{align}\label{eq:c11}
C_{11}^{(l)}-C_{11}^{(l_0)}
=\sum_{j=l_0+1}^{l} \lmk C_{11}^{(j)}-C_{11}^{(j-1)}\rmk
=(l-l_0)\lmk
C_{11}^{(l_0+1)}-C_{11}^{(l_0)}
\rmk,\quad
l\ge l_0+1.
\end{align}
Note that we also have 
$
C_{11}^{(l_0)}-C_{11}^{(l_0)}=0
=(l_0-l_0)\lmk
C_{11}^{(l_0+1)}-C_{11}^{(l_0)}
\rmk$.
Set $c':=-l_0\lmk C_{11}^{(l_0+1)}-C_{11}^{(l_0)}\rmk+ C_{11}^{(l_0)}$
and $c=C_{11}^{(l_0+1)}-C_{11}^{(l_0)}$.
Then from (\ref{eq:ccc}) and (\ref{eq:c11}), we have
\begin{align*}
C_{\alpha\beta}^{(l)}
&=C_{11}^{(l_0)}-C_{\beta1}^{(l_0)}+C_{\alpha1}^{(l)}\\
&=C_{11}^{(l_0)}-C_{\beta1}^{(l_0)}+\lmk
C_{\alpha1}^{(l)} -C_{\alpha 1}^{(l_0)}\rmk
+C_{\alpha1}^{(l_0)}\\
&=C_{11}^{(l_0)}-C_{\beta1}^{(l_0)}+\lmk
C_{11}^{(l)} -C_{1 1}^{(l_0)}\rmk
+C_{\alpha1}^{(l_0)}\\
&=C_{11}^{(l_0)}-C_{\beta1}^{(l_0)}+
(l-l_0)\lmk
C_{11}^{(l_0+1)}-C_{11}^{(l_0)}
\rmk
+C_{\alpha1}^{(l_0)}\\
&=C_{\alpha1}^{(l_0)}-C_{\beta1}^{(l_0)}+
c\cdot l+c',
\end{align*}
for all $\alpha,\beta\in\{1,\ldots,n_0\}$, $l\ge l_0$.
Substituting this to (\ref{eq:ccc}), we get $c'=0$.
Hence we obtain
\begin{align}\label{eq:snana}
C_{\alpha\beta}^{(l)}
=C_{\alpha1}^{(l_0)}-C_{\beta1}^{(l_0)}+
c\cdot l,
\quad l\ge l_0,\quad
\alpha,\beta=1,\ldots,n_0.
\end{align}

Set $\tilde I:=\sum_{\alpha=1}^{n_0}C_{\alpha,1}^{(l_0)}\zeij{\alpha\alpha}\in\mnz$.
Recall $\tilde J$ in (\ref{eq:claimxi}), and set $J:=\tilde J+\tilde I$.
Then substituting (\ref{eq:snana}), we have
for $l\ge l_0$, and
$\alpha,\beta=1,\ldots,n_0$,
\begin{align*}
&x_{\alpha\beta}^{(l)}
=\lmk 1-\zeij{\alpha\alpha}\rmk
x_{\alpha\beta}^{(l)}
+\zeij{\alpha\alpha}x_{\alpha\beta}^{(l)}\lmk 1-\zeij{\beta\beta}\rmk
+\zeij{\alpha\alpha}x_{\alpha\beta}^{(l)}\zeij{\beta\beta}
= \tilde J \zeij{\alpha\beta}
-\zeij{\alpha\beta}\tilde J
+C_{\alpha\beta}^{(l)}\zeij{\alpha\beta} \\
&= \tilde J \zeij{\alpha\beta}
-\zeij{\alpha\beta}\tilde J
+\lmk C_{\alpha1}^{(l_0)}-C_{\beta1}^{(l_0)}+
c\cdot l\rmk
\zeij{\alpha\beta} 
=\tilde J \zeij{\alpha\beta}
-\zeij{\alpha\beta}\tilde J
+\tilde I\zeij{\alpha\beta}-\zeij{\alpha\beta}\tilde I
+c\cdot l\cdot \zeij{\alpha\beta}\\
&= J \zeij{\alpha\beta}
-\zeij{\alpha\beta}J
+c\cdot l\cdot
\zeij{\alpha\beta} .
\end{align*}
\end{proof}
\begin{proofof}[Lemma \ref{lem:c6in}]
By {\it Condition 6-i}, each $y_{0,\alpha,\beta}^{(l)}$, $l\ge l_0$, and
$\alpha,\beta=1,\ldots,n_0$, is of the form
\[
y_{0,\alpha,\beta}^{(l)}
=\zeij{\alpha\beta}\otimes\Lambda_\lal^l \lmk \unit+Y\rmk^l
+\sum_{j=i+1}^{k}x_{\alpha,\beta,j}^{(l)}\otimes E_{j,j-(i+1)}^{(0,k)}
+Y_{\alpha,\beta}^{(l)},
\]
with $x_{\alpha,\beta,j}^{(l)}\in\mnz$ and 
$Y_{\alpha,\beta}^{(l)}\in \mnz\otimes \sum_{a,a':a-a'\ge i+2} E_{aa}^{(0,k)}\Mat_{k+1} E_{a'a'}^{(0,k)}$.
Recall that
\[y_{0,\alpha_1,\beta_1}^{(l_1)}y_{0,\alpha_2,\beta_2}^{(l_2)}
=\delta_{\beta_1\alpha_2} y_{0,\alpha_1,\beta_2}^{(l_1+l_2)},\quad
\alpha_1,\alpha_2,\beta_1,\beta_2=1,\ldots,n_0,\quad l_1,l_2\ge l_0,
\]
by {\it Condition 6-i}.
With the above representation, the left hand side of this equation
can be written
\begin{align*}
&y_{0,\alpha_1,\beta_1}^{(l_1)}y_{0,\alpha_2,\beta_2}^{(l_2)}\\
&=\delta_{\beta_1\alpha_2}\zeij{\alpha_1\beta_2}
\otimes\Lambda_{\lal}^{l_1+l_2}\lmk \unit +Y \rmk^{l_1+l_2}
+\sum_{j=i+1}^k
\lmk
\lambda_j^{l_1}\zeij{\alpha_1\beta_1}x_{\alpha_2,\beta_2,j}^{(l_2)}
+\lambda_{j-(i+1)}^{l_2}x_{\alpha_1,\beta_1,j}^{(l_1)}\zeij{\alpha_2\beta_2}
\rmk\otimes
E_{j,j-(i+1)}^{(0,k)}\\
&\quad +
\text{ an element of }
\mnz\otimes \sum_{a,a':a-a'\ge i+2} E_{aa}^{(0,k)}\Mat_{k+1} E_{a'a'}^{(0,k)}.
\end{align*}
Compare the $ \mnz\otimes \sum_{a,a':a-a'= i+1} E_{aa}^{(0,k)}\Mat_{k+1} E_{a'a'}^{(0,k)}$ part of this and
that of the right hand side:
\begin{align}
\delta_{\beta_1\alpha_2} y_{0,\alpha_1,\beta_2}^{(l_1+l_2)}
=\delta_{\beta_1\alpha_2}\lmk
\zeij{\alpha_1\beta_2}\otimes\Lambda_\lal^{l_1+l_2}\lmk \unit+Y\rmk^{l_1+l_2}
+\sum_{j=i+1}^{k}x_{\alpha_1,\beta_2,j}^{(l_1+l_2)}\otimes E_{j,j-(i+1)}^{(0,k)}
+Y_{\alpha_1,\beta_2}^{(l_1+l_2)}
\rmk.
\end{align}
Then we obtain
\begin{align*}
\lambda_j^{l_1}\zeij{\alpha_1\beta_1}x_{\alpha_2,\beta_2,j}^{(l_2)}
+\lambda_{j-(i+1)}^{l_2}x_{\alpha_1,\beta_1,j}^{(l_1)}\zeij{\alpha_2\beta_2}
=\delta_{\beta_1\alpha_2}x_{\alpha_1,\beta_2,j}^{(l_1+l_2)}.
\end{align*}
Set $\widetilde{x}_{\alpha,\beta,j}^{(l)}:=
\lambda_{j-(i+1)}^{-l }x_{\alpha,\beta,j}^{(l)}$.
From the above equality,  $\widetilde{x}_{\alpha,\beta,j}^{(l)}$
satisfies condition of Lemma \ref{lem:key}
with $\lambda=\lambda_{j-(i+1)}/\lambda_j$.
Applying Lemma \ref{lem:key},
we obtain $\{J_j\}_{j={i+1}}^{k}\subset \mnz$
and
$\{c_j\}_{j={i+1}}^{k}\subset \cc$
such that 
\begin{align}\label{eq:xform}
x_{\alpha,\beta,j}^{(l)}=
\lambda_{j-(i+1)}^l J_j\zeij{\alpha\beta}-
\lambda_j^l\zeij{\alpha\beta}J_j
+\delta_{\lambda_j,\lambda_{j-(i+1)}}c_jl\lambda_j^l
\zeij{\alpha\beta},\quad
\alpha,\beta=1,\ldots,n_0,\quad l\ge l_0.
\end{align}
Hence, for any $i+1\le j\le k$,
$\alpha,\beta=1,\ldots,n_0$, and
$l\ge l_0$, we have
\begin{align}\label{eq:yfm}
&\lmk\unit\otimes E_{jj}^{(0,k)}\rmk
\lmk
y^{(l)}_{0,\alpha,\beta}-
\zeij{\alpha\beta}\otimes \Lambda_{\lal}^l\lmk 1+ Y\rmk ^l
\rmk
\lmk\unit\otimes E_{j-(i+1),j-(i+1)}^{(0,k)}\rmk\notag\\
&=
\lmk
\lambda_{j-(i+1)}^l J_j\zeij{\alpha\beta}-
\lambda_j^l\zeij{\alpha\beta}J_j
+\delta_{\lambda_j,\lambda_{j-(i+1)}}c_jl\lambda_j^l
\zeij{\alpha\beta}
\rmk
	\otimes E_{j,j-(i+1)}^{(0,k)}.
\end{align}

Now we check {\it 1,2} of the Lemma.
We start from {\it 1}.
Set $\hat J:=\sum_{j=i+1}^{k} J_j\otimes E_{j,j-(i+1)}^{(0,k)}$, and
$R:=\unit-\hat J$.
Clearly, $\hat J\in\mnz\otimes\DT_{0,k+1}$.
As $\hat J^{k+1}=0$, $R$ is invertible and
$R^{-1}=\unit+\hat J+\cdots+\hat J^{k}$.
Consider (\ref{eq:yfm}) with $j=i+1$.
Then by {\it Condition 5} (iv) (3) for $\{y^{(l)}_{a,\alpha,\beta}\}$,
we get 
\begin{align*}
0=\lmk\unit\otimes E_{{i+1},{i+1}}^{(0,k)}\rmk
 \lmk
y_{0,\alpha,\beta}^{(l)}-\zeij{\alpha\beta}\otimes \Lambda_{\lal}^l\lmk 1+ Y\rmk ^l
\rmk
\lmk\unit\otimes E_{00}^{(0,k)}\rmk
=\lmk
 J_{i+1}\zeij{\alpha\beta}-
\lambda_{i+1}^l\zeij{\alpha\beta}J_{i+1}
\rmk
	\otimes E_{i+1,0}^{(0,k)},
\end{align*}
for all $\alpha,\beta=1,\ldots,n_0$, and  $l\ge l_0$.
This implies $J_{i+1}=0$.
Therefore, we have
$\hat J\lmk \unit\otimes E_{00}^{(0,k)}\rmk =J_{i+1}\otimes E_{i+1,0}^{(0,k)}=0$,
which implies
\begin{align}\label{eq:re}
R\lmk \unit\otimes E_{00}^{(0,k)}\rmk
=R^{-1}\lmk \unit\otimes E_{00}^{(0,k)}\rmk
= \unit\otimes E_{00}^{(0,k)}.
\end{align}

To prove {\it 2},
we first show that the septuplet $(n_0,k,\oo,R\vv R^{-1},\lal,l_0,\{R y_{a,\alpha,\beta}^{(l)}R^{-1}\})$
satisfies {\it Condition 5}.
(i) follows from the assumption that {\it Condition 6-i} holds.
As all of $v_{\mu}, R,R^{-1}$ belongs to $\mnz\otimes \DT_{k+1}$, 
we have $R v_\mu R^{-1}\in \mnz\otimes \DT_{k+1}$, proving (ii).
(iii) can be checked as follows:
\begin{align*}
&\lmk \unit\otimes E_{aa}^{(0,k)}\rmk Rv_\mu R^{-1}\lmk \unit\otimes E_{aa}^{(0,k)}\rmk
=\lmk \unit\otimes E_{aa}^{(0,k)}\rmk R\lmk \unit\otimes E_{aa}^{(0,k)}\rmk
v_\mu\lmk \unit\otimes E_{aa}^{(0,k)}\rmk R^{-1}\lmk \unit\otimes E_{aa}^{(0,k)}\rmk\\
&=\lmk \unit\otimes E_{aa}^{(0,k)}\rmk
v_\mu
\lmk \unit\otimes E_{aa}^{(0,k)}\rmk=\lambda_a \omega_\mu\otimes  E_{aa}^{(0,k)}
\end{align*} 
for all $a=0,\ldots,k$ and $\mu=1,\ldots,n$.
As $v_{\mu}, R,R^{-1}\in \mnz\otimes \DT_{k+1}$,
only "diagonal parts" of $R,v_\mu, R^{-1}$ are left when $Rv_\mu R^{-1}$ is sandwiched by $\unit\otimes E_{aa}^{(0,k)}$.
This corresponds to the first equality.
The second equality is due to $R-\unit, R^{-1}-\unit\in\mnz\otimes\DT_{0,k+1}$.
The last equality is because of {\it Condition 5} (iii) for $\vv,\oo$.
(1), (2) of (iv) follows from the fact that 
$y_{a,\alpha,\beta}^{(l)}$ and
$Ry_{a,\alpha,\beta}^{(l)}R^{-1}$ are similar to each other,
and $v_\mu$ and $Rv_\mu R^{-1}$ are similar to each other with the common 
invertible operator $R$.
(3) of (iv) follows from the following calculation:
\begin{align*}
&Ry_{0,\alpha,\beta}^{(l)}R^{-1}\lmk \unit\otimes E_{00}^{(0,k)}\rmk=Ry_{0,\alpha,\beta}^{(l)}\lmk \unit\otimes E_{00}^{(0,k)}\rmk
=R\lmk \unit\otimes E_{00}^{(0,k)}\rmk y_{0,\alpha,\beta}^{(l)}\lmk \unit\otimes E_{00}^{(0,k)}\rmk\\
&=\lmk \unit\otimes E_{00}^{(0,k)}\rmk y_{0,\alpha,\beta}^{(l)}\lmk \unit\otimes E_{00}^{(0,k)}\rmk
=\lmk \unit\otimes E_{00}^{(0,k)}\rmk R y_{0,\alpha,\beta}^{(l)} R^{-1}\lmk \unit\otimes E_{00}^{(0,k)}\rmk.
\end{align*}
The first and third equality follows from (\ref{eq:re}). The second equality is from {\it Condition 5} (iv)(3) of 
$y_{0,\alpha,\beta}^{(l)}$.
The last equality is from (\ref{eq:re}) and $R-\unit\in \mnz\otimes\DT_{0,k+1}$.

Set
$Y':=\sum_{i+1\le j\le k:\lambda_j=\lambda_{j-(i+1)}} c_j E_{j,j-(i+1)}^{(0,k)}$.
We show (ii), (iii) of {\it Conditions 6-(i+1)} for $(n_0,k,\oo,R\vv R^{-1},\lal,l_0,\{R y_{a,\alpha,\beta}^{(l)}R^{-1}\})$
with respect to $Y+Y'$.
(ii) is clear from the definition.
To see (iii), note that 
\begin{align}\label{eq:YY}
&\lmk \unit +Y+Y'\rmk^l
=\lmk \unit +Y\rmk^l
+\sum_{k=0}^{l-1} \lmk 1+Y\rmk^kY'\lmk \unit+Y\rmk^{l-(k+1)}
+\text{ terms with more than one }Y'\nonumber \\
&=\lmk \unit +Y\rmk^l
+l Y'
+\text{ elements in } \sum_{a,a':a-a'\ge i+2} E_{aa}^{(0,k)}\Mat_{k+1}E_{a'a'}^{(0,k)}.
\end{align}
We used the fact that terms with more than one $Y'$
belong to
$ \sum_{a,a':a-a'\ge i+2} E_{aa}^{(0,k)}\Mat_{k+1}E_{a'a'}^{(0,k)}$
because
$Y'\in \sum_{a,a':a-a'= i+1} E_{aa}^{(0,k)}\Mat_{k+1}E_{a'a'}^{(0,k)}$.
Similarly, the terms with one $Y'$ and more than zero $Y\in  \sum_{a,a':a-a'\ge 1} E_{aa}^{(0,k)}\Mat_{k+1}E_{a'a'}^{(0,k)}$ belong to
$\sum_{a,a':a-a'\ge i+2} E_{aa}^{(0,k)}\Mat_{k+1}E_{a'a'}^{(0,k)}$.
By (\ref{eq:xform}) and (\ref{eq:YY}), we have
\begin{align*}
& \widetilde{Y_{\alpha,\beta}^{(l)}}
:=y_{0,\alpha,\beta}^{(l)}
-\zeij{\alpha\beta}\otimes\Lambda_\lal^l \lmk \unit+Y\rmk^l
-\hat J\lmk \zeij{\alpha\beta}\otimes\Lambda_\lal^l \lmk \unit+Y\rmk^l\rmk
+\lmk \zeij{\alpha\beta}\otimes\Lambda_\lal^l \lmk \unit+Y\rmk^l\rmk\hat J
-l\zeij{\alpha\beta}\otimes Y'\Lambda_\lal^l\\
&=
y_{0,\alpha,\beta}^{(l)}
-\zeij{\alpha\beta}\otimes\Lambda_\lal^l \lmk \unit+Y\rmk^l
-\sum_{j=i+1}^{k}
\lmk \lambda_{j-(i+1)}^l J_j\zeij{\alpha\beta}-
\lambda_j^l\zeij{\alpha\beta}J_j
+\delta_{\lambda_j,\lambda_{j-(i+1)}}c_jl\lambda_j^l
\zeij{\alpha\beta}
\rmk
\otimes E_{j,j-(i+1)}^{(0,k)}\\
&-\hat J\lmk \zeij{\alpha\beta}\otimes\Lambda_\lal^l\lmk  \lmk \unit+Y\rmk^l-\unit\rmk\rmk
+\lmk \zeij{\alpha\beta}\otimes\Lambda_\lal^l \lmk \lmk \unit+Y\rmk^l-\unit\rmk\rmk \hat J\\
&=Y_{\alpha,\beta}^{(l)}-\hat J\lmk \zeij{\alpha\beta}\otimes\Lambda_\lal^l\lmk  \lmk \unit+Y\rmk^l-\unit\rmk\rmk+\lmk \zeij{\alpha\beta}\otimes\Lambda_\lal^l \lmk \lmk \unit+Y\rmk^l-\unit\rmk\rmk \hat J\\
&\in \mnz\otimes\sum_{a,a':a-a'\ge i+2} E_{aa}^{(0,k)}\Mat_{k+1}E_{a'a'}^{(0,k)},
\end{align*}
for $\alpha,\beta=1,\ldots, n_0$ and $l\ge l_0$.
The last inclusion is because of 
$\hat J\in \mnz\otimes\sum_{a,a':a-a'= i+1} E_{aa}^{(0,k)}\Mat_{k+1}E_{a'a'}^{(0,k)}$
and $Y\in  \sum_{a,a':a-a'\ge 1} E_{aa}^{(0,k)}\Mat_{k+1}E_{a'a'}^{(0,k)}$.
Recall that
$R:=\unit-\hat J$
and 
$R^{-1}=\unit+\hat J+\cdots+\hat J^{k}$.
We have
\begin{align}\label{eq:RR}
&Ry_{0,\alpha,\beta}^{(l)}R^{-1}-\zeij{\alpha\beta}
\otimes\Lambda_\lal^l \lmk \unit+Y+Y'\rmk^l\notag\\
&=
\lmk \unit-\hat J\rmk\lmk
\zeij{\alpha\beta}\otimes\Lambda_\lal^l \lmk \unit+Y\rmk^l
+\hat J\lmk \zeij{\alpha\beta}\otimes\Lambda_\lal^l \lmk \unit+Y\rmk^l\rmk
-
\zeij{\alpha\beta}\otimes\lmk \Lambda_\lal^l 
\lmk \unit+Y\rmk^l\rmk\hat J
+l\zeij{\alpha\beta}\otimes Y'\Lambda_\lal^l
\rmk
\lmk
\unit+\hat J+\cdots+\hat J^{k}
\rmk\notag\\
&\quad +\lmk \unit-\hat J\rmk\
\widetilde{Y_{\alpha,\beta}^{(l)}}
\lmk
\unit+\hat J+\cdots+\hat J^{k}
\rmk-\zeij{\alpha\beta}
\otimes\Lambda_\lal^l \lmk \unit+Y+Y'\rmk^l\notag\\
&=
\text{ terms with more than one } \hat J+\text{ terms with more than zero }\hat J \text{ and one } Y'\notag\\
&\quad +\lmk \unit-\hat J\rmk
\widetilde{Y_{\alpha,\beta}^{(l)}}
\lmk
\unit+\hat J
\rmk-\zeij{\alpha\beta}
\otimes\Lambda_\lal^l\lmk  \lmk \unit+Y+Y'\rmk^l
- \lmk \unit+Y\rmk^l-l Y'
\rmk,
\end{align}
for $\alpha,\beta=1,\ldots, n_0$, and $l\ge l_0$.

Note that
$\hat J\in \mnz\otimes\sum_{a,a':a-a'= i+1} E_{aa}^{(0,k)}\Mat_{k+1}E_{a'a'}^{(0,k)}$,
$Y\in  \sum_{a,a':a-a'\ge 1} E_{aa}^{(0,k)}\Mat_{k+1}E_{a'a'}^{(0,k)}$,
$Y'\in \sum_{a,a':a-a'= i+1} E_{aa}^{(0,k)}\Mat_{k+1}E_{a'a'}^{(0,k)}$,
and $\widetilde{Y_{\alpha,\beta}^{(l)}}\in  \mnz\otimes\sum_{a,a':a-a'\ge i+2} E_{aa}^{(0,k)}\Mat_{k+1}E_{a'a'}^{(0,k)}$.
This observation and (\ref{eq:YY}) implies that the right hand side of (\ref{eq:RR}) is in
$ \mnz\otimes\sum_{a,a':a-a'\ge i+2} E_{aa}^{(0,k)}\Mat_{k+1}E_{a'a'}^{(0,k)}$.
This proves (iii).
\end{proofof}

\begin{lem}\label{lem:hy}
Let $(n_0,k,\oo,\vv,\lal,l_0,\{y_{a,\alpha,\beta}^{(l)}\})$ be a the septuplet 
satisfying {\it Condition 6-0}. Suppose that there exist strictly positive operators 
$h_a$ in $\mnz$, $a=0,\ldots, k$ with $h_0=\unit$ such that
\[
y_{a,\alpha,\beta}^{(l)}\lmk \unit\otimes E_{00}^{(0,k)}\rmk
=h_a^{\frac 12}\zeij{\alpha\beta}\otimes E_{a0}^{(0,k)},\quad
l\ge l_0, \quad a=0,\ldots, k,\quad \alpha,\beta=1,\ldots,n_0.
\]
Then there exist $R\in\mnz\otimes\Mat_{k+1}$, $Y\in \DT_{0,k+1}$,
and $\{ \hat y_{a,\alpha,\beta}^{(l)}\}_{a=0,\ldots,k, \alpha,\beta=1,\ldots,n_0,l\ge l_0}$
satisfying the followings.
\begin{enumerate}
\item
We have $[\Lambda_{\lal},Y]=0$.
\item  For any $\alpha,\beta=1,\ldots,n_0$ , and $l\ge l_0$, we have
$\hat y_{0,\alpha,\beta}^{(l)}=\zeij{\alpha\beta}\otimes \Lambda_{\lal}^l\lmk \unit+Y\rmk^l$.
\item The septuplet 
$(n_0,k,\oo,R\vv R^{-1},\lal,l_0,\{\hat y_{a,\alpha,\beta}^{(l)}\})$
satisfies {\it Condition 5}.
\item  For any $a=0,\ldots,k$,
$\alpha,\beta=1,\ldots,n_0$ , and $l\ge l_0$, we have
\[
\hat y_{a,\alpha,\beta}^{(l)}\lmk \unit\otimes E_{00}^{(0,k)}\rmk
=\zeij{\alpha,\beta}\otimes E_{a,0}^{(0,k)}.
\]
\item
If $X\in \caK_l(R\vv R^{-1})$, $l\ge l_0$,
satisfies $X\lmk \unit\otimes E_{00}^{(0,k)}\rmk=0$, then $X=0$.
\item
Set $\tilde \Lambda:=\Lambda_{\lal}\lmk \unit+Y\rmk$ and
 $\hat \Lambda:=\unit\otimes \tilde \Lambda$.
Then we have
\begin{align}\label{eq:ly}
&\hat\Lambda^l \hat y_{a,\alpha,\beta}^{(l_1)}
=\sum_{a'=1}^{k}\braket{f_{a'}^{(0,k)}}{\tilde\Lambda^{l}f_a^{(0,k)}}\hat y_{a',\alpha,\beta}^{(l_1)}\hat\Lambda^l, \notag\\
&\hat y_{a,\alpha,\beta}^{(l_1)}\hat\Lambda^l 
=\sum_{a'=1}^{k}\braket{f_{a'}^{(0,k)}}{\tilde\Lambda^{-l}f_a^{(0,k)}}\hat\Lambda^l\hat y_{a',\alpha,\beta}^{(l_1)} ,
\end{align}
for all $l\in\nan$,  $l_1\ge l_0$, $a=1,\ldots,k$, and $\alpha,\beta=1,\ldots,n_0$.
\item For all $l\ge l_0$, $a=1,\ldots,k$, and $\alpha,\beta=1,\ldots,n_0$, we have
\begin{align*}
\hat y_{a,\alpha,\beta}^{(l)}
\in \mnz\otimes \sum_{\stackrel{a',a^{''}=0,\ldots,k}{\lambda_{a'}=\lambda_{a''}\lambda_a}}\eijn{a'a'}\mk\eijn{a^{''}a^{''}}.
\end{align*}

\end{enumerate}

\end{lem}
\begin{proof}
By Lemma \ref{lem:c6in} and the assumption that {\it Condition 6-0} holds,
we obtain invertible matrices $R_i\in\mnz\otimes\Mat_{k+1}$, $i=1,\ldots,k$,
and lower triangular matrices $Y_i\in\DT_{0,k+1} $, $i=1,\ldots,k$.
They satisfy the following properties.
\begin{enumerate}
\item $R_i-\unit\in\mnz\otimes \DT_{0,k+1}$,
$R_i\lmk \unit\otimes E_{00}^{(0,k)}\rmk
=R_i^{-1}\lmk \unit\otimes E_{00}^{(0,k)}\rmk
= \unit\otimes E_{00}^{(0,k)}$.
\item
The septuplet $(n_0,k,\oo,\lmk R_i\cdots R_1\rmk \vv \lmk R_i\cdots R_1\rmk^{-1},\lal,l_0,
\{ \lmk R_i\cdots R_1\rmk y_{a,\alpha,\beta}^{(l)}\lmk R_i\cdots R_1\rmk^{-1}\})$
satisfies {\it Condition 6-i} with respect to $Y_1+\cdots Y_i$.
\end{enumerate}
Set $R:= R_{k}\cdots R_1$ and $Y:=Y_1+\cdots+Y_k$.
Then $R$ is invertible, $R-\unit\in\mnz\otimes \DT_{0,k+1}$
and
$R\lmk \unit\otimes E_{00}^{(0,k)}\rmk
=R^{-1}\lmk \unit\otimes E_{00}^{(0,k)}\rmk
= \unit\otimes E_{00}^{(0,k)}$.
The matrix $Y$ belongs to $\DT_{0,k+1}$ and  
the septuplet $(n_0,k,\oo,R\vv R^{-1},\lal,l_0,
\{  Ry_{a,\alpha,\beta}^{(l)}R^{-1}\})$
satisfies {\it Condition 6-k} with respect to $Y$.

Now we would like to change the basis $\{  Ry_{a,\alpha,\beta}^{(l)}R^{-1}\}$
so that the new basis satisfies {\it 4} of the current Lemma.
We set $\tilde R:=\sum_{a=0}^{k}h_a^{\frac 12}\otimes \eijn{aa}$ and define
\begin{align*}
\hat y_{a,\alpha,\beta}^{(l)}
= \sum_{i=1}^{k}\sum_{\alpha'=1}^{n_0}\braket{\cnz{\alpha'}\otimes \fiiz{i}}{\tilde R^{-1}R^{-1}\lmk \cnz{\alpha}\otimes \fiiz{a}\rmk } Ry_{i,\alpha',\beta}^{(l)}R^{-1},\quad
1\le a\le k,\quad \alpha,\beta=1,\ldots,n_0,\quad l\ge l_0.
\end{align*}
Furthermore, we set
$
\hat y_{0,\alpha,\beta}^{(l)}
:= Ry_{0,\alpha,\beta}^{(l)}R^{-1}$,
$ \alpha,\beta=1,\ldots,n_0$, $l\ge l_0$.

This defines the change of basis of $\caK_l(R\vv R^{-1})$ for $l\ge l_0$.
Note that
\begin{align*}
R y_{a,\alpha,\beta}^{(l)} R^{-1}
= \sum_{i=1}^{k}\sum_{\alpha'=1}^{n_0}\braket{\cnz{\alpha'}\otimes \fiiz{i}}{R\tilde R\lmk \cnz{\alpha}\otimes \fiiz{a}\rmk } 
\hat y_{i,\alpha',\beta}^{(l)},\quad
1\le a\le k,\quad \alpha,\beta=1,\ldots,n_0,\quad l\ge l_0.
\end{align*}

Let us check the properties {\it 1,2,4,5} of Lemma \ref{lem:hy} for
$R\in\mnz\otimes\Mat_{k+1}$, $Y\in \DT_{0,k+1}$,
and $\{ \hat y_{a,\alpha,\beta}^{(l)}\}_{a=0,\ldots,k, \alpha,\beta=1,\ldots,n_0,l\ge l_0}$.
As 
the septuplet $(n_0,k,\oo,R\vv R^{-1},\lal,l_0,
\{  Ry_{a,\alpha,\beta}^{(l)}R^{-1}\})$
satisfies {\it Condition 6-k} with respect to $Y$,
 {\it 1} of Lemma \ref{lem:hy} holds.
The third condition of   {\it Condition 6-k} 
implies
\[
\hat y_{0,\alpha,\beta}^{(l)}= Ry_{0,\alpha,\beta}^{(l)}R^{-1}
=\zeij{\alpha\beta}\otimes\Lambda_\lal^l \lmk \unit +Y\rmk^l,\quad \alpha,\beta=1,\ldots,n_0,\quad l\ge l_0,
\]
proving {\it 2} of Lemma \ref{lem:hy}.

Next we prove {\it 4} of Lemma \ref{lem:hy}. 
For any $1\le a\le k$,
$\alpha,\beta=1,\ldots,n_0$ , and $l\ge l_0$, we have
\begin{align*}
&\hat y_{a,\alpha,\beta}^{(l)}\lmk \unit\otimes E_{00}^{(0,k)}\rmk
=\sum_{i=1}^{k}\sum_{\alpha'=1}^{n_0}\braket{\cnz{\alpha'}\otimes \fiiz{i}}{\tilde R^{-1}R^{-1}\lmk \cnz{\alpha}\otimes \fiiz{a}\rmk } Ry_{i,\alpha',\beta}^{(l)}R^{-1}\lmk \unit\otimes E_{00}^{(0,k)}\rmk\\
&=\sum_{i=1}^{k}\sum_{\alpha'=1}^{n_0}\braket{\cnz{\alpha'}\otimes \fiiz{i}}{\tilde R^{-1}R^{-1}\lmk \cnz{\alpha}\otimes \fiiz{a}\rmk } Ry_{i,\alpha',\beta}^{(l)}\lmk \unit\otimes E_{00}^{(0,k)}\rmk\\
&=\sum_{i=1}^{k}\sum_{\alpha'=1}^{n_0}\braket{\cnz{\alpha'}\otimes \fiiz{i}}{\tilde R^{-1}R^{-1}\lmk \cnz{\alpha}\otimes \fiiz{a}\rmk } R
\lmk h_i^{\frac 12} \zeij{\alpha'\beta}\otimes E_{i0}^{(0,k)}\rmk
\\
&=\sum_{i=1}^{k}\sum_{\alpha'=1}^{n_0}\sum_{i'=1}^{k}\sum_{\alpha^{''}=1}^{n_0}
\braket{\cnz{\alpha'}\otimes \fiiz{i}}{\tilde R^{-1}R^{-1}\lmk \cnz{\alpha}\otimes \fiiz{a}\rmk }
\braket{\cnz{\alpha^{''}}\otimes \fiiz{i'}}{R\tilde R\lmk \cnz{\alpha'}\otimes \fiiz{i}\rmk}
\lmk \zeij{\alpha^{''}\beta}\otimes E_{i'0}^{(0,k)}\rmk
\\
&=\sum_{i=0}^{k}\sum_{\alpha'=1}^{n_0}\sum_{i'=1}^{k}\sum_{\alpha^{''}=1}^{n_0}
\braket{\cnz{\alpha'}\otimes \fiiz{i}}{\tilde R^{-1}R^{-1}\lmk \cnz{\alpha}\otimes \fiiz{a}\rmk }
\braket{\cnz{\alpha^{''}}\otimes \fiiz{i'}}{R\tilde R\lmk \cnz{\alpha'}\otimes \fiiz{i}\rmk}
\lmk \zeij{\alpha^{''}\beta}\otimes E_{i'0}^{(0,k)}\rmk\\
&= \zeij{\alpha\beta}\otimes E_{a0}^{(0,k)}.
\end{align*}
The second equality is from $R^{-1}\lmk \unit\otimes E_{00}^{(0,k)}\rmk
= \unit\otimes E_{00}^{(0,k)}$, while the third one is by the assumption.
For the fourth and fifth equality, we used the fact that $R\tilde R,\tilde R^{-1}R^{-1}\in\mnz\otimes\DT_{k+1}$.
We also have
\begin{align}\label{eq:yzz}
\hat y_{0,\alpha,\beta}^{(l)}\lmk \unit\otimes E_{00}^{(0,k)}\rmk
=\zeij{\alpha\beta}\otimes\Lambda_\lal^l\lmk \unit + Y\rmk^l E_{00}^{(0,k)}
=\zeij{\alpha\beta}\otimes E_{00}^{(0,k)}
,\quad \alpha,\beta=1,\ldots,n_0,\quad l\ge l_0,
\end{align}
because $Y\in\DT_{0,k+1}$ and $\lambda_0=1$.

The proof of {\it 5} is the same as the proof  of Lemma \ref{lem:yaab} (4), using {\it 4}
and the fact that $\{\hat y_{a,\alpha,\beta}^{(l)}\}$
is a basis of $\caK_l(R\vv R^{-1})$.

Now we prove {\it 3}, of Lemma \ref{lem:hy}, i.e., that
the septuplet 
$(n_0,k,\oo,R\vv R^{-1},\lal,l_0,\{\hat y_{a,\alpha,\beta}^{(l)}\})$
satisfies {\it Condition 5}.
As the septuplet $(n_0,k,\oo,R\vv R^{-1},\lal,l_0,
\{  Ry_{a,\alpha,\beta}^{(l)}R^{-1}\})$ satisfies {\it Condition 5},
(i), (ii), (iii) of {\it Condition 5} for $(n_0,k,\oo,R\vv R^{-1},\lal,l_0,\{\hat y_{a,\alpha,\beta}^{(l)}\})$ hold.
(1) of (iv) is already shown, and (3) of (iv) is clear from (\ref{eq:yzz}).
The proof of (iv)(2)  is the same as the proof  of Lemma \ref{lem:yaab} (2), using {\it 4}, {\it 5}.

Next we prove the first line of (\ref{eq:ly}) for $l\ge l_0$.
We extend this to all $l\in\nan$ after that.
First, note that $\hat\Lambda^l=\sum_{\alpha=1}^{n_0}\hat y_{0,\alpha,\alpha}^{(l)}\in\caK_l\lmk
R\vv R^{-1}\rmk$ for any $l\ge l_0$.
Therefore, for any  $l,l_1\ge l_0$, $a=1,\ldots,k$, and $\alpha,\beta=1,\ldots,n_0$,
we have
\begin{align*}
\hat\Lambda^l \hat y_{a,\alpha,\beta}^{(l_1)}
-\sum_{a'=1}^{k}\braket{f_{a'}^{(0,k)}}{\tilde\Lambda^{l}f_a^{(0,k)}}\hat y_{a',\alpha,\beta}^{(l_1)}\hat\Lambda^l 
\in\caK_{l+l_1}\lmk R\vv R^{-1}\rmk.
\end{align*}
On the other hand, we have
\begin{align*}
&\lmk
\hat\Lambda^l \hat y_{a,\alpha,\beta}^{(l_1)}
-\sum_{a'=1}^{k}\braket{f_{a'}^{(0,k)}}{\tilde\Lambda^{l}f_a^{(0,k)}}\hat y_{a',\alpha,\beta}^{(l_1)}\hat\Lambda^l 
\rmk\lmk\unit\otimes\eijn{00}\rmk\\
&=\lmk
\hat\Lambda^l 
\rmk\lmk\zeij{\alpha\beta}\otimes\eijn{a0}\rmk
-
\lmk
\sum_{a'=1}^{k}\braket{f_{a'}^{(0,k)}}{\tilde\Lambda^{l}f_a^{(0,k)}}
\lmk\zeij{\alpha\beta}\otimes\eijn{a'0}\rmk\rmk=0.
\end{align*}
Therefore, from {\it 5}, we obtain the first line of (\ref{eq:ly}) with $l\ge l_0$. 

Lastly, we prove {\it 6}, {\it 7}.
As we have just proven, the first line of  (\ref{eq:ly}) holds for $l\ge l_0$,
Therefore, for all
 $l\ge \max \{l_0, 2k\}$, $l_1\ge l_0$, $a=1,\ldots,k$, 
$b,b'=0,\ldots, k$, and $\alpha,\beta=1,\ldots,n_0$, we have
\begin{align*}
&\sum_{j=0}^k{}_lC_{j}\cdot
\lambda_b^l\lmk \unit\otimes \eijn{bb}\rmk
 \lmk \unit\otimes Y^j \rmk\hat y_{a,\alpha,\beta}^{(l_1)}
\lmk \unit\otimes \eijn{b^{'}b^{'}}\rmk
= 
\lmk \unit\otimes \eijn{bb}\rmk
\hat\Lambda^l \hat y_{a,\alpha,\beta}^{(l_1)}
\lmk \unit\otimes \eijn{b^{'}b^{'}}\rmk\\
&=\lmk \unit\otimes \eijn{bb}\rmk
\lmk
\sum_{a'=1}^{k}\braket{f_{a'}^{(0,k)}}{\tilde\Lambda^{l}f_a^{(0,k)}}\hat y_{a',\alpha,\beta}^{(l_1)}\hat\Lambda^l 
\rmk
\lmk \unit\otimes \eijn{b^{'}b^{'}}\rmk\\
&=\lambda_a^l\lambda_{b'}^l\sum_{j_1=0}^k{}_lC_{j_1}\sum_{j_2=0}^k{}_lC_{j_2}
\lmk \unit\otimes \eijn{bb}\rmk
\lmk
\sum_{a'=1}^{k}\braket{f_{a'}^{(0,k)}}{Y^{j_1}f_a^{(0,k)}}\hat y_{a',\alpha,\beta}^{(l_1)} \lmk \unit\otimes Y^{j_2}\rmk 
\rmk
\lmk \unit\otimes \eijn{b^{'}b^{'}}\rmk\\
&=\lambda_a^l\lambda_{b'}^l\sum_{j_1=0}^k\sum_{j_2=0}^k
\sum_{i=0}^{j_1+j_2}
\alpha_{(j_1,j_2)}^{(i)}\cdot {}_lC_i
\lmk \unit\otimes \eijn{bb}\rmk
\lmk
\sum_{a'=1}^{k}\braket{f_{a'}^{(0,k)}}{Y^{j_1}f_a^{(0,k)}}\hat y_{a',\alpha,\beta}^{(l_1)}\lmk\unit\otimes  Y^{j_2}\rmk 
\rmk
\lmk \unit\otimes \eijn{b^{'}b^{'}}\rmk.
\end{align*}
In the last line we used Lemma \ref{lem:ccac}.
Let $\{s_i\}_{i=1}^{m_1}$ be the set of distinct elements
in $\{\lambda_j\}_{j=0}^{k}\cup \{\lambda_j\lambda_{j'}\}_{j,j'=1}^{k}$.
Applying Lemma C.7 of Part I, with
$(l,k,m)$ replaced by $(\max\{l_0, 2k+1\}, 2k+1, m_1)$, and $\{s_i\}_{i=1}^{m_1}$, we obtain
\begin{align*}
&\lmk \unit\otimes \eijn{bb}\rmk
\lmk \unit\otimes Y^j \rmk\hat y_{a,\alpha,\beta}^{(l_1)}
\lmk \unit\otimes \eijn{b^{'}b^{'}}\rmk\\
&=\begin{cases}
\sum_{\stackrel{j_1,j_2=0,\ldots,k:}{j\le j_1+j_2}}
\quad \alpha_{(j_1,j_2)}^{(j)}
\lmk \unit\otimes \eijn{bb}\rmk
\lmk
\sum_{a'=1}^{k}\braket{f_{a'}^{(0,k)}}{Y^{j_1}f_a^{(0,k)}}\hat y_{a',\alpha,\beta}^{(l_1)} \lmk\unit\otimes Y^{j_2}\rmk 
\rmk
\lmk \unit\otimes \eijn{b^{'}b^{'}}\rmk &\text{ if }\lambda_b=\lambda_a\lambda_{b'},\\
0&\text{ if }\lambda_b\neq \lambda_a\lambda_{b'}
\end{cases},
\end{align*}
for $j=0,\ldots, k$, $l_1\ge l_0$, $a=1,\ldots,k$, 
$b,b'=0,\ldots, k$, and $\alpha,\beta=1,\ldots,n_0$.
The case with $j=0$ proves {\it 7}.

Similarly, we have
\begin{align*}
0=
\sum_{\stackrel{j_1,j_2=0,\ldots,k:}{j\le j_1+j_2}}
\quad \alpha_{(j_1,j_2)}^{(j)}
\lmk \unit\otimes \eijn{bb}\rmk
\lmk
\sum_{a'=1}^{k}\braket{f_{a'}^{(0,k)}}{Y^{j_1}f_a^{(0,k)}}\hat y_{a',\alpha,\beta}^{(l_1)}\lmk \unit\otimes  Y^{j_2}\rmk 
\rmk
\lmk \unit\otimes \eijn{b^{'}b^{'}}\rmk,
\end{align*}
for $k< j\le 2k$, if $\lambda_b=\lambda_a\lambda_{b'}$,
 $a=1,\ldots,k$, 
$b,b'=0,\ldots, k$,  $l_1\ge l_0$, and $\alpha,\beta=1,\ldots,n_0$.

Let us consider the case with $\lambda_b= \lambda_a\lambda_{b'}$.
We claim
\begin{align}\label{eq:eachj}
&\lmk \unit\otimes \eijn{bb}\rmk
\lmk
\sum_{a'=1}^{k}\hat y_{a',\alpha,\beta}^{(l_1)}
\lmk
\unit\otimes
\braket{f_{a'}^{(0,k)}}
{\lmk Y\otimes \unit_{\Mat_{k+1}} +\lmk\unit_{\Mat_{k+1}}+ Y\rmk\otimes Y \rmk^j f_a^{(0,k)}}
\rmk
\rmk
\lmk \unit\otimes \eijn{b^{'}b^{'}}\rmk\notag\\
&=
\begin{cases}
\lmk \unit\otimes \eijn{bb}\rmk
 \lmk\unit\otimes Y^j \rmk\hat y_{a,\alpha,\beta}^{(l_1)}
\lmk \unit\otimes \eijn{b^{'}b^{'}}\rmk&\text{ if }0\le j\le k\\
0&\text{ if }k<j\le 2k
\end{cases},
\end{align}
for $a=1,\ldots,k$, 
$b,b'=0,\ldots, k$,  $l_1\ge l_0$, and $\alpha,\beta=1,\ldots,n_0$.
Here,  $Y\otimes \unit_{\Mat_{k+1}} +\lmk\unit_{\Mat_{k+1}}+ Y\rmk\otimes Y
\in\Mat_{k+1}\otimes\Mat_{k+1}$, and
$\braket{f_{a'}^{(0,k)}}
{\lmk Y\otimes \unit_{\Mat_{k+1}} +\lmk\unit_{\Mat_{k+1}}+ Y\rmk\otimes Y \rmk^j f_a^{(0,k)}}$ denotes a matrix in $\mk$
such that
\[
\braket{\xi}{
\braket{f_{a'}^{(0,k)}}
{\lmk Y\otimes \unit_{\Mat_{k+1}} +\lmk\unit_{\Mat_{k+1}}+ Y\rmk\otimes Y \rmk^j f_a^{(0,k)}
}\eta
}=\braket{\lmk f_{a'}^{(0,k)}\otimes\xi\rmk}
{\lmk Y\otimes \unit_{\Mat_{k+1}} +\lmk\unit_{\Mat_{k+1}}+ Y\rmk\otimes Y \rmk^j \lmk f_a^{(0,k)}\otimes\eta\rmk},
\]
for
$\xi,\eta\in \cc^{k+1}$.
From Lemma \ref{lem:ccac}, we have 
\begin{align}\label{eq:decomac}
&
\sum_{\stackrel{j_1,j_2=0,\ldots,k:}{j\le j_1+j_2}}
\quad \alpha_{(j_1,j_2)}^{(j)}
\lmk \unit\otimes \eijn{bb}\rmk
\lmk
\sum_{a'=1}^{k}\braket{f_{a'}^{(0,k)}}{Y^{j_1}f_a^{(0,k)}}\hat y_{a',\alpha,\beta}^{(l_1)} \lmk \unit\otimes Y^{j_2}\rmk 
\rmk
\lmk \unit\otimes \eijn{b^{'}b^{'}}\rmk\notag\\
&=\sum_{\stackrel{j_1,j_2=0,\ldots,k:}{j_2\le j\le j_1+j_2}}
\quad {}_jC_{j_2}\cdot\sum_{i=0}^{j_2}\delta_{j_1,j-i}\cdot {}_{j_2}C_{i}\cdot
\lmk \unit\otimes \eijn{bb}\rmk
\lmk
\sum_{a'=1}^{k}\braket{f_{a'}^{(0,k)}}{Y^{j_1}f_a^{(0,k)}}\hat y_{a',\alpha,\beta}^{(l_1)} \lmk\unit\otimes Y^{j_2}\rmk 
\rmk
\lmk \unit\otimes \eijn{b^{'}b^{'}}\rmk.
\end{align}
If $0\le j\le k$, we have
\begin{align}
&(\ref{eq:decomac})
=
\sum_{j_2=0}^j\sum_{j_1=j-j_2}^{j}
\quad {}_jC_{j_2}\cdot {}_{j_2}C_{j-j_1}\cdot
\lmk \unit\otimes \eijn{bb}\rmk
\lmk
\sum_{a'=1}^{k}\braket{f_{a'}^{(0,k)}}{Y^{j_1}f_a^{(0,k)}}\hat y_{a',\alpha,\beta}^{(l_1)} \lmk\unit\otimes  Y^{j_2}\rmk 
\rmk
\lmk \unit\otimes \eijn{b^{'}b^{'}}\rmk\notag\\
&=
\sum_{j_2=0}^j\sum_{i=0}^{j_2}
\quad {}_jC_{j_2}\cdot {}_{j_2}C_{i}\cdot
\lmk \unit\otimes \eijn{bb}\rmk
\lmk
\sum_{a'=1}^{k}\braket{f_{a'}^{(0,k)}}{Y^{j-i}f_a^{(0,k)}}\hat y_{a',\alpha,\beta}^{(l_1)} \lmk\unit\otimes  Y^{j_2}\rmk 
\rmk
\lmk \unit\otimes \eijn{b^{'}b^{'}}\rmk
\notag\\
&=
\sum_{j_2=0}^j\sum_{i=0}^{j_2}
\quad {}_jC_{j_2}\cdot {}_{j_2}C_{i}\cdot
\lmk \unit\otimes \eijn{bb}\rmk
\lmk
\sum_{a'=1}^{k}\braket{f_{a'}^{(0,k)}}{Y^{j-j_2} Y^{j_2-i}f_a^{(0,k)}}\hat y_{a',\alpha,\beta}^{(l_1)} \lmk\unit\otimes  Y^{j_2}\rmk 
\rmk
\lmk \unit\otimes \eijn{b^{'}b^{'}}\rmk
\notag\\
&=
\sum_{j_2=0}^j
\quad {}_jC_{j_2}\cdot 
\lmk \unit\otimes \eijn{bb}\rmk
\lmk
\sum_{a'=1}^{k}\braket{f_{a'}^{(0,k)}}{Y^{j-j_2} \lmk \unit +Y\rmk^{j_2}f_a^{(0,k)}}\hat y_{a',\alpha,\beta}^{(l_1)} \lmk\unit\otimes  Y^{j_2}\rmk 
\rmk
\lmk \unit\otimes \eijn{b^{'}b^{'}}\rmk
\notag\\
&= 
\lmk \unit\otimes \eijn{bb}\rmk
\lmk
\sum_{a'=1}^{k}\hat y_{a',\alpha,\beta}^{(l_1)}\lmk\unit\otimes 
\braket{f_{a'}^{(0,k)}}
{\lmk Y\otimes \unit +\lmk\unit+ Y\rmk\otimes Y \rmk^j f_a^{(0,k)}}
\rmk\rmk
\lmk \unit\otimes \eijn{b^{'}b^{'}}\rmk.
\end{align}
If $k<j\le 2k$, then
\begin{align}
&(\ref{eq:decomac})
=
\sum_{j_2=0}^k\sum_{j_1=j-j_2}^{k}
\quad {}_jC_{j_2}\cdot {}_{j_2}C_{j-j_1}\cdot
\lmk \unit\otimes \eijn{bb}\rmk
\lmk
\sum_{a'=1}^{k}\braket{f_{a'}^{(0,k)}}{Y^{j_1}f_a^{(0,k)}}\hat y_{a',\alpha,\beta}^{(l_1)} \lmk\unit\otimes  Y^{j_2}\rmk 
\rmk
\lmk \unit\otimes \eijn{b^{'}b^{'}}\rmk\notag\\
&=
\sum_{j_2=0}^j\sum_{j_1=j-j_2}^{j}
\quad {}_jC_{j_2}\cdot {}_{j_2}C_{j-j_1}\cdot
\lmk \unit\otimes \eijn{bb}\rmk
\lmk
\sum_{a'=1}^{k}\braket{f_{a'}^{(0,k)}}{Y^{j_1}f_a^{(0,k)}}\hat y_{a',\alpha,\beta}^{(l_1)} \lmk\unit\otimes  Y^{j_2}\rmk 
\rmk
\lmk \unit\otimes \eijn{b^{'}b^{'}}\rmk\notag\\
&= 
\lmk \unit\otimes \eijn{bb}\rmk
\lmk
\sum_{a'=1}^{k}\hat y_{a',\alpha,\beta}^{(l_1)}\lmk
\unit\otimes 
\braket{f_{a'}^{(0,k)}}
{\lmk Y\otimes \unit +\lmk\unit+ Y\rmk\otimes Y \rmk^j f_a^{(0,k)}}
\rmk\rmk
\lmk \unit\otimes \eijn{b^{'}b^{'}}\rmk.
\end{align}
In the second equality, we used $Y^{j_1}=0$ for $j_1>k$.
Hence we proved the claim (\ref{eq:eachj}).

By (\ref{eq:eachj}) and {\it 7}, we obtain
\begin{align*}
&\hat\Lambda^l \hat y_{a,\alpha,\beta}^{(l_1)}
=\sum_{j=0}^{\min\{k,l\}}{}_lC_{j}\cdot\lmk\unit\otimes  \Lambda_{\lal}^l\rmk 
\lmk\unit\otimes Y^j \rmk\hat y_{a,\alpha,\beta}^{(l_1)}\\
&=\sum_{j=0}^{\min\{k,l\}}{}_lC_{j}\cdot\lmk\unit\otimes  \Lambda_{\lal}^l\rmk 
\lmk
\sum_{a'=1}^{k}\hat y_{a',\alpha,\beta}^{(l_1)}
\lmk \unit \otimes \braket{f_{a'}^{(0,k)}}
{\lmk Y\otimes \unit +\lmk\unit+ Y\rmk\otimes Y \rmk^j f_a^{(0,k)}}\rmk
\rmk
\notag\\
&=\sum_{j=0}^{\min\{k,l\}}{}_lC_{j}\cdot\lambda_{a}^l
\lmk
\sum_{a'=1}^{k}\hat y_{a',\alpha,\beta}^{(l_1)} \lmk\unit\otimes  \Lambda_{\lal}^l\rmk 
\lmk\unit\otimes \braket{f_{a'}^{(0,k)}}
{\lmk Y\otimes \unit +\lmk\unit+ Y\rmk\otimes Y \rmk^j f_a^{(0,k)}}
\rmk\rmk
\notag\\
&=\sum_{j=0}^{l}{}_lC_{j}\cdot\lambda_{a}^l
\lmk
\sum_{a'=1}^{k}\hat y_{a',\alpha,\beta}^{(l_1)} \lmk\unit\otimes  \Lambda_{\lal}^l\rmk 
\lmk\unit\otimes \braket{f_{a'}^{(0,k)}}
{\lmk Y\otimes \unit +\lmk\unit+ Y\rmk\otimes Y \rmk^j f_a^{(0,k)}}
\rmk\rmk
\notag\\
&=\lambda_{a}^l
\lmk
\sum_{a'=1}^{k}\hat y_{a',\alpha,\beta}^{(l_1)} \lmk\unit\otimes  \Lambda_{\lal}^l
\braket{f_{a'}^{(0,k)}}
{\lmk \lmk\unit+ Y\rmk\otimes \lmk \unit+ Y \rmk\rmk ^l f_a^{(0,k)}}
\rmk\rmk
\notag\\
&=\sum_{a'=1}^{k}\braket{f_{a'}^{(0,k)}}{\tilde\Lambda^{l}f_a^{(0,k)}}\hat y_{a',\alpha,\beta}^{(l_1)}\hat\Lambda^l ,
\end{align*}
for {\it all}  $l\in\nan$, $l_1\ge l_0$, $a=1,\ldots,k$, and $\alpha,\beta=1,\ldots,n_0$.
In the second and third equality we used {\it 7.} and the fact that
$\braket{f_{a'}^{(0,k)}}
{\lmk Y\otimes \unit +\lmk\unit+ Y\rmk\otimes Y \rmk^j f_a^{(0,k)}}=0$
unless $\lambda_{a'}=\lambda_a$.
For the forth equality we used (\ref{eq:eachj}) for $k<j\le 2k$
and $Y^{k+1}=0$.
This proves
the first line of
 {\it 6}.
The second line can be checked by substituting the first line to the right hand side of the second equality.
\end{proof}

\begin{lem}\label{lem:dyl}
Let $(n_0,k,\oo,\vv,\lal,l_0,\{y_{a,\alpha,\beta}^{(l)}\})$ be a septuplet 
satisfying {\it Condition 6-0}.
Suppose that there exist strictly positive elements $h_a\in \Mat_{n_0}$, $a=0,\ldots, 1$
with $h_0=\unit$ such that 
\[
y_{a,\alpha,\beta}^{(l)}\lmk \unit\otimes E_{00}^{(0,k)}\rmk
=h_a^{\frac 12}\zeij{\alpha\beta}\otimes E_{a0}^{(0,k)},\quad
l\ge l_0, \quad a=0,\ldots, k,\quad \alpha,\beta=1,\ldots,n_0.
\]
Then there exist $R\in\mnz\otimes\Mat_{k+1}$, $Y\in \DT_{0,k+1}$, and
$D_a\in\DT_{0,k+1}$,  $a=1,\ldots,k$, satisfying the followings.
\begin{enumerate}
\renewcommand{\labelenumi}{\roman{enumi}}
\item We have $[\Lambda_\lal,Y ]=0$.
\item For any $a=1,\ldots,k$, we have
$\Lambda_\lal D_a=\lambda_{a}D_a\Lambda_{\lal} $.
\item  For any $a=1,\ldots,k$, we have 
$D_a\eijn{00}=\eijn{a0}$.
\item  The set of matrices $\{D_a\}_{a=1}^{k}\cup\{1\}$
is linearly independent.
\item For any $a,a'=1,\ldots,k$, 
$D_aD_{a'}$ belongs to the linear span of $\{D_b\mid \lambda_a\lambda_{a'}=\lambda_b\}$.
\item 
Set $\tilde \Lambda:=\Lambda_{\lal}\lmk \unit+Y\rmk$.
Then we have
\begin{align}\label{eq:ld}
&\tilde \Lambda^l D_a
=\sum_{a'=1}^{k}\braket{f_{a'}^{(0,k)}}{\tilde\Lambda^{l}f_a^{(0,k)}}D_{a'}\tilde\Lambda^l \notag\\
&D_a\tilde \Lambda^l 
=\sum_{a'=1}^{k}\braket{f_{a'}^{(0,k)}}{\tilde\Lambda^{-l}f_a^{(0,k)}}\tilde \Lambda^l D_{a'}
\end{align}
for all $l\in\nan$ and  $a=1,\ldots,k$.
\item 
For any $l\ge l_0$,
\[
\left \{\zeij{\alpha\beta}\otimes \tll{l}\mid \alpha,\beta=1,\ldots,n_0\right\}
\cup \left \{\zeij{\alpha\beta}\otimes D_a\tll{l}\mid \alpha,\beta=1,\ldots,n_0,a=1,\ldots,k\right\}
\]
is a basis of $\caK_{l}\lmk R\vv R^{-1}\rmk$.
\item
For each $\mu=1,\ldots,n$, there exist unique $x_{\mu a}\in\mnz$, $a=1,\ldots, k$
such that
\[
R v_{\mu} R^{-1}
=\omega_\mu\otimes \tl
+\sum_{a=1}^k x_{\mu a}\otimes D_a\tl.
\]
\end{enumerate}\end{lem}
\begin{proof}
Let $R\in\mnz\otimes\Mat_{k+1}$, $Y\in \DT_{0,k+1}$,
and $\{ \hat y_{a,\alpha,\beta}^{(l)}\}_{a=0,\ldots,k, \alpha,\beta=1,\ldots,n_0,l\ge l_0}$
given in Lemma \ref{lem:hy}.
We claim that there exist $D_a\in\Mat_{k+1}$,  $a=1,\ldots,k$, such that
\begin{align}\label{eq:yda}
\hat y_{a,\alpha,\beta}^{(l)}=\zeij{\alpha,\beta}\otimes D_a  \tll{l},\quad
\alpha,\beta=1,\ldots,n_0,\; a=1,\ldots,k, \; l\ge l_0.
\end{align}
Set $\tilde y_{a,\alpha,\beta}^{(l)}:=\hat y_{a,\alpha,\beta}^{(l)} \lmk\unit\otimes \tll{-l} \rmk$, and decompose it 
as $\tilde y_{a,\alpha,\beta}^{(l)}=\sum_{ij=0}^{k} x_{ij}^{(a,\alpha,\beta,l)}\otimes \eijn{ij}$
with $x_{ij}^{(a,\alpha,\beta,l)}\in \mnz$.
Note that 
\begin{align*}
&\hat y_{0,\alpha_1,\beta_1}^{(l_1)}\hat y_{a,\alpha,\beta}^{(l)}\hat y_{0,\alpha_2,\beta_2}^{(l_2)}\lmk\unit\otimes \eijn{00}\rmk
=\lmk \zeij{\alpha_1\beta_1}\otimes \tll{l_1}\rmk
\lmk \zeij{\alpha,\beta}\zeij{\alpha_2\beta_2}\otimes E_{a,0}^{(0,k)}\rmk\\
&=\delta_{\beta,\alpha_2}\delta_{\beta_1,\alpha}
\lmk \zeij{\alpha_1\beta_2}\otimes \tll{l_1} E_{a,0}^{(0,k)}\rmk
=\delta_{\beta,\alpha_2}\delta_{\beta_1,\alpha}
\sum_{a'=1}^k  \braket{\fiin{a'}}{\tll{l_1}\fiin{a}}
 \zeij{\alpha_1\beta_2}\otimes E_{a',0}^{(0,k)}\\
&=\delta_{\beta,\alpha_2}\delta_{\beta_1,\alpha}
\sum_{a'=1}^k  \braket{\fiin{a'}}{\tll{l_1}\fiin{a}}
\hat y_{a',\alpha_1,\beta_2}^{(l_1+l_2+l)}\lmk\unit\otimes \eijn{00}\rmk,
\end{align*}
for any $l,l_1,l_2\ge l_0$, $\alpha,\beta,\alpha_1,\beta_1,\alpha_2,\beta_2=1,\ldots, n_0$,
$a=1,\ldots,k$.
As $\hat y_{0,\alpha_1,\beta_1}^{(l_1)}\hat y_{a,\alpha,\beta}^{(l)}\hat y_{0,\alpha_2,\beta_2}^{(l_2)}-\delta_{\beta,\alpha_2}\delta_{\beta_1,\alpha}
\sum_{a'=1}^k  \braket{\fiin{a'}}{\tll{l_1}\fiin{a}}
\hat y_{a',\alpha_1,\beta_2}^{(l_1+l_2+l)}$ belongs to $\caK_{l+l_1+l_2}(R\vv R^{-1})$,
this equality and {\it 5} of Lemma \ref{lem:hy} implies 
\begin{align}\label{eq:yyy}
\hat y_{0,\alpha_1,\beta_1}^{(l_1)}\hat y_{a,\alpha,\beta}^{(l)}\hat y_{0,\alpha_2,\beta_2}^{(l_2)}
=\delta_{\beta,\alpha_2}\delta_{\beta_1,\alpha}
\sum_{a'=1}^k  \braket{\fiin{a'}}{\tll{l_1}\fiin{a}}
\hat y_{a',\alpha_1,\beta_2}^{(l_1+l_2+l)},
\end{align}for any $l,l_1,l_2\ge l_0$, $\alpha,\beta,\alpha_1,\beta_1,\alpha_2,\beta_2=1,\ldots, n_0$, $a=1,\ldots,k$.
The left hand side can be written
\begin{align}\label{eq:yyl}
&\hat y_{0,\alpha_1,\beta_1}^{(l_1)}\hat y_{a,\alpha,\beta}^{(l)}\hat y_{0,\alpha_2,\beta_2}^{(l_2)}
=\lmk \zeij{\alpha_1\beta_1}\otimes \tll{l_1}\rmk
\tilde y_{a,\alpha,\beta}^{(l)}\lmk\unit\otimes \tll{l}\rmk
\lmk \zeij{\alpha_2\beta_2}\otimes \tll{l_2}\rmk\notag\\
&=\sum_{ij=0}^{k} \zeij{\alpha_1\beta_1}x_{ij}^{(a,\alpha,\beta,l)}\zeij{\alpha_2\beta_2}\otimes 
\tll{l_1}\eijn{ij}\tll{l+l_2}.
\end{align}
Let $\beta\neq\alpha_2$ or $\beta_1\neq\alpha$ in (\ref{eq:yyy}).
Comparing with (\ref{eq:yyl}), we obtain 
\[
 \braket{\cnz{\beta_1}}{x_{ij}^{(a,\alpha,\beta,l)}\cnz{\alpha_2}}=0,
\]
for  any $i,j=0,\ldots,k$,
$l\ge l_0$, 
$a=1,\ldots,k$, with
$\alpha,\beta,\beta_1,\alpha_2=1,\ldots, n_0$,
such that  $\beta\neq\alpha_2$ or $\beta_1\neq\alpha$.
This means that there exist $c_{ij}^{(a,\alpha,\beta,l)}\in\cc$,
$i,j=0,\ldots,k$,
$l\ge l_0$, 
$a=1,\ldots,k$, 
$\alpha,\beta=1,\ldots, n_0$, such that
\[
x_{ij}^{(a,\alpha,\beta,l)}=c_{ij}^{(a,\alpha,\beta,l)}\zeij{\alpha\beta}.
\]
Set $Z_{a,\alpha,\beta}^{(l)}:=\sum_{ij=0}^{k}c_{ij}^{(a,\alpha,\beta,l)}\eijn{ij}$.
Then we have
\[
\tilde y_{a,\alpha,\beta}^{(l)}=\zeij{\alpha,\beta}\otimes Z_{a,\alpha,\beta}^{(l)},\quad
\alpha,\beta=1,\ldots,n_0, \; a=1,\ldots,k, \; l\ge l_0.
\]
Furthermore, considering the case $\beta=\alpha_2$ and $\beta_1=\alpha$ in (\ref{eq:yyy}), we see that
$\hat y_{0,\alpha_1,\alpha}^{(l_1)}\hat y_{a,\alpha,\beta}^{(l)}\hat y_{0,\beta,\beta_2}^{(l_2)}$ is independent of
$\alpha, \beta$.
Therefore, $Z_{a,\alpha,\beta}^{(l)}$ is $\alpha,\beta$-independent. We denote this  $\alpha,\beta$-independent
matrix by $\tilde Z_a^{(l)}$.
Lastly, we would like to show that  $\tilde Z_a^{(l)}$ is $l$-independent.
Note that for any $\alpha,\beta=1,\ldots,n_0$, $a=1,\ldots,k$, and $l_1,l_2\ge l_0$,
we have
\begin{align*}
\hat y_{a,\alpha,\beta}^{(l_1)}\lmk \unit \otimes \tll{l_2}\rmk
=\hat y_{a,\alpha,\beta}^{(l_1)}\sum_{\gamma=1}^{n_0} \hat y_{0,\gamma,\gamma}^{(l_2)}
=\hat y_{a,\alpha,\beta}^{(l_1+l_2)}
=\hat y_{a,\alpha,\beta}^{(l_2)}\sum_{\gamma=1}^{n_0} \hat y_{0,\gamma,\gamma}^{(l_1)}
=\hat y_{a,\alpha,\beta}^{(l_2)}\lmk \unit \otimes \tll{l_1}\rmk,
\end{align*}
by {\it Condition 5}.
Therefore, multiplying $ \unit \otimes  \tilde\Lambda^{-(l_1+l_2)}$
from right, we have
$
\tilde y_{a,\alpha,\beta}^{(l_1)}
=\tilde y_{a,\alpha,\beta}^{(l_2)}$.
This implies that $\tilde Z_a^{(l)}$ is $l$-independent, $a=1,\ldots,k$. 
We set $D_a:=\tilde Z_a^{(l)}$, $a=1,\ldots,k$, and obtain the claim (\ref{eq:yda}).

Next we show $D_a\in \DT_{0,k+1}$.
As the septuplet 
$(n_0,k,\oo,R\vv R^{-1},\lal,l_0,\{\hat y_{a,\alpha,\beta}^{(l)}\})$
satisfies {\it Condition 5},
we have $R v_{\mu} R^{-1},\hat y_{a,\alpha,\beta}^{(l)}\in\mnz\otimes \DT_{k+1}$.
We also have $\unit\otimes\tll{l}\in
\mnz\otimes \DT_{k+1}$ because $Y\in \DT_{0,k+1}$.
Therefore, we have $\zeij{\alpha,\beta}\otimes D_a =\hat y_{a,\alpha,\beta}^{(l)}\lmk \unit\otimes \tll{-l}\rmk\in\mnz\otimes \DT_{k+1}$.
Furthermore, for any $i=0,\ldots,k$ and $a=1,\ldots,k$, we have
\begin{align*}
&\zeij{\alpha,\beta}\otimes \eijn{ii} D_a \eijn{ii}
= \lmk \unit\otimes \eijn{ii}\rmk \hat y_{a,\alpha,\beta}^{(l)} \lmk \unit\otimes \eijn{ii}\rmk\lmk \unit\otimes  \eijn{ii}\tll{-l} \eijn{ii}\rmk
=\lmk \unit\otimes \eijn{i0}\rmk \hat y_{a,\alpha,\beta}^{(l)} \lmk \unit\otimes \eijn{0i}\rmk\\
&=\lmk \unit\otimes \eijn{i0}\rmk \lmk\zeij{\alpha\beta}\otimes \eijn{a0}\rmk \lmk \unit\otimes \eijn{0i}\rmk=0.
\end{align*}
Here, we used the fact that 
as $\hat y_{a,\alpha,\beta}^{(l)}, \unit\otimes \tll{-l} \in\mnz\otimes \DT_{k+1}$.
From this fact, only the "diagonal part" of these matrices can contribute  
 when $\hat y_{a,\alpha,\beta}^{(l)}\lmk \unit\otimes \tll{-l}\rmk$ is sandwiched by $\unit\otimes E_{ii}^{(0,k)}$.
This corresponds to the first equality.
The similar consideration and (iii) of {\it Condition 5} for $(n_0,k,\oo,R\vv R^{-1},\lal,l_0,\{\hat y_{a,\alpha,\beta}^{(l)}\})$
implies the second equality.
For the third equality, we used {\it 4} of Lemma \ref{lem:hy}.
This proves $D_a\in \DT_{0,k+1}$.

Now let us check the properties i-vii.
 i. is from  Lemma \ref{lem:hy}.
ii. follows from the definition of $D_a$, and {\it 1,7} of Lemma \ref{lem:hy}.
Note that $Y\eijn{00}=0$ because $Y\in\DT_{0,k+1}$,  $[Y,\Lambda_\lal]=0$,
and $\lambda_i\neq1$ if $i\neq 0$,
From this and {\it 4} of Lemma \ref{lem:hy}  and (\ref{eq:yda}),
we obtain iii.
iii implies iv.
By {\it 6} of Lemma \ref{lem:hy} and (\ref{eq:yda}), we have
\begin{align*}
&\zeij{\alpha,\beta}\otimes\tll{l} D_a  \tll{l_1},
=\sum_{a'=1}^{k}\braket{f_{a'}^{(0,k)}}{\tilde\Lambda^{l}f_a^{(0,k)}}
\zeij{\alpha,\beta}\otimes D_{a'}  \tll{l_1+l},
\end{align*}
for all $l\in\nan$,  $l_1\ge l_0$, $a=1,\ldots,k$, and $\alpha,\beta=1,\ldots,n_0$.
This implies the first equation of vi.
The second one follows from this.
From {\it 6} of Lemma \ref{lem:hy} and (\ref{eq:yda}), we have
\begin{align*}
&\unit\otimes D_aD_{a'}
=\sum_{\alpha=1}^{n_0}\hat y_{a,\alpha,\alpha}^{(l)}\lmk\unit\otimes\tll{-l}\rmk
\hat y_{a',\alpha,\alpha}^{(l)}\lmk\unit\otimes\tll{-l}\rmk
=\sum_{\alpha=1}^{n_0}
\sum_{b=1}^{k}\braket{f_{b}^{(0,k)}}{\tilde\Lambda^{-l}f_{a'}^{(0,k)}}
\hat y_{a,\alpha,\alpha}^{(l)}\hat y_{b,\alpha,\alpha}^{(l)}
\lmk\unit\otimes\tll{-2l}\rmk\\
&\in
\caK_{2l}\lmk R\vv R^{-1}\rmk\lmk\unit\otimes\tll{-2l}\rmk
=\spa\left\{
\hat y_{b,\alpha,\beta}^{(2l)}\lmk\unit\otimes\tll{-2l}\rmk
\right\}_{b=0,\ldots,k,\; \alpha,\beta=1,\ldots, n_0}
=\Mat_{n_0}\otimes\spa\lmk \{\unit\}\cup\{D_b\}_{b=1}^k\rmk,
\end{align*}
for all $a,a'=1,\ldots, k$, $l\ge l_0$.

On the other hand, 
from ii,
we have
$
\Lambda_{\lal} D_aD_{a'}
=\lambda_a\lambda_{a'}
D_aD_{a'}\Lambda_{\lal}$.
Combining these and iv, we obtain v.
vii follows directly from {\it Condition 5} (iv)(1) of  
$(n_0,k,\oo,R\vv R^{-1},\lal,l_0,\{\hat y_{a,\alpha,\beta}^{(l)}\})$.

Finally, we prove viii. Set $\hat \Lambda:=\unit\otimes \tilde \Lambda$.
Note that $\hat \Lambda^l=\unit\otimes\tilde\Lambda^l=\sum_{\gamma=1}^{n_0}\hat y_{0\gamma\gamma}^{(l)}\in \caK_l\lmk
R\vv R^{-1}\rmk$ for $l\ge l_0$.
Let $\mu=1,\ldots,n$, and $l_2\ge l_0$.
Then we have
\begin{align*}
&Rv_{\mu} R^{-1}\hll{l_2}\lmk\unit\otimes \eijn{00}\rmk
=Rv_{\mu} R^{-1}\lmk\unit\otimes \eijn{00}\rmk\\
&=\sum_{a=1}^{k}\sum_{\alpha,\beta=1}^{n_0}
\braket{\cnz{\alpha}\otimes\fiin{a}}{Rv_{\mu} R^{-1}
\lmk\cnz{\beta}\otimes\fiin{0}\rmk}
\lmk\zeij{\alpha\beta}\otimes \eijn{a0}\rmk\\
&\quad
+ 
\lmk\unit\otimes\eijn{00}\rmk
Rv_{\mu} R^{-1}
\lmk\unit\otimes \eijn{00}\rmk\\
&=\sum_{a=1}^{k}\sum_{\alpha,\beta=1}^{n_0}
\braket{\cnz{\alpha}\otimes\fiin{a}}{Rv_{\mu} R^{-1}
\lmk\cnz{\beta}\otimes\fiin{0}\rmk}
\lmk\zeij{\alpha\beta}\otimes \eijn{a0}\rmk\\
&\quad
+ 
\omega_{\mu}\otimes \eijn{00}\\
&=\sum_{a=1}^{k}\sum_{\alpha,\beta=1}^{n_0}
\braket{\cnz{\alpha}\otimes\fiin{a}}{Rv_{\mu} R^{-1}
\lmk\cnz{\beta}\otimes\fiin{0}\rmk}
\hat y_{a,\alpha,\beta}^{(l_2+1)}\lmk\unit\otimes \eijn{00}\rmk\\
&\quad
+\sum_{\alpha,\beta=1}^{n_0}
\braket{\cnz{\alpha}}{\omega_\mu\cnz{\beta}}
\hat y_{0,\alpha,\beta}^{(l_2+1)}\lmk\unit\otimes \eijn{00}\rmk
\end{align*}
As $ \hat y_{b,\alpha,\beta}^{(l_2+1)}$ and $Rv_{\mu} R^{-1}\hll{l_2}$
belong to $\caK_{l_2+1}(Rv_{\mu} R^{-1})$,
from {\it 5} of Lemma \ref{lem:hy}, we obtain
\begin{align*}
&Rv_{\mu} R^{-1}\hll{l_2}\\
&=\sum_{a=1}^{k}\sum_{\alpha,\beta=1}^{n_0}
\braket{\cnz{\alpha}\otimes\fiin{a}}{Rv_{\mu} R^{-1}
\lmk\cnz{\beta}\otimes\fiin{0}\rmk}
\hat y_{a,\alpha,\beta}^{(l_2+1)}
+\sum_{\alpha,\beta=1}^{n_0}
\braket{\cnz{\alpha}}{\omega_\mu\cnz{\beta}}
\hat y_{0,\alpha,\beta}^{(l_2+1)}\\
&=\sum_{a=1}^{k}\sum_{\alpha,\beta=1}^{n_0}
\braket{\cnz{\alpha}\otimes\fiin{a}}{Rv_{\mu} R^{-1}
\lmk\cnz{\beta}\otimes\fiin{0}\rmk}
\zeij{\alpha\beta}\otimes D_{a}\tll{l_2+1}
+\omega_\mu\otimes \tll{l_2+1}.
\end{align*}
Multiplying $\tll{-l_2}$ from right,
we obtain
\begin{align*}
&Rv_{\mu} R^{-1}\\
&=\sum_{a=1}^{k}\sum_{\alpha,\beta=1}^{n_0}
\braket{\cnz{\alpha}\otimes\fiin{a}}{Rv_{\mu} R^{-1}
\lmk\cnz{\beta}\otimes\fiin{0}\rmk}
\zeij{\alpha\beta}\otimes D_{a}\tl
+\omega_\mu\otimes \tl.
\end{align*}
This gives the representation of $R v_{\mu}R^{-1}$ in viii.
The uniqueness of the representation follows from iv.
\end{proof}
We would like to apply the last Lemma to our setting.
In order to do so, we need the following Lemma.
\begin{lem}\label{lem:lrc}
Let $n_0^{(\sigma)},n \in\nan$ and  $\oo^{(\sigma)}\in\Primz_u(n,n_0^{(\sigma)})$,
for each $\sigma=L,R$. Let $\rho_\sigma$ be the faithful
$T_{\oo^{(\sigma)}}$-invariant state.
Let $\tilde \rho$ be the density matrix of $\rho_L$.
Let $\varphi$ be a state on $\caA_{\bbZ}$.For each $N\in \nan$, let $D_N$ be the density matrix of $\varphi\vert_{\caA_{[0,N-1]}}$, and assume $\sup_N\rank D_N<\infty$.
Assume that $(\Mat_{n_0^{(R)}}, \oo^{(R)},\rho_R)$ right-generates $\varphi$ and
that $(\Mat_{n_0^{(L)}}, \oo^{(L)},\rho_L)$ left-generates $\varphi$.
Then there exists an antiunitary 
$J:\cc^{n_0^{(R)}}\to\cc^{n_0^{(L)}}$ and $c\in\bbT$ such that
\[
\omega_{\mu}^{(R)}
=cJ^*\tilde\rho^{-\frac 12}\lmk \omega_\mu^{(L)}\rmk^*\tilde\rho^{\frac 12}J
\quad
\mu=1,\ldots,n.
\]
\end{lem}
\begin{proof}
As $\oo^{(L)}$ belongs to $\Primz_u(n,n_0^{(L)})$, the density matrix $\tilde\rho$
 is strictly positive and  has a decomposition $\tilde\rho=\sum_{i=1}^{n_0^{(L)}}\lambda_i
\ket{\xi_i}\bra{\xi_i}$ with a CONS $\{\xi_i\}_i$ of $\cc^{n_0^{(L)}}$ and 
$\lambda_i>0$.
Let $J_0$ be the operator of complex conjugation with respect to this CONS $\{\xi_i\}_i$.
Then $J_0$ is an antiunitary operator on $\cc^{n_0^{(L)}}$ such that $J_0^2=1$ and $J_0=J_0^*$.
Furthermore, we have 
\begin{align}\label{eq:rxj}
\rho_L\lmk J_0 X^* J_0\rmk
=
\Tr\tilde\rho\lmk J_0 X^* J_0\rmk=\Tr \tilde\rho X
=\rho_L\lmk X\rmk,\quad X\in \Mat_{n_0^{(L)}}.
\end{align}
Define $\tilde\oo^{(L)}=(\tilde\omega_\mu^{(L)})_{\mu=1}^n$ by
\[
\tilde\omega^{(L)}_\mu:= J_0 \tilde\rho^{-\frac 12}\lmk \omega_\mu^{(L)}\rmk^*\tilde\rho^{\frac 12}J_0,\quad
\mu=1,\ldots,n.
\]
As $\oo^{(L)}$ belongs to $\Primz_u(n,n_0^{(L)})$, 
and $\rho$ is $T_{\oo^{(L)}}$-invariant,
$\tilde \oo^{(L)}$ also belongs to $\Primz_u(n,n_0^{(L)})$. From (\ref{eq:rxj}) and the fact that $T_{\oo^{(L)}}$ is unital, $\rho_L$ is $T_{\tilde \oo^{(L)}}$-invariant.
We claim that $(\Mat_{n_0^{(L)}}, \tilde\oo^{(L)},\rho_L)$ right generates $\varphi$.
This can be seen by
\begin{align*}
&\rho_L\lmk
\tilde\omega^{(L)}_{\mu_{a}}\tilde\omega^{(L)}_{\mu_{a+1}}\cdots
\tilde\omega^{(L)}_{\mu_{a+l-1}}
\lmk \tilde\omega^{(L)}_{\nu_{a+l-1}}\rmk^*\cdots \lmk\tilde\omega^{(L)}_{\nu_{a+1}}\rmk^*
\lmk\tilde\omega^{(L)}_{\nu_{a}}\rmk^*
\rmk\\
&=\rho_L\lmk
 J_0 \tilde\rho^{-\frac 12}\lmk
\omega^{(L)}_{\mu_a}\rmk^*\tilde\rho^{\frac 12}J_0
J_0 \tilde\rho^{-\frac 12}\lmk \omega^{(L)}_{\mu_{a+1}}\rmk^*\tilde\rho^{\frac 12}J_0
\cdots J_0 \tilde\rho^{-\frac 12}\lmk \omega^{(L)}_{\mu_{a+l-1}}\rmk ^*\tilde\rho^{\frac 12}J_0
J_0\tilde\rho^{\frac 12} \omega^{(L)}_{\nu_{a+l-1}}\tilde\rho^{-\frac 12}J_0 
\cdots 
J_0\tilde\rho^{\frac 12}\omega^{(L)}_{\nu_a}\tilde\rho^{-\frac 12}J_0 
\rmk\\
&=\rho_L\lmk
 J_0 \tilde\rho^{-\frac 12}\lmk \omega^{(L)}_{\mu_a}\rmk^*\lmk \omega^{(L)}_{\mu_{a+1}}\rmk^*
\cdots\lmk  \omega^{(L)}_{\mu_{a+l-1}}\rmk^*\tilde\rho\omega^{(L)}_{\nu_{a+l-1}}
\cdots 
\omega^{(L)}_{\nu_{a+1}}
\omega^{(L)}_{\nu_a}\tilde\rho^{-\frac 12}J_0 
\rmk\\
&=\rho_L\lmk
\tilde\rho^{-\frac 12}\lmk \omega^{(L)}_{\nu_{a+l-1}}\cdots \omega^{(L)}_{\nu_{a+1}}
\omega^{(L)}_{\nu_a}\rmk^*
\tilde\rho\omega^{(L)}_{\mu_{a+l-1}}\cdots \omega^{(L)}_{\mu_{a+1}}
\omega^{(L)}_{\mu_a}\tilde\rho^{-\frac 12}
\rmk\\
&=\rho_L\lmk
\omega^{(L)}_{\mu_{a+l-1}}\cdots \omega^{(L)}_{\mu_{a+1}}
\omega^{(L)}_{\mu_a}
\lmk \omega^{(L)}_{\nu_{a+l-1}}\cdots \omega^{(L)}_{\nu_{a+1}}
\omega^{(L)}_{\nu_a}\rmk^*
\rmk=\varphi\lmk
\bigotimes_{i=a}^{a+l-1}e_{\mu_i\nu_i}^{(n)}
\rmk,
\end{align*}
for $a\in\bbZ$, $l\in\nan$, $\mu_i,\nu_i\in\{1,\ldots,n\}$.
For the third equality, we used (\ref{eq:rxj}).

By this observation, we can apply Lemma \ref{lem:ouo}
to $\oo^{(R)}$ and $\tilde \oo^{(L)}$. We then obtain
a unitary $V:\cc^{n_0^{(R)}}\to\cc^{n_0^{(L)}}$ and $c\in\bbT$
such that
\[
V\omega_{\mu}^{(R)}=c\tilde \omega_{\mu}^{(L)}V
=cJ_0\tilde\rho^{-\frac 12}\lmk \omega_\mu^{(L)}\rmk^*\tilde\rho^{\frac 12}J_0V
\quad
\mu=1,\ldots,n.
\]
The operator $J:=J_0V$ is an antiunitary from $\cc^{n_0^{(R)}}$ to $\cc^{n_0^{(L)}}$ such that
\[
\omega_{\mu}^{(R)}
=cJ^*\tilde\rho^{-\frac 12}\lmk \omega_\mu^{(L)}\rmk^*\tilde\rho^{\frac 12}J
\quad
\mu=1,\ldots,n.
\]
\end{proof}
\begin{lem}\label{lem:lrf}
Assume [A1],[A3],[A4], and [A5]. 
There exist $n_0\in\nan$,
$\oo\in \Primz_u(n,n_0)$, $k_L,k_R\in\nan\cup\{0\}$,
$\hat \vv^{(L)}\in \lmk \Mat_{{n_{0}}}\otimes\Mat_{k_L+1}\rmk^{\times n}$,
$\hat \vv^{(R)}\in \lmk \Mat_{{n_{0}}}\otimes\Mat_{k_R+1}\rmk^{\times n}$,
 $\hat \lal^{(L)}=(\hat  \lambda^{(L)}_b)_{b=0,\ldots,k_L}\in \cc^{k_L+1}$,
$\hat \lal^{(R)}=(\hat \lambda^{(R)}_a)_{a=0,\ldots,k_R}\in \cc^{k_R+1}$,
$\hat Y_L\in\UT_{0, k_L+1} $, $\hat Y_R\in\DT_{0, k_R+1}$,
$\hat D_b^{(L)}\in\UT_{0, k_L+1} $, $\hat D_a^{(R)}\in\DT_{0, k_R+1}$,
$b=1,\ldots,k_L$, $a=1,\ldots, k_R$, and 
$l_{0}\in\nan$,
satisfying the followings.
\begin{enumerate}
\item We have $[\Lambda_{\hat \lal^{(\sigma)}},\hat Y_\sigma]=0$, for $\sigma=L,R$.
\item For each $\sigma=L,R$, $\hat \lambda_0^{(\sigma)}=1$ and $0<\lv\hat \lambda_a^{(\sigma)}\rv<1$ for all $a\ge 1$.
\item There exist $\{\hat x_{\mu,b}^{(L)}\}_{\mu=1,\ldots,n, b=1,\ldots, k_L}\subset\mnz$
and $\{\hat x_{\mu,a}^{(R)}\}_{\mu=1,\ldots,n, a=1,\ldots, k_R}\subset\mnz$
such that 
\begin{align*}
&\hat v_{\mu}^{(L)}=
\omega_\mu\otimes \Lambda_{\hat \lal^{(L)}}\lmk \unit+\hat Y_L\rmk
+\sum_{b=1}^{k_L} \hat x_{\mu b}^{(L)}\otimes \Lambda_{\hat \lal^{(L)}}\lmk \unit+\hat Y_L\rmk \hat D_b^{(L)},\\
&\hat v_{\mu}^{(R)}=
\omega_\mu\otimes \Lambda_{\hat \lal^{(R)}}\lmk \unit+\hat Y_R\rmk
+\sum_{a=1}^{k_R} \hat x_{\mu a}^{(R)}\otimes \hat D_a ^{(R)}\Lambda_{\hat \lal^{(R)}}\lmk \unit+\hat Y_R\rmk.
\end{align*}
\item We have
\begin{align*}
&\Lambda_{\hat \lal^{(R)}} \hat D_a^{(R)}=\hat \lambda_{a}^{(R)}
\hat D_a^{(R)}\Lambda_{\hat \lal^{(R)}} \quad a=1,\ldots,k_R,\\
&\Lambda_{\hat \lal^{(L)}} \hat D_b^{(L)}=\lmk \hat \lambda_{b}^{(L)}\rmk^{-1}
\hat D_b^{(L)}\Lambda_{\hat \lal^{(L)}} \quad b=1,\ldots,k_L.
\end{align*} 
\item We have
\begin{align*}
\hat D_a^{(R)}E_{00}^{(0,k_R)}=E_{a0}^{(0,k_R)},\quad a=1,\ldots,k_R,\\
E_{00}^{(0,k_L)}\hat D_b^{(L)}=E_{0b}^{(0,k_L)},\quad b=1,\ldots,k_L.
\end{align*}
\item The set of matrices $\{\hat D^{(\sigma)}_a\}_{a=1}^{k_\sigma}\cup\{\unit_{k_\sigma+1} \}$
is linearly independent for each $\sigma=L,R$.
\item For any $a,a'=1,\ldots,k_\sigma$, $\sigma=L,R$, 
$\hat D^{(\sigma)}_a\hat D^{(\sigma)}_{a'}$ belongs to the linear span of $\{\hat D^{(\sigma)}_b\mid \hat \lambda^{(\sigma)}_a\hat \lambda^{(\sigma)}_{a'}=\hat \lambda^{(\sigma)}_b\}$.
\item
Set $\tilde \Lambda_\sigma:=\Lambda_{\hat \lal^{(\sigma)}}\lmk \unit+\hat Y_\sigma\rmk$, $\sigma=L,R$.
Then we have
\begin{align}
&\tilde \Lambda_R^l \hat D^{(R)}_a
=\sum_{a'=1}^{k_R}\braket{f_{a'}^{(0,k_R)}}{\tilde\Lambda_R^{l}f_a^{(0,k_R)}}\hat D^{(R)}_{a'}\tilde\Lambda_R^l, \quad\hat D^{(R)}_a\tilde \Lambda_R^l 
=\sum_{a'=1}^{k_R}\braket{f_{a'}^{(0,k_R)}}{\tilde\Lambda_R^{-l}f_a^{(0,k_R)}}\tilde \Lambda_R^l \hat D^{(R)}_{a'},\notag\\
&\hat D^{(L)}_b\tilde \Lambda_L^l 
=\sum_{b'=1}^{k_L}\braket{f_{b}^{(0,k_L)}}{\tilde\Lambda_L^lf_{b'}^{(0,k_L)}}\tilde\Lambda_L^l\hat D^{(L)}_{b'} ,\quad
\tilde \Lambda_L^l \hat D^{(L)}_b
=\sum_{b'=1}^{k_L}\braket{f_{b}^{(0,k_L)}}{\tilde\Lambda_L^{-l}f_{b'}^{(0,k_L)}}\hat D^{(L)}_{b'}\tilde \Lambda_L^l ,
\end{align}
for all $l\in\nan$ and  $a=1,\ldots,k_R$, $b=1,\ldots, k_L$.
\item
For any $l\ge l_0$,
\[
\left \{\zeij{\alpha\beta}\otimes \tilde \Lambda_R^l\mid \alpha,\beta=1,\ldots,n_0\right\}
\cup \left \{\zeij{\alpha\beta}\otimes \hat D^{(R)}_a\tilde \Lambda_R^l\mid \alpha,\beta=1,\ldots,n_0,a=1,\ldots,k_R\right\}
\]
is a basis of $\caK_{l}\lmk \hat \vv^{(R)}\rmk$, and 
\[
\left \{\zeij{\alpha\beta}\otimes \tilde \Lambda_L^l\mid \alpha,\beta=1,\ldots,n_0\right\}
\cup \left \{\zeij{\alpha\beta}\otimes\tilde \Lambda_L^l \hat D^{(L)}_b\mid \alpha,\beta=1,\ldots,n_0,b=1,\ldots,k_L\right\}
\]
is a basis of $\caK_{l}\lmk \hat \vv^{(L)}\rmk$.
\item For  the state $\omega_\sigma$ in  [A4] and $ l\in\nan$,
 $\tau_{y_\sigma}\lmk s(\omega_{\sigma}\vert_{\caA_{\sigma,l}})\rmk$ is equal to
the orthogonal projection onto $\Gamma_{l,\hat \vv^{(\sigma)}}^{(R)}\lmk \mnz\otimes \Mat_{k_\sigma+1}\rmk$,
where $y_R=0$ and $y_L=l$.
\item Let $\rho_0$ be the $T_{\oo}$-invariant state on $\mnz$. Then 
$(\mnz,\oo,\rho_0)$
 right-generates $\omega_\infty$.
\end{enumerate}
\end{lem}
\begin{proof}
Recall the Notation \ref{nota:vv} and Notation \ref{nota:oa}.
Let $\sigma=L,R$.
We apply Lemma \ref{lem:dyl} to   
$\oo^{(\sigma)}\in \Primz_u(n,n_0^{(\sigma)})$,
$\vv^{(\sigma)}\in \lmk \Mat_{{n_{0}^{(\sigma)}}}\otimes\Mat_{k_\sigma+1}\rmk^{\times n}$,
 $\lal^{(\sigma)}=(\lambda^{(\sigma)}_a)_{a=0,\ldots,k_\sigma}\in \cc^{k_\sigma+1}$,
${l_{0}^{(\sigma)}}\in\nan$,
and 
 $\{y_{a,\alpha,\beta}^{(l,\sigma)}\}_{a=0,\ldots,k_\sigma,
\alpha,\beta=1,\ldots,{n_{0}^{(\sigma)}}}$
given by Lemma \ref{lem:ohy}.
We then obtain $R_\sigma\in\Mat_{n_0^{(\sigma)}}\otimes\Mat_{k_\sigma+1}$, $Y_\sigma\in \DT_{0,k_\sigma+1}$, and
$D_a^{(\sigma)}\in\DT_{0,k_\sigma+1}$,  $a=1,\ldots,k_\sigma$,
 satisfying the properties i-viii in Lemma  \ref{lem:dyl}.
Let $\rho_\sigma$ be the $T_{\oo^{(\sigma)}}$-invariant state for $\sigma=L,R$.

We set $\oo:=\oo^{(R)}$, $n_0:=n_0^{(R)}$
$\hat v_\mu^{(R)}:=R_R v_{\mu}^{(R)}R_R^{-1}$, $\mu=1,\ldots, n$,
$\hat\lal^{(R)}:=\lal^{(R)}$, $\hat Y_R:=Y_R$,
 $\hat D_a^{(R)}:= D_a^{(R)}$, and $l_0:=\max\{l_0^{(L)}, l_0^{(R)}\}$.

In order to define the left parts, we use Lemma \ref{lem:lrc}.
As $\omega_\infty$ is in $\caS_{\bbZ}\lmk H\rmk$, it satisfies the condition
in Lemma \ref{lem:lrc} for $\varphi$ because of [A1] and Lemma \ref{lem:vr}.
Recall  that $(\mnz,\oo,\rho_R\vert_{\mnz})$ right-generates $\omega_\infty$ and
that $(\Mat_{n_0^{(L)}},\oo^{(L)},\rho_L\vert_{\Mat_{n_0}^{(L)})}$ left-generates $\omega_\infty$.
Applying Lemma \ref{lem:lrc}, we obtain an antiunitary 
$J:\cc^{n_0}\to\cc^{n_0^{(L)}}$ and $c\in\bbT$ such that
\[
\omega_{\mu}
=cJ^*\tilde\rho^{-\frac 12}\lmk \omega_\mu^{(L)}\rmk^*\tilde\rho^{\frac 12}J
\quad
\mu=1,\ldots,n,
\]
where $\tilde \rho $ is a strictly positive element in $\Mat_{n_0^{(L)}}$.

Let $J_L:\bbC^{k_L+1}\to \bbC^{k_L+1}$ be the complex conjugation with respect to
the standard basis $\{f_{i}^{(0,k_L)}\}_{i=0,\ldots,k_L}$.
We set 
\[
\hat v_\mu^{(L)}
:=c
\lmk
J^*\otimes J_L^*
\rmk\lmk \tilde \rho^{-\frac 12}\otimes \unit \rmk
\lmk R_Lv_\mu^{(L)}R_L^{-1}
\rmk^*\lmk \tilde \rho^{\frac 12}\otimes \unit \rmk
\lmk
J\otimes J_L
\rmk
,\quad \mu=1,\ldots,k_L.
\]
Furthermore, we set 
$\hat\lal^{(L)}:=\lal^{(L)}$,  $\hat D_a^{(L)}:= \lmk D_a^{(L)}\rmk^t$, 
$\hat Y_L:=Y_L^t$.
It is straightforward to check that all the conditions in the Lemma are satisfied.
\end{proof}
We make a further simplification.
\begin{lem}
In Lemma \ref{lem:lrf}, we may assume $\hat \lal^{(\sigma)}$ to satisfy
\begin{align}\label{eq:ordl}
1=\lv\hat \lambda_{0}^{(L)}\rv>
\lv\hat \lambda_{1}^{(L)}\rv\ge \lv\hat \lambda_{2}^{(L)}\rv\cdots
\ge \lv\hat \lambda_{k_L}^{(L)}\rv,\notag\\
 1=\lv\hat \lambda_{0}^{(R)}\rv>\lv\hat \lambda_{1}^{(R)}\rv\ge \lv\hat \lambda_{2}^{(R)}\rv\cdots
\ge \lv\hat \lambda_{k_R}^{(R)}\rv
\end{align}
\end{lem}
\begin{proof}
We prove for $\sigma=L$. The proof for $\sigma=R$ is the same.
As there is nothing to prove if $k_L=0$, we may assume $k_L\in\nan$.
We define an order $\preceq$ on $\cc$, by
$\lambda\preceq \zeta$ if and only if
$|\lambda|<|\zeta|$ or $|\lambda|=|\zeta|$ and 
$\lambda=|\lambda|e^{i\theta}$, $\zeta=|\zeta|e^{i\varphi}$, with
$0\le \theta\le\varphi<2\pi$.
We write $\lambda\prec \zeta$ if $\lambda\preceq \zeta$ and $\lambda\neq \zeta$.
For $\hat\lal^{(L)}$ given in  Lemma \ref{lem:lrf}, we denote the disjoint elements of $\{\hat\lambda_{i}^{(L)}\}_{i=0}^{k_L}$ by 
$s_t$, $t=0,\ldots m$, $m\in\nan$, which is ordered as
$s_m\prec\cdots\prec s_2\prec s_1\prec s_0=1$.
For each $t=0,\ldots,m$, we define a set ${\mathfrak S}_t$ by
${\mathfrak S}_t:=\left \{
i=0,\ldots,k_L\mid \hat \lambda_{i}^{(L)}=s_t
\right \}$, and set $n_t:=|{\mathfrak S}_t|$, the number of elements in ${\mathfrak S}_t$.
Furthermore, we label the elements in ${\mathfrak S}_t$
as ${i_1^{(t)}}, {i_2^{(t)}}, \ldots,{i_{n_t}^{(t)}}$,
with labels ordered so that
${i_1^{(t)}}<{i_2^{(t)}}<\cdots <{i_{n_t}^{(t)}}$.
For each $i=0,\ldots,k_L$, there exists a unique pair, $(t,j)$, $t=0,\ldots,m$ and $j=1,\ldots,n_t$
such that $i=i_{j}^{(t)}$. We use this label to explain the rearrangement of 
$0,\ldots,k_L$. 
We permute $0,\ldots,k_L$ as follows.
We know that $0$ corresponds to $(0,1)$ with respect to the label, and from the beginning,
it is located at the first position.
First, we move $(1,1)$ to left one by one up to when 
it get next to $(0,1)$.
Second, we move $(1,2)$ to left one by one up to when 
it gets next to $(1,1)$.
We continue this up to when 
the first $n_1+1$ elements of the sequence become
$(0,0), (1,1), (1,2), \ldots (1, n_1)$.
When this is completed we move $(2,1)$, to left one by one up to when 
it gets next to $(1,n_1)$.
Repeating the same procedure for $(2,j)$, $j=1,\ldots,n_2$, 
the sequence is rearranged so that 
the first $n_1+n_2+1$ elements of the sequence become
$(0,0), (1,1), (1,2), \ldots (1, n_1),(2,1),\ldots,(2,n_2)$.
We repeat this procedure up to $t=m$ and obtain
the rearranged sequence 
$(0,0), (1,1), (1,2), \ldots (1, n_1),(2,1),\ldots,(2,n_2),\ldots, (m,1), \ldots,(m,n_m)$.
We denote by $S_{k_L+1}$ the permutation group of $\{0,\ldots,k_L\}$ .
Note that the above procedure consists of
the sequence of transpositions $\sigma_k$, $k=1,\ldots,N$ of the form
$\sigma_k=(i_k,i_k+1)\in S_{k_L+1}$ with $i_k\in\{1,\ldots,k_L-1\}$.
Furthermore these transpositions satisfy $\sigma_k(0)=0$, 
\begin{align}\label{eq:ls}
\hat\lambda^{(L)}_{\lmk \sigma_{k-1}\circ\cdots \circ \sigma_1\rmk^{-1}\lmk i_k\rmk}
\prec \hat\lambda^{(L)}_{\lmk \sigma_{k-1}\circ\cdots \circ \sigma_1\rmk^{-1}\lmk i_k+1\rmk},
\end{align}
and 
\begin{align}\label{eq:tij}
\lmk \sigma_k\circ\cdots\circ\sigma_1\rmk\lmk i_j^{(t)}\rmk
<\lmk \sigma_k\circ\cdots\circ\sigma_1\rmk\lmk i_{j'}^{(t)}\rmk,\quad
t=0,\ldots,m,\quad 1\le j<j'\le n_t,\quad k=1,\ldots,N .
\end{align}
Define $\sigma:=\sigma_N\circ\cdots\circ\sigma_1$.
Then we have
\[
\tilde\lal^{(L)}:=\lmk\tilde \lambda_{i}^{(L)}\rmk_{i=0,\ldots,k_L}:=
\lmk\hat \lambda^{(L)}_{\sigma^{-1}(0)}, \hat \lambda^{(L)}_{\sigma^{-1}(1)},\ldots,
\hat \lambda^{(L)}_{\sigma^{-1}(k_L)}\rmk
=\lmk
\hat\lambda^{(L)}_{0}, \hat\lambda^{(L)}_{i_1^{(1)}},\ldots,\hat\lambda^{(L)}_{i_{n_1}^{(1)}},
\hat\lambda^{(L)}_{i_1^{(2)}},\ldots,\hat\lambda^{(L)}_{i_{n_2}^{(2)}},\ldots,
\hat\lambda^{(L)}_{i_1^{(m)}},\ldots,\hat\lambda^{(L)}_{i_{n_m}^{(m)}}
\rmk.
\]
Note that
\[
1=\lv\tilde \lambda_{0}^{(L)}\rv>
\lv\tilde \lambda_{1}^{(L)}\rv\ge \lv\tilde \lambda_{2}^{(L)}\rv\cdots
\ge \lv\tilde \lambda_{k_L}^{(L)}\rv.
\]

For each $\sigma'\in S_{k_L+1}$, define $U_{\sigma'}\in \Mat_{k_L+1}$
such that $\braket{f_{i}^{(0,k_L)}}{U_{\sigma'}f_{j}^{(0,k_L)}}:=\delta_{\sigma'(i),j}$,
$i,j=0,\ldots,k_L$.
It is easy to check
\begin{align}
 \braket{f_{i}^{(0,k_L)}}{U_{\sigma'}^* XU_{\sigma'} f_{j}^{(0,k_L)}}
=\braket
{f_{\lmk \sigma'\rmk^{-1}(i)}^{(0,k_L)}}
{Xf_{\lmk \sigma'\rmk^{-1}(j)}^{(0,k_L)}},\quad i,j=0,\ldots,k_L,\quad X\in \Mat_{k+1}.
\end{align}
In particular, $U_{\sigma'}$ is unitary, and
$U_\sigma^*\Lambda_{\hat \lal^{(L)}}U_\sigma=\Lambda_{\tilde\lal^{(L)}}$.
Furthermore, we have $U_{\sigma'_2}U_{\sigma'_1}=U_{\sigma_1'\circ\sigma_2'}$,
for any $\sigma_1',\sigma_2'\in S_{k_L+1}$.

We set
\[
\tilde Y_L:=U_\sigma^* \hat Y_L U_\sigma,\quad
 \tilde D_b^{(L)}:=U_\sigma^* \hat D_{\sigma^{-1}\lmk b\rmk}^{(L)}U_\sigma,\quad
\tilde v_{\mu}^{(L)}:=\lmk\unit\otimes U_\sigma^*\rmk \hat  v_{\mu}^{(L)}\lmk\unit\otimes U_\sigma\rmk.
\]
It is straightforward to check that $n_0\in\nan$,
$\oo\in \Primz_u(n,n_0)$, $k_L\in\nan\cup\{0\}$ of Lemma \ref{lem:lrf},
$\tilde \vv^{(L)}\in \lmk \Mat_{{n_{0}}}\otimes\Mat_{k_L+1}\rmk^{\times n}$,
 $\tilde \lal^{(L)}=(\tilde  \lambda^{(L)}_b)_{b=0,\ldots,k_L}\in \cc^{k_L+1}$,
$\tilde Y_L\in\Mat_{k_L+1} $,
$\tilde D_b^{(L)}\in\Mat_{ k_L+1} $,
$b=1,\ldots,k_L$, and 
$l_{0}\in\nan$ (given in Lemma \ref{lem:lrf}),
satisfy the "left  part" of  {\it 1-11} of Lemma \ref{lem:lrf},
(replacing $\hat{} $ by $\tilde{}$.)

We have to prove that $\tilde Y_L, \tilde D_b\in\UT_{0,k_L+1}$.
To prove $\tilde Y_L\in\UT_{0,k_L+1}$, note that
\[
\braket
{f_{i}^{(0,k_L)}}
{\tilde Y_Lf_{j}^{(0,k_L)}}
=
\braket
{f_{\lmk \sigma\rmk^{-1}(i)}^{(0,k_L)}}
{\hat Y_Lf_{\lmk \sigma\rmk^{-1}(j)}^{(0,k_L)}}
,\quad i,j=0,\ldots,k_L.
\]
The right hand side is zero if $\hat\lambda^{(L)}_{\lmk \sigma\rmk^{-1}(i)}\neq \hat\lambda^{(L)}_{\lmk \sigma\rmk^{-1}(j)}$.
If $\hat\lambda^{(L)}_{\lmk \sigma\rmk^{-1}(i)}= \hat\lambda^{(L)}_{\lmk \sigma\rmk^{-1}(j)}$ and
$i\ge j$, then by (\ref{eq:tij}), we have
$\lmk \sigma\rmk^{-1}(i)\ge \lmk \sigma\rmk^{-1}(j)$.
As $\hat Y_L\in \UT_{0,k_L+1}$, in this case, the right hand side of the above equation is zero.
Therefore, we obtain $\tilde Y_L\in \UT_{0,k_L+1}$.

Next, to prove $\tilde D_b\in \UT_{0,k_L+1}$, note that
\begin{align*}
 \braket{f_{i}^{(0,k_L)}}{U_{\sigma}^* \hat D_b^{(L)}U_{\sigma} f_{j}^{(0,k_L)}}
=\braket
{f_{\lmk \sigma\rmk^{-1}(i)}^{(0,k_L)}}
{\hat D_b^{(L)}f_{\lmk \sigma\rmk^{-1}(j)}^{(0,k_L)}},\quad i,j=0,\ldots,k_L.
\end{align*}
We consider the following proposition for $k=0,\ldots,N$:
\begin{align*}
(P_k)  : \text{ If } 0\le j\le i\le k_L , \text{ then } 
\braket
{f_{\lmk \sigma_k\circ\cdots\circ\sigma_1\rmk^{-1}(i)}^{(0,k_L)}}
{\hat D_b^{(L)}f_{\lmk\sigma_k\circ\cdots\circ\sigma_1\rmk^{-1}(j)}^{(0,k_L)}}
 =0.
\end{align*}
(Regard $\sigma_k\circ\cdots\circ\sigma_1$ to be identity for $k=0$.)
It suffices to prove $(P_N)$.
$(P_0)$ is true, for $\hat D_b^{(L)}\in\UT_{0,k_L+1}$.

Assume that $(P_{k-1})$ holds for some $1\le k\le N$.
We claim that $(P_{k})$ holds.
Consider $ 0\le j\le i\le k_L$. If $i=j$, then
$\braket
{f_{\lmk \sigma_k\circ\cdots\circ\sigma_1\rmk^{-1}(i)}^{(0,k_L)}}
{\hat D_b^{(L)}f_{\lmk\sigma_k\circ\cdots\circ\sigma_1\rmk^{-1}(i)}^{(0,k_L)}}
$
is zero because the diagonal elements of $\hat D_b^{(L)}$ are zero.
Suppose that $j<i$.
Note that
\begin{align}\label{eq:skk}
\braket
{f_{\lmk \sigma_k\circ\cdots\circ\sigma_1\rmk^{-1}(i)}^{(0,k_L)}}
{\hat D_b^{(L)}f_{\lmk\sigma_k\circ\cdots\circ\sigma_1\rmk^{-1}(j)}^{(0,k_L)}}
=\braket
{f_{\lmk \sigma_{k-1}\circ\cdots\circ\sigma_1\rmk^{-1}\circ\sigma_k^{-1}(i)}^{(0,k_L)}}
{\hat D_b^{(L)}f_{\lmk\sigma_{k-1}\circ\cdots\circ\sigma_1\rmk^{-1}\circ\sigma_k^{-1}(j)}^{(0,k_L)}}.
\end{align}
Recall that $\sigma_k=(i_k,i_k+1)$.
If $i,j\notin \{i_k,i_k+1\}$, then
$\sigma_k^{-1}(i)=i>j=\sigma_{k}^{-1}(j)$.
If $i\in \{i_k,i_k+1\}$ and  $j\notin \{i_k,i_k+1\}$
(resp. $j\in \{i_k,i_k+1\}$ and  $i\notin \{i_k,i_k+1\}$), then 
\[
\sigma_k^{-1}(i)=i_k \text{ or  } i_k+1>j=\sigma_{k}^{-1}(j),\quad
(\text{ resp. }\sigma_k^{-1}(j)=i_k \text{ or  } i_k+1<i=\sigma_{k}^{-1}(i)).
\]
Therefore, if $i\notin \{i_k,i_k+1\}$ or $j\notin \{i_k,i_k+1\}$,
we have $\sigma_k^{-1}(i)>\sigma_{k}^{-1}(j)$, and 
the assumption $(P_{k-1})$ implies the right hand side of (\ref{eq:skk}) to be zero.

Let us consider the case $i,j\in \{i_k,i_k+1\}$, i.e., $i=i_k+1$ and $j=i_k$.
By (\ref{eq:ls})
and the definition of $\prec$ combined with the fact that $\lv\hat\lambda_{b}^{(L)}\rv<1$,  we have
\[
\hat\lambda^{(L)}_{\lmk \sigma_{k-1}\circ\cdots \circ \sigma_1\rmk^{-1}\lmk i_k+1\rmk}
\neq \hat\lambda^{(L)}_{\lmk \sigma_{k-1}\circ\cdots \circ \sigma_1\rmk^{-1}\lmk i_k\rmk}\hat\lambda_{b}^{(L)}.
\]
Therefore, in this case, the right hand side of  (\ref{eq:skk}) is
\begin{align*}
&\braket
{f_{\lmk \sigma_{k-1}\circ\cdots\circ\sigma_1\rmk^{-1}\circ\sigma_k^{-1}(i)}^{(0,k_L)}}
{\hat D_b^{(L)}f_{\lmk\sigma_{k-1}\circ\cdots\circ\sigma_1\rmk^{-1}\circ\sigma_k^{-1}(j)}^{(0,k_L)}}
=\braket
{f_{\lmk \sigma_{k-1}\circ\cdots\circ\sigma_1\rmk^{-1}\circ\sigma_k^{-1}(i_{k}+1)}^{(0,k_L)}}
{\hat D_b^{(L)}f_{\lmk\sigma_{k-1}\circ\cdots\circ\sigma_1\rmk^{-1}\circ\sigma_k^{-1}(i_k)}^{(0,k_L)}}\\
&=\braket
{f_{\lmk \sigma_{k-1}\circ\cdots\circ\sigma_1\rmk^{-1}\lmk i_{k}\rmk}^{(0,k_L)}}
{\hat D_b^{(L)}f_{\lmk\sigma_{k-1}\circ\cdots\circ\sigma_1\rmk^{-1}\lmk i_k+1\rmk}^{(0,k_L)}}=0,
\end{align*}
by property {\it 4} of Lemma \ref{lem:lrf}.
This completes the proof of the induction $(P_{k})$.
\end{proof}
\section{Derivation of $\bb\in\ClassA$}\label{sec:nana}
In this section, we prove the following Lemma.
\begin{lem}\label{lem:dercla}
Assume [A1],[A3],[A4], and [A5]. 
Then there exists $\bb\in\Class A$ with respect to a septuplet $(n_0,k_R,k_L,\lal, \bbD,\bbG,Y)$, 
satisfying the following conditions.
\begin{enumerate}
\item For $m'\ge m_\bb$, we have $\caS_{[0,\infty)}(H_{\Phi_{m',\bb}})=\caS_{[0,\infty)}(H)$, and 
$\caS_{(-\infty,-1]}(H_{\Phi_{m',\bb}})=\caS_{(-\infty,-1]}(H)$.
\item For  the state $\omega_L$ in  [A4] ,
 $\tau_{l}\lmk s(\omega_{L}\vert_{\caA_{L,l}})\rmk$ is equal to the orthogonal
projection onto $\Gamma_{l,\bb}^{(R)}\lmk \mnz\otimes \pd \Mat_{k_L+k_R+1}\pd \rmk$, for all $l\in\nan$.
\item   For  the state $\omega_R$ in  [A4] .
 $s(\omega_{R}\vert_{\caA_{R,l}})$ is equal to
the orthogonal projection onto \\
$\Gamma_{l,\bb}^{(R)}\lmk \mnz\otimes \pu \Mat_{k_L+k_R+1}\pu\rmk$, for all $l\in\nan$.
\item We have $\omega_\infty=\omega_{\bb,\infty}$, where $\omega_{\bb,\infty}$ is given in Lemma 3.16 of Part I.
\end{enumerate}
\end{lem}

The $n$-tuple $\bb$ in the above Lemma is defined as follows.
\begin{lem}\label{lem:form}
Assume [A1],[A3],[A4], and [A5]. 
Then there exist $n_0\in\nan$, $k_L, k_R\in\nan\cup \{0\}$, 
$\oo\in \Primz_u(n,n_0)$, 
$(\lal,\bbD,\bbG,Y)\in \caT(k_R,k_L)$, $l_0\in\nan$,  and $\{x_{\mu,b}^{(L)}\}_{\mu=1,\ldots,n, b=1,\ldots, k_L},\{x_{\mu,a}^{(R)}\}_{\mu=1,\ldots,n, a=1,\ldots, k_R}
\subset\mnz$ satisfying the followings.
\begin{enumerate}
\item Define $\bb\in\lmk \mnz\otimes\Mat_{k_L+k_R+1}\rmk^{\times n}$ by
\begin{align}\label{eq:bdef}
&B_{\mu}=
\omega_\mu\otimes \Lambda_{\lal}\lmk \unit+Y\rmk
+\sum_{b=1}^{k_L} x_{\mu b}^{(L)}\otimes \Lambda_{\lal}\lmk \unit+Y\rmk I_L^{(k_R,k_L)}\lmk G_b\rmk
+\sum_{a=1}^{k_R} x_{\mu a}^{(R)}\otimes I_R^{(k_R,k_L)} \lmk D_a \rmk\Lambda_{\lal}\lmk \unit+Y\rmk,
\end{align}
for $\mu=1,\ldots,n$.
Then 
\begin{align*}
&\caK_l\lmk\bb\rmk \hpu
=\mnz\otimes \spa\left\{ I_R^{(k_R,k_L)} \lmk \unit \rmk, I_R^{(k_R,k_L)} \lmk D_a \rmk, a=1,\ldots,k_R
\right\}\Lambda_{\lal}^l\lmk \unit+Y\rmk^l\\
&\hpd \caK_l\lmk\bb\rmk
=\mnz\otimes \Lambda_{\lal}^l\lmk \unit+Y\rmk^l\spa\left\{ I_L^{(k_R,k_L)} \lmk \unit \rmk, I_L^{(k_R,k_L)} \lmk G_b\rmk, b=1,\ldots,k_L
\right\},
\end{align*}
for all $l\ge l_0$.
\item For  the state $\omega_L$ in  [A4] ,
 $\tau_{l}\lmk s(\omega_{L}\vert_{\caA_{L,l}})\rmk$ is equal to
the orthogonal projection onto \\
$\Gamma_{l,\bb}^{(R)}\lmk \mnz\otimes \pd \Mat_{k_L+k_R+1}\pd \rmk$, for all $l\in\nan$.
\item For  the state $\omega_R$ in  [A4] ,
 $s(\omega_{R}\vert_{\caA_{R,l}})$ is equal to
the orthogonal projection onto \\
$\Gamma_{l,\bb}^{(R)}\lmk \mnz\otimes \pu \Mat_{k_L+k_R+1}\pu\rmk$, for all $l\in\nan$.
\item Let $\rho_0$ be the $T_{\oo}$-invariant state on $\mnz$. Then 
$(\mnz,\oo,\rho_0)$
 right-generates $\omega_\infty$.
\end{enumerate}
\end{lem}

\begin{proof}Let $n_0\in\nan$,
$\oo\in \Primz_u(n,n_0)$, $k_L,k_R\in\nan\cup\{0\}$,
$\hat \vv^{(L)}\in \lmk \Mat_{{n_{0}}}\otimes\Mat_{k_L+1}\rmk^{\times n}$,
$\hat \vv^{(R)}\in \lmk \Mat_{{n_{0}}}\otimes\Mat_{k_R+1}\rmk^{\times n}$,
 $\hat \lal^{(L)}=(\hat  \lambda^{(L)}_b)_{b=0,\ldots,k_L}\in \cc^{k_L+1}$,
$\hat \lal^{(R)}=(\hat \lambda^{(R)}_a)_{a=0,\ldots,k_R}\in \cc^{k_R+1}$,
$\hat Y_L\in\UT_{0, k_L+1} $, $\hat Y_R\in\DT_{0, k_R+1}$,
$\hat D_b^{(L)}\in\UT_{0, k_L+1} $, $\hat D_a^{(R)}\in\DT_{0, k_R+1}$,
$b=1,\ldots,k_L$, $a=1,\ldots, k_R$, and 
$l_{0}\in\nan$,
satisfying the properties {\it 1-10} of Lemma \ref{lem:lrf} and (\ref{eq:ordl}).
We also use $\{\hat x_{\mu,b}^{(L)}\}_{\mu=1,\ldots,n, b=1,\ldots, k_L}\subset\mnz$
and $\{\hat x_{\mu,a}^{(R)}\}_{\mu=1,\ldots,n, a=1,\ldots, k_R}\subset\mnz$ from {\it 3} of Lemma \ref{lem:lrf},
and set $x_{\mu,b}^{(L)}=\hat x_{\mu,b}^{(L)}$ and $x_{\mu,a}^{(R)}=\hat x_{\mu,a}^{(R)}$.

If $k_R\in\nan$, we define 
\[
D_a:=\sum_{i,j=-k_R}^{0} \lmk\hat D_{a}^{(R)}\rmk_{-i,-j}\eijr{i,j},\quad a=1,\ldots, k_R.
\]
As $\hat D_a^{(R)}\in\DT_{0, k_R+1}$, we have $D_a\in \UT_{0, k_R+1}$.
We have 
\[
D_a\eijr{00}=\sum_{i=-k_R}^{0} \lmk\hat D_{a}^{(R)}\rmk_{-i,0}\eijr{i,0}=\sum_{i=-k_R}^{0} \lmk\hat D_{a}^{(R)}E_{00}^{(0,k_R)}\rmk_{-i,0}\eijr{i,0}
=\sum_{i=-k_R}^{0} \lmk E_{a0}^{(0,k_R)}\rmk_{-i,0}\eijr{i,0}=\eijr{-a,0}.
\]
By {\it 7} of Lemma \ref{lem:lrf}, The linear span of 
$\{D_a\}_{a=1}^{k_R}$ is a subalgebra of $\UT_{0,k_R+1}$.
Therefore, $\bbD=(D_1,\ldots,D_{k_R})$ belongs to $\caC^R(k_R)$.
If $k_L\in\nan$, we set $G_b:=\hat D_b^{(L)}$, $b=1,\ldots,k_L$.
That $\bbG\in \caC^L(k_L)$ follows from the corresponding properties of 
$\hat\bbD^{(L)}=(\hat D_1^{(L)},\ldots,\hat D_{k_L}^{(L)})$.

We define $\lal=(\lambda_{-k_R},\ldots, \lambda_{k_L})\in \bbC^{k_L+k_R+1}$ by
\begin{align*}
\lambda_i
:=
\begin{cases}
\hat \lambda_{-i}^{(R)},&i=-k_R,\ldots,0\\
\hat \lambda_{i}^{(L)},&i=1,\ldots,k_L
\end{cases}
.
\end{align*}
Then, by (\ref{eq:ordl}), we have $\lal\in  \Wo(k_R,k_L)$.

Let 
\[
\tilde Y_R:=\sum_{i,j=-k_R}^{0} \lmk \hat Y_R\rmk_{-i,-j}\eijr{i,j}.
\]
Then  we define 
\[
Y:=I_R^{(k_R,k_L)}\lmk \tilde Y_R\rmk+I_L^{(k_R,k_L)}\lmk \hat Y_L\rmk\in\UT_{0,k_R+k_L+1} .
\]

The properties {\it 1,4,8} of Lemma \ref{lem:lrf} guarantees that $(\lal,\bbD,\bbG,Y)\in \caT(k_R,k_L)$.

Now we prove that our $n_0$, $k_L, k_R$, 
$\oo$, 
$(\lal,\bbD,\bbG,Y)$, $l_0$,  and $\{x_{\mu,b}^{(L)}\},\{x_{\mu,a}^{(R)}\}$ satisfy conditions
{\it 1-4} of the Lemma.
 First, note that 
\begin{align}\label{eq:bc}
&B_{\mu^{(l)}}\hpu=\lmk id\otimes I_R^{(k_R,k_L)}\rmk\lmk \sum_{i,j=-k_R}^{0} \lmk\hat v_{\mu^{(l)}}^{(R)}\rmk_{-i,-j}\eijr{i,j}\rmk,\notag\\
&\hpd B_{\mu^{(l)}}=\lmk id\otimes I_L^{(k_R,k_L)}\rmk\lmk \hat v_{\mu^{(l)}}^{(L)}\rmk,
\end{align}
for all $l\in\nan$, and $\mu^{(l)}\in\{1,\ldots,n\}^{\times l}$.
This and {\it 9} of Lemma \ref{lem:lrf} implies {\it 1}.
Furthermore, (\ref{eq:bc}) and {\it 10} of Lemma \ref{lem:lrf} implies {\it 2,3}.
{\it 4} is {\it 11} of Lemma \ref{lem:lrf}.
\end{proof}
Next we would like to show that for this $\bbB$, $\caK_l(\bb)$ coincides with
$
\mnz\otimes\lmk
 \caD(k_R,k_L,\bbD,\bbG)\lmk \Lambda_{\lal}\lmk 1+Y\rmk\rmk^{l}
\rmk
$ for $l$ large enough.
Note that
we have
\begin{align}\label{eq:dd}
&\caV_l:=\mnz\otimes \caD(k_R,k_L,\bbD,\bbG)\lmk \Lambda_{\lal}\lmk 1+Y\rmk\rmk^{l}\notag\\
&=\mnz\otimes
\spa\lmk \lmk \Lambda_{\lal}\lmk 1+Y\rmk\rmk^{l}, \{I_R^{(k_R,k_L)}(D_a)\lmk \Lambda_{\lal}\lmk 1+Y\rmk\rmk^{l}\}_{a=1}^{k_R} \cup\{\lmk \Lambda_{\lal}\lmk 1+Y\rmk\rmk^{l}I_L^{(k_R,k_L)}(G_b)\}_{b=1}^{k_L}\rmk+\cn.
\end{align}
First we show the following inclusion.
\begin{lem}\label{lem:kv}
Assume [A1],[A3],[A4], and [A5]. 
Recall $n_0$, $k_L, k_R$, 
$\oo$, 
$(\lal,\bbD,\bbG,Y)$, $l_0$,  and $\{x_{\mu,b}^{(L)}\},\{x_{\mu,a}^{(R)}\}$
given in Lemma \ref{lem:form}, and $\bb$ defined by (\ref{eq:bdef}).
We use the notation $\check\Lambda:= \Lambda_{\lal}\lmk 1+Y\rmk$.
Then we have
\begin{align}
\caK_l(\bb)\subset 
\mnz\otimes \caD(k_R,k_L,\bbD,\bbG)\cl{l}, \quad l\in\nan.
\end{align}
In particular, for any $l_0\le l$, $\alpha,\beta=1,\ldots,n_0$, $a=0,\ldots,k_R$, $b=0,\ldots,k_L$,
there exist $A^{(R,l)}_{\alpha\beta a}, A^{(L,l)}_{\alpha\beta b}\in \caK_l(\bb)$
of the form
\begin{align}\label{eq:eoa}
&A^{(R,l)}_{\alpha\beta 0}=\zeij{\alpha\beta}\otimes \cl{l}
+\sum_{b=1}^{k_L} z_{\alpha\beta 0 b}^{(R,l)}\otimes \cl{l}\il{G_b}
+O_{\alpha,\beta,0}^{(R,l)}
\notag\\
 &A^{(R,l)}_{\alpha\beta a}=\zeij{\alpha\beta}\otimes \ir{D_{a}}\cl{l}
+\sum_{b=1}^{k_L} z_{\alpha\beta a b}^{(R,l)}\otimes \cl{l}\il{G_b}
+O_{\alpha,\beta,a}^{(R,l)}\notag \\
&A^{(L,l)}_{\alpha\beta 0}=\zeij{\alpha\beta}\otimes \cl{l}
+\sum_{a=1}^{k_R} z_{\alpha\beta a 0}^{(L,l)}\otimes\ir{D_a} \cl{l}
+O_{\alpha,\beta,0}^{(L,l)}
\notag\\
 &A^{(L,l)}_{\alpha\beta b}=\zeij{\alpha\beta}\otimes \cl{l}\il{G_{b}}
+\sum_{a=1}^{k_R} z_{\alpha\beta a b}^{(L,l)}\otimes \ir{D_a}\cl{l}+O_{\alpha,\beta,b}^{(L,l)}.
\end{align}
Here, $z_{\alpha\beta a b}^{(\sigma,l)}\in\mnz$ and $O_{\alpha,\beta,k}^{(\sigma,l)}\in\cn$.
\end{lem}
\begin{proof}
From
(16) of Part I,  it is easy to check $\caV_{l_1}\caV_{l_2}\subset \caV_{l_1+l_2}$.
As $\caK_1(\bb)\subset \caV_1$, this proves the first claim of the Lemma, inductively.
The second claim can be checked from the first one and {\it 1} of Lemma \ref{lem:form}.
\end{proof}

Next we prove the opposite inclusion for large $l$, inductively.
The following Lemma is used for the induction.
\begin{lem}\label{lem:indaa}
Assume [A1],[A3],[A4], and [A5]. We use the notations in  Lemma \ref{lem:form} and $\check\Lambda:= \Lambda_{\lal}\lmk 1+Y\rmk$.
For each $a=0,\ldots,k_R$ and $l\in\nan$, set
\[
W_{a}^{(l)}:=\mnz\otimes\spa\left\{I_R^{(k_R,k_L)} \lmk D_{a'} \rmk, a'>a
\right\}\cl{l}+\cn.
\]
We consider the following two propositions, for $a=0,\ldots, k_R$.
\begin{quote}
$({\mathbb P}_{a})$  : There exists an $l_{a,R}\in\nan$ satisfying the following:  for any $l_{a,R}\le l\in\nan$, there exist
$X_{\alpha\beta a'}^{(l)}\in \caK_l(\bb)$, $\alpha,\beta=1,\ldots,n_0$, $a'=0,\ldots,a$
such that
\begin{align*}
&X_{\alpha\beta 0}^{(l)}-\zeij{\alpha\beta}\otimes \cl{l}\in W_a^{(l)},\\
&X_{\alpha\beta a'}^{(l)}-\zeij{\alpha\beta}\otimes \ir{D_{a'}}\cl{l}\in W_{a}^{(l)}, \quad 1\le a'\le a.
\end{align*}
\end{quote}
\begin{quote}
$(\tilde {\mathbb P}_{a})$  : There exist $\tilde l_{a,R}\in\nan$ and $Z\in  \caK_{\tilde l_{a,R}}(\bb)$ with the form 
\begin{align*}
Z=\sum_{a'\ge a+1} z_{a'}\otimes \ir{D_{a'}}\cl{\tilde l_{a,R}}+\text{ an element in } \cn,
\end{align*}
with
\[
z_{a+1}\neq 0.
\]
\end{quote}
If $({\mathbb P}_{a})$ and $(\tilde {\mathbb P}_{a})$ hold for some $a<k_R$, then $({\mathbb P}_{a+1})$ holds.
\end{lem}
\begin{proof}
First we note the following properties which can be checked from Definition 1.7, 1.8, Remark 1.9 of Part I.
\begin{align}\label{eq:ww}
&W_a^{(l_1)}\lmk\mnz\otimes  \cl{l_2}\rmk,
\lmk\mnz\otimes  \cl{l_2}\rmk W_a^{(l_1)}\subset W_a^{(l_1+l_2)}\notag\\
&W_a^{(l_1)}\lmk\mnz\otimes  \ir{D_b}\cl{l_2}\rmk,
\lmk\mnz\otimes  \ir{D_b}\cl{l_2}\rmk W_a^{(l_1)},
\subset W_{a+1}^{(l_1+l_2)}\notag\\
&W_a^{(l_1)}\cn\subset\cn,\quad \cn W_a^{(l_1)}=0,\quad\cn\cdot\cn=0,\notag\\
&W_a^{(l_1)}\cdot W_{a'}^{(l_2)}\subset W_{\max\{a,a'\}+1}^{(l_1+l_2)},\quad
\hpd W_b^{(l)}=0,
\end{align}
for all $l_1,l_2\in\nan$, $a, a'=0,\ldots,k_R$, and $b=1,\ldots,k_R$.
Now, assume that $({\mathbb P}_{a})$ and $(\tilde {\mathbb P}_{a})$ hold for some $a<k_R$.
Then using (\ref{eq:ww}), we get
\begin{align*}
&\zeij{\alpha_1\beta_1}z_{a+1}\zeij{\alpha_2\beta_2}\otimes\cl{l_1}\ir{D_{a+1}}\cl{\tilde l_{a,R}}\cl{l_2}
=X_{\alpha_1\beta_1 0}^{(l_1)}ZX_{\alpha_2\beta_2 0}^{(l_2)}+
\text{ an element in } W_{a+1}^{(l_1+l_2+\tilde l_{a,R})}\\
&\in \caK_{l_1+l_2+\tilde l_{a,R}}\lmk \bb\rmk+ W_{a+1}^{(l_1+l_2+\tilde l_{a,R})},
\end{align*}
for any $l_1,l_2\ge l_{a,R}$, $\alpha_1,\beta_1,\alpha_2,\beta_2=1,\ldots,n_0$.
Further calculation using (9)-(13) of Part I shows that
\begin{align*}
&\zeij{\alpha_1\beta_1}z_{a+1}\zeij{\alpha_2\beta_2}\otimes\ir{D_{a+1}}\cl{l_1+l_2+\tilde l_{a,R}}\\
&=\lambda_{-a-1}^{-l_1}\zeij{\alpha_1\beta_1}z_{a+1}\zeij{\alpha_2\beta_2}\otimes\cl{l_1}\ir{D_{a+1}}\cl{\tilde l_{a,R}}\cl{l_2}+\text{ an element in } W_{a+1}^{(l_1+l_2+\tilde l_{a,R})}\\
&\in 
\caK_{l_1+l_2+\tilde l_{a,R}}\lmk \bb\rmk+ W_{a+1}^{(l_1+l_2+\tilde l_{a,R})}.
\end{align*}
As $z_{a+1}\neq 0$, there exists $\alpha_2,\beta_1$ such that 
$\braket{\cnz{\beta_1}}{z_{a+1}\cnz{\alpha_2}}
\neq 0$.
Hence we obtain 
\begin{align}\label{eq:aok}
\zeij{\alpha_1\beta_2}\otimes\ir{D_{a+1}}\cl{l_1+l_2+\tilde l_{a,R}}
\in \caK_{l_1+l_2+\tilde l_{a,R}}\lmk \bb\rmk+ W_{a+1}^{(l_1+l_2+\tilde l_{a,R})},\quad l_1,l_2\ge l_{a,R}, \quad \alpha_1,\beta_2=1,\ldots,n_0.
\end{align}
From this, we get
\begin{align*}
 W_{a}^{(l)}
\subset \caK_{l}\lmk \bb\rmk+ W_{a+1}^{(l)},\quad
 2l_{a,R}+\tilde l_{a,R}\le l.
\end{align*} 
Therefore, from
$({\mathbb P}_a)$, 
we obtain
\begin{align*}
\zeij{\alpha\beta}\otimes\cl{l},
\zeij{\alpha\beta}\otimes\ir{D_{a'}}\cl{l}\in \caK_{l}\lmk \bb\rmk+W_{a}^{(l)}
\subset \caK_{l}\lmk \bb\rmk+ W_{a+1}^{(l)}
\end{align*}
for any $1\le a'\le a$, $ 2l_{a,R}+\tilde l_{a,R}\le l$, $\alpha,\beta=1,\ldots,n_0$.
This and (\ref{eq:aok}) implies $(\bbP_{a+1})$ with $l_{a+1,R}=2l_{a,R}+\tilde l_{a,R}$.
\end{proof}
In order to apply Lemma \ref{lem:indaa}, we have to find $Z$ as in $(\tilde \bbP_a)$.
In the proof of the following Lemma, we use Lemma C.7 of Part I to find such $Z$.
\begin{lem}\label{lem:ind}
Assume [A1],[A3],[A4], and [A5]. We use the notations in  Lemma \ref{lem:form}.
Assume that $(\bbP_a)$ holds for some $a<k_R$. Then $(\bbP_{a+1})$ holds.
\end{lem}
\begin{proof}
For $X_{\alpha\beta a'}^{(l)}$ given in $(\bbP_a)$, we define
$\tilde X_{\unit}^{(l)}:=\sum_{\alpha=1}^{n_0} X_{\alpha\alpha 0}^{(l)}\in\caK_l\lmk \bb\rmk$, $l_{a,R}\le l$.
We have 
\[\tilde X_{\unit}^{(l)}
-\unit \otimes \cl{l}\in W_a^{(l)}.
\]
Set $l_0':=l_{aR}+k_L+k_R+1$.
For $A^{(R,l_0)}_{\alpha\beta a+1}$ in Lemma \ref{lem:kv}, using (\ref{eq:ww}), we obtain
\begin{align*}
&\tilde X_{\unit}^{(j+l'_{0})}A^{(R,l_0)}_{\alpha\beta a+1}\tilde X_{\unit}^{(l-j-l'_0-l_0)}\\
&=\zeij{\alpha\beta}\otimes\cl{j+l'_{0}}
\ir{D_{a+1}}\cl{l-j-l'_0}+\sum_{b=1}^{k_L} z_{\alpha,\beta, a+1, b}^{(R,l_0)}\otimes \cl{j+l'_{0}+l_0}\il{G_b}\cl{l-j-l'_0-l_0}\\
&+\text{ an element in } W_{a+1}^{l}\\
&=\lambda_{-a-1}^{j+l'_{0}}\cdot \zeij{\alpha\beta}\otimes
\ir{D_{a+1}}\cl{l}+
\sum_{b=1}^{k_L} \sum_{b'=1}^{k_L}\braket{f_{b}^{(k_R,k_L)}}{\cl{l-j-l'_0-l_0}f_{b'}^{(k_R,k_L)}}
z_{\alpha,\beta, a+1, b}^{(R,l_0)}\otimes \cl{l}\il{G_{b'}}
,\\
&+\text{ an element in } W_{a+1}^{l},
\end{align*}
for all $l\ge 2(k_L+k_R+1)^2+l_0+2l'_{0}$, and $0\le j\le (k_L+k_R+1)^2$.
Here, for the second equality, we used (12), (13) and Remark 1.9 of Part I.
Note that we have $\braket{f_{b}^{(k_R,k_L)}}{\cl{l-j-l'_0-l_0}f_{b'}^{(k_R,k_L)}}
=\lambda_b^{l-l_0'-l_0-j}\sum_{\gamma=0}^{k_L+k_R} {}_{l-l_0'-l_0-j}C_{\gamma}
\braket{f_{b}^{(k_R,k_L)}}{Y^\gamma f_{b'}^{(k_R,k_L)}}
$.
Hence, from Lemma C.7 of Part I, there exists $\xi\in \cc^{(k_L+k_R+1)^2}$ such that
\begin{align*}
\caK_{\tilde l_{a,R}}(\bb)\ni \sum_{j=0}^{(k_L+k_R+1)^2-1}\xi(j)\tilde X_{\unit}^{(j+l'_{0})}A^{(R,l_0)}_{\alpha\beta a+1}\tilde X_{\unit}^{(\tilde l_{a,R}-j-l'_0-l_0)}
=\zeij{\alpha\beta}\otimes
\ir{D_{a+1}}\cl{\tilde l_{a,R}}+\text{ an element in } W_{a+1}^{\tilde l_{a,R}}.
\end{align*}
 where $\tilde l_{a,R}=2(k_L+k_R+1)^2+l_0+2l'_{0}$. 
In other words, $(\tilde \bbP_{a})$ holds.
Applying  Lemma \ref{lem:indaa},  $(\bbP_{a+1})$ holds.
\end{proof}
With this induction step, we obtain the following Lemma.
\begin{lem}\label{lem:kc}
Assume [A1],[A3],[A4], and [A5]. We use the notations in  Lemma \ref{lem:form} and $\check\Lambda:= \Lambda_{\lal}\lmk 1+Y\rmk$.
Then there exists an $l_0'\in\nan$  such that 
\begin{align}
\mnz\otimes \spa\lmk \cl{l}, \{I_R^{(k_R,k_L)}(D_a)\cl{l}\}_{a=1}^{k_R} \cup\{\cl{l}I_L^{(k_R,k_L)}(G_b)\}_{b=1}^{k_L}\rmk
\subset \caK_l(\bb)+\cn,
\end{align}
for any $l_0'\le l\in\nan$.
\end{lem}
\begin{proof}
$(\bbP_{k_R})$ combined with Lemma \ref{lem:kv} (\ref{eq:eoa}) corresponds to the claim. From Lemma \ref{lem:ind},
it suffices to check $(\bbP_0)$. However, this is clear from the existence of $A^{(L,l)}_{\alpha\beta 0}$
in Lemma \ref{lem:kv}.
\end{proof}

\begin{lem}\label{lem:cnkb}
Assume [A1],[A3],[A4], and [A5]. We use the notations in  Lemma \ref{lem:form}.
Then there exists an $\tilde l_0\ge l_0'$ (with $l_0'$ given in Lemma \ref{lem:kc}) such that
\begin{align}
\cn
\subset \caK_l(\bb),
\end{align}
for any $\tilde l_0\le l\in\nan$.
\end{lem}
\begin{proof}
First we claim 
\begin{align}\label{eq:fcl}
&\ir{D_a}\opd=\qu{-(a+1)}\ir{D_a}\opd,\quad a=1,\ldots,k_R,\notag\\
&\opu \il{G_b}=\opu \il{G_b}\qd{b+1},\quad b=1,\ldots,k_L.
\end{align}
To see this, let $-k_R\le i\le 0$ be a number satisfying
$\eij{ii}\ir{D_a}\opd\neq 0$ for $a=1,\ldots,k_R$.
As we have
\[
\ir{\eijr{ii}D_a\sum_{j=-k_R}^{-1}\eijr{jj}}
=\eij{ii}\ir{D_a}\opd\neq 0,
\]
it means there is $j\in\{-k_R,\ldots,-1\}$ such that $\eijr{ii}D_a\eijr{jj}\neq 0$.
By Definition 1.8 (9) in Part I,  this implies
 $\lambda_{i}=\lambda_{-a}\lambda_j$.
In particular, we have $|\lambda_i|=|\lambda_{-a}\lambda_j|<|\lambda_{-a}|$, because $|\lambda_{j}|<1$ for $j\in\{-k_R,\ldots,-1\}$.
As $\lal\in\Wo(k_R,k_L)$, this implies $i\le -a-1$.
This proves the first line of the claim. The second one can be proven similarly.

Let $\check\Lambda:= \Lambda_{\lal}\lmk 1+Y\rmk$ as before.
Note that
\begin{align}\label{eq:dg}
&\ir{D_a}\cl{l_1+l_2}\il{G_b}
=\ir{D_a}\pu\pd\cl{l_1+l_2}\il{G_b}\notag\\
&=\ir{D_a}\eij{00}\cl{l_1+l_2}\il{G_b}
=\ir{D_a}\eij{00}\il{G_b}\notag\\
&= \ir{D_a\eijr{00}}\il{\eijl{00}G_b}
=\ir{\eijr{-a0}}\il{\eijl{0b}}=\eij{-ab},
\end{align}
for any $l_1,l_2\in \nan$ and $a=1,\ldots,k_R$, $b=1,\ldots,k_L$.
By the claim (\ref{eq:fcl}), we have
\begin{align}\label{eq:zg}
&Z\cl{l_2}\il{G_b}
=Z\cl{l_2}\opu\il{G_b}
=Z \cl{l_2}\opu\il{G_b}\qd{b+1}\in 
\opd\Mat_{k_L+k_R+1}\qd{b+1},\notag\\
&Z\in \opd\Mat_{k_L+k_R+1}\opu,\quad b=1,\ldots,k_L,\quad l_2\in\nan,
\end{align}
and
\begin{align}\label{eq:dz}
&\ir{D_a}\cl{l_1} Z
=\ir{D_a}\opd\cl{l_1} Z
= \qu{-a-1}\ir{D_a}\opd\cl{l_1} Z
\in\qu{-a-1}\Mat_{k_L+k_R+1}\qd{b},\notag\\
&Z\in \opd\Mat_{k_L+k_R+1}\qd{b},\quad a=1,\ldots,k_R,b=1,\ldots,k_L,\quad l_1\in\nan.
\end{align}

For $a=1,\ldots,k_R$ and $b=1,\ldots,k_L$, 
we define the subset $\caM_{a,b}$ of $\{1,\ldots,k_R\}\times \{1,\ldots,k_L\}$ by
\begin{align*}
\caM_{a,b}:=
\left\{
(a',b')\mid a'=1,\ldots,k_R,\quad 1\le b'<b
\right\} \cup
\left\{
(a',b')\mid 1\le a'\le a,\quad b'=b
\right\}.
\end{align*}
We also set
\[
\caP_{a,b}:=\mnz\otimes \qu{-a-1}\Mat_{k_L+k_R+1}\qd{b}+\mnz\otimes \opd\Mat_{k_L+k_R+1}\qd{b+1}.
\]
Note that
\[
\mnz\otimes\eij{-a',b'}\subset \caP_{a,b},\quad (a',b')\in 
\lmk \{1,\ldots,k_R\}\times \{1,\ldots,k_L\}\rmk\setminus \caM_{a,b}.
\]

For $(a,b)\in\{1,\ldots,k_R\}\times \{1,\ldots,k_L\}$,
we consider the following proposition:\\\\
$({\mathfrak P}_{a,b})$ : There exists an $l_{a,b}\in\nan$ with $l_0'\le l_{a,b}$ such that
\begin{align}
\mnz\otimes\eij{-a',b'}\subset \caK_l(\bb)
+\caP_{a,b},
\end{align}
for all $(a',b')\in\caM_{a,b}$ and $l\ge l_{a,b}$.

First we show ${\mathfrak P}_{11}$. 
By Lemma \ref{lem:kc}, for any $l_0'\le l_1,l_2\in\nan$ and $\alpha,\beta=1,\ldots,n_0$, 
there exist $Z_{1,\alpha,\beta}^{(l_1)}, Z_{2,\beta,\beta}^{(l_2)}\in \cn$ such that
\begin{align*}
\zeij{\alpha\beta}\otimes \ir{D_1}\cl{l_1}+Z_{1,\alpha,\beta}^{(l_1)}\in \caK_{l_1}(\bb),\quad
\zeij{\beta\beta}\otimes \cl{l_2}\il{G_1}+Z_{2,\beta,\beta}^{(l_2)}\in \caK_{l_2}(\bb).
\end{align*}
By (\ref{eq:ww}), we have
\begin{align}\label{eq:ii}
&\caK_{l_1+l_2}(\bb)\ni\lmk \zeij{\alpha\beta}\otimes \ir{D_1}\cl{l_1}+Z_{1,\alpha,\beta}^{(l_1)}\rmk
\lmk \zeij{\beta\beta}\otimes \cl{l_2}\il{G_1}+Z_{2,\beta,\beta}^{(l_2)}\rmk\notag\\
&=\zeij{\alpha\beta}\otimes \ir{D_1}\cl{l_1+l_2}\il{G_1}
+Z_{1,\alpha,\beta}^{(l_1)}\lmk \zeij{\beta\beta}\otimes \cl{l_2}\il{G_1}\rmk
+\lmk \zeij{\alpha\beta}\otimes \ir{D_1}\cl{l_1}\rmk Z_{2,\beta,\beta}^{(l_2)}.
\end{align}
From (\ref{eq:dg}), (\ref{eq:zg}), and (\ref{eq:dz}), this implies
$
\zeij{\alpha\beta}\otimes\eij{-11}\in \caP_{11}
+\caK_l\lmk\bb\rmk,
$
for any $\alpha,\beta=1,\ldots,n_0$ and $l\ge 2l_0'$.
This proves (${\mathfrak P}_{1,1}$).

Assume that (${\mathfrak P}_{a,b}$) holds for some $a<k_R$ and $b=1,\ldots,k_L$.
We would like to show that  (${\mathfrak P}_{a+1,b}$) holds. 
By (${\mathfrak P}_{a,b}$), we have
\[
\cn\subset \caK_l(\bb)+\caP_{a,b},\quad l\ge l_{a,b}.
\]
This and Lemma \ref{lem:kc} implies
\begin{align*}
\mnz\otimes  \ir{D_{a+1}}\cl{l}, \mnz\otimes \cl{l}\il{G_b}\subset \caK_l(\bb)+\cn
\subset  \caK_l(\bb)+\caP_{a,b},\quad l\ge l_{a,b}.
\end{align*}
Therefore,
 for any $l_{a,b}\le l_1,l_2\in\nan$ and $\alpha,\beta=1,\ldots,n_0$, 
there exist $Z_{1,\alpha,\beta}^{(l_1)}, Z_{2,\beta,\beta}^{(l_2)}\in \caP_{a,b}$ such that
\begin{align*}
\zeij{\alpha\beta}\otimes \ir{D_{a+1}}\cl{l_1}+Z_{1,\alpha,\beta}^{(l_1)}\in \caK_{l_1}(\bb),\quad
\zeij{\beta\beta}\otimes \cl{l_2}\il{G_b}+Z_{2,\beta,\beta}^{(l_2)}\in \caK_{l_2}(\bb).
\end{align*}
From (\ref{eq:dg}), (\ref{eq:zg}), and (\ref{eq:dz}), this implies
$
\zeij{\alpha\beta}\otimes\eij{-(a+1),b}\in \caP_{a+1,b}
+\caK_l\lmk\bb\rmk,
$for any $\alpha,\beta=1,\ldots,n_0$ and $l\ge 2l_{a,b}$.
Hence we have
\[
\mnz\otimes\eij{-a',b'}\subset \caK_l(\bb)
+\caP_{a,b}\subset  \caK_l(\bb)
+\caP_{a+1,b}
\]
for all $(a',b')\in\caM_{a+1,b}$ and $l\ge2 l_{a,b}$.
This proves (${\mathfrak P}_{a+1,b}$).

Assume that (${\mathfrak P}_{k_R,b}$) holds for some $b<k_L$.
We would like to show that  (${\mathfrak P}_{1,b+1}$) holds.
By (${\mathfrak P}_{k_R,b}$), we have
\[
\cn\subset \caK_l(\bb)+\caP_{k_R,b},\quad l\ge l_{k_R,b}.
\]
This and Lemma \ref{lem:kc} implies
\begin{align*}
\mnz\otimes  \ir{D_{1}}\cl{l}, \mnz\otimes \cl{l}\il{G_{b+1}}\subset \caK_l(\bb)+\cn
\subset  \caK_l(\bb)+\caP_{k_R,b},\quad l\ge l_{k_R,b}.
\end{align*}
Therefore,
 for any $l_{k_R,b}\le l_1,l_2\in\nan$ and $\alpha,\beta=1,\ldots,n_0$, 
there exist $Z_{1,\alpha,\beta}^{(l_1)}, Z_{2,\beta,\beta}^{(l_2)}\in \caP_{k_R,b}$ such that
\begin{align*}
\zeij{\alpha\beta}\otimes \ir{D_{1}}\cl{l_1}+Z_{1,\alpha,\beta}^{(l_1)}\in \caK_{l_1}(\bb),\quad
\zeij{\beta\beta}\otimes \cl{l_2}\il{G_{b+1}}+Z_{2,\beta,\beta}^{(l_2)}\in \caK_{l_2}(\bb).
\end{align*}
From (\ref{eq:dg}), (\ref{eq:zg}), and (\ref{eq:dz}), this implies
$
\zeij{\alpha\beta}\otimes\eij{-1,b+1}\in \caP_{1,b+1}
+\caK_l\lmk\bb\rmk$,
for any $\alpha,\beta=1,\ldots,n_0$ and $l\ge 2l_{k_R,b}$.
Hence we have
\[
\mnz\otimes\eij{-a',b'}\subset \caK_l(\bb)
+\caP_{k_R,b}\subset  \caK_l(\bb)
+\caP_{1,b+1}
\]
for all $(a',b')\in\caM_{1,b+1}$ and $l\ge2 l_{k_R,b}$.
This proves (${\mathfrak P}_{1,b+1}$).

Hence we have proven (${\mathfrak P}_{k_R,k_L}$), inductively.
As $\caP_{k_R,k_L}=\{0\}$ and $\caM_{k_R,k_L}=\{1,\ldots,k_R\}\times \{1,\ldots,k_L\}$, we have
\[
\mnz\otimes\eij{-a',b'}\subset \caK_l(\bb),
\]
for all $(a',b')\in\{1,\ldots,k_R\}\times \{1,\ldots,k_L\}$ and $l\ge l_{k_R,k_L}$.
This proves the Lemma.
\end{proof}
\begin{proofof}[Lemma \ref{lem:dercla}]
Let us consider the a septuplet $(n_0,k_R,k_L,\lal, \bbD,\bbG,Y)$ 
and $\oo$ given in Lemma \ref{lem:form}.
We define $\bb$ by {\it 1} of Lemma \ref{lem:form}.
By Lemma \ref {lem:kv}, Lemma \ref{lem:kc}, and Lemma \ref{lem:cnkb}, our $\bb$ belongs to $\ClassA$.
The properties {\it 2,3} of Lemma \ref{lem:dercla} corresponds to {\it 2,3} of Lemma \ref{lem:form}.
By these properties, we have 
\[
\omega_L\circ \tau_x\lmk h_{m',\bb}\rmk=0,\quad x\le -m',\quad
\omega_R\circ \tau_x\lmk h_{m',\bb}\rmk=0,\quad 0\le x,
\]
for $m'\ge m_\bb$,
because $\Ran\Gamma_{l,\bb}^{(R)}$
is a subspace of $\ker \tau^{(R)}_x\lmk h_{m',\bb}\rmk$ for $0\le x\le l-m'$
and $l\ge m'$.
Recall that  if $\psi_\sigma\in\caS_{\sigma}(H)$ for $\sigma=L,R$, 
by Lemma \ref{lem:a1a4}, 
we have $\psi_\sigma\le d_1\cdot\omega_{\sigma}$.
Therefore, we have 
\[
\psi_L\circ \tau_x\lmk h_{m',\bb}\rmk=0,\quad x\le -m',\quad
\psi_R\circ \tau_x\lmk h_{m',\bb}\rmk=0,\quad 0\le x.
\]
This means $\psi_L\in\caS_{(-\infty,-1]}(H_{\Phi_{m',\bb}})$ and  $\psi_R\in\caS_{[0,\infty)}(H_{\Phi_{m',\bb}})$ for $m'\ge m_\bb$.
(Recall Lemma \ref{lem:vr} and the fact that $h_{m',\bb}$ satisfies  [A1]-[A5].)
This proves $\caS_{L}(H)\subset \caS_{(-\infty,-1]}(H_{\Phi_{m',\bb}})$ and 
$\caS_{R}(H)\subset\caS_{[0,\infty)}(H_{\Phi_{m',\bb}})$ for $m'\ge m_\bb$.
Conversely, let  $\psi_L\in\caS_{(-\infty,-1]}(H_{\Phi_{m',\bb}})$ for $m'\ge m_\bb$.
Then, from Lemma 3.15 of Part I, there exists $\sigma_L\in {\mathfrak E}_{n_0(k_L+1)}$ such that
$ \psi_L(A)=\Xi_L(\sigma_L)(A)
:=\sigma_{L}(y_{\bb}^{\frac 12} \bbL_\bb(A)y_{\bb}^{\frac 12})$.
By the definition of $\bbL_\bb$, 
 it is easy to see that
 $\tau_{l}\lmk s\lmk \psi_L\vert_{\caA_{[-l,-1]}}\rmk\rmk$ 
is under the orthogonal projection onto $\Gamma_{l,\bb}^{(R)}\lmk \mnz\otimes \pd \Mat_{k_L+k_R+1}\pd \rmk$
for any $l\in \nan$.
By {\it 2} of  Lemma \ref{lem:form}, this means that the support of $\psi_L\vert_{\caA_{[-l,-1]}}$ is under
$ s(\omega_{L}\vert_{\caA_{[-l,-1]}})$. As $ s(\omega_{L}\vert_{\caA_{[-l,-1]}})$ is under the projection
onto the kernel of $\tau_x(h)$, for $-l\le x\le -m$ by Lemma \ref{lem:vr}, we obtain
$\psi_L\circ\tau_x(h)=0$. Hence we have $\psi_L\in \caS_L(H)$.
Hence we get $\caS_{(-\infty,-1]}(H_{\Phi_{m',\bb}})\subset \caS_L(H) $. The proof for
$\caS_{[0,\infty)}(H_{\Phi_{m',\bb}})\subset \caS_R(H)$ is the same.

To prove {\it 4}, note that $\omega_\infty$ is translation invariant because of
the uniqueness of $\caS_{\bbZ}(H)$.
Note that $\omega_\infty\vert_{\caA_{[0,\infty)}}\in \caS_R(H)=\caS_{[0,\infty)}(H_{\Phi_{m',\bb}})$ because of
{\it 1} and Lemma \ref{lem:vr}. By the translation invariance of $\omega_\infty$ and  Lemma \ref{lem:vr},
this means $\omega_\infty\in\caS_{\bbZ}(H_{\Phi_{m',\bb}})=\{\omega_{\bb,\infty}\}$, $m'\ge m_\bb$.
Hence we have $\omega_\infty=\omega_{\bb,\infty}$.
\end{proofof}
\section{Proof of the main Theorem}\label{sec:hachi}
In this section we complete the proof of Theorem \ref{thm:main} and Corollary \ref{cor:cl}.
\begin{proofof}[ Theorem \ref{thm:main} ]
Let us consider the $\bb\in\ClassA$ with respect to
the septuplet $(n_0,k_R,k_L,\lal, \bbD,\bbG,Y)$ given in Lemma \ref{lem:dercla}.
Recall the definition of
$l_{\bb}(n,n_0,k_R,k_L,\lal,\bbD,\bbG,Y)$ from Part I Definition 1.13.
We denote it by $l_\bb$ throughout the proof.

We fix an arbitrary  $m_1\ge \max\{2l_{\bb}(n,n_0,k_R,k_L,\lal,\bbD,\bbG,Y),\frac{\log \lmk n_0^2(k_L+1)(k_R+1)+1\rmk}{\log n}\}$.
First we claim that
 there exist a constant $\tilde C>0$ and a natural number 
$\tilde N_0\in\nan$ such that
\begin{align}
&\lv
{\Tr_{[0,N-1]}\lmk \lmk 1- G_{N,\bb} \rmk G_N \rmk }
\rv\le \tilde C {s}_1^{\frac N2},\quad  \tilde N_0\le N.
\end{align}
(Here, $s_1$ is given in [A4].)
Recall Theorem 1.18 of Part I.
By (ii) of the latter theorem, there exist $\gamma_{m_1,\bb}>0$ and $\tilde N_{m_1,\bb}\in\nan$
such that
\begin{align}\label{eq:gapb}
\gamma_{m_1,\bb}\lmk 1-G_{N,\bb}\rmk
\le \lmk H_{\Phi_{m_1,\bb}}\rmk_{[0,N-1]},\text{ for all } N\ge \tilde N_{m_1,\bb}.
\end{align}
On the other hand, 
note that $\omega_R\lmk\tau_x\lmk h_{m_1,\bb}\rmk\rmk=0$, $0\le x$
and $\omega_L\lmk\tau_x\lmk h_{m_1,\bb}\rmk\rmk=0$, $x\le -m_1$.
This is because of  Lemma \ref{lem:dercla}, {\it 1} .
Therefore, combining [A1] and [A4] with this, we obtain
\begin{align}\label{eq:epb}
&\lv
{\Tr_{[0,N-1]}\lmk G_{N}\lmk\tau_x\lmk h_{m_1,\bb}\rmk \rmk\rmk}
\rv\le d_1C_1 s_1^{\max\{N-(m_1+x),x\}}
\end{align}
for all $0\le x\le N-m_1$, and $N\ge N_3$.
Set $\tilde N_0:=\max\{\tilde N_{m_1,\bb}, N_3\}$.
By (\ref{eq:gapb}) and (\ref{eq:epb}), we obtain
\begin{align*}
&\gamma_{m_1,\bb}\Tr_{[0,N-1]}\lmk \lmk 1-G_{N,\bb}\rmk G_N\rmk
\le \Tr_{[0,N-1]}\lmk \lmk H_{\Phi_{m_1,\bb}}\rmk_{[0,N-1]}G_N\rmk\\
&=\sum_{0\le x\le N-m_1 }\Tr_{[0,N-1]}\lmk \tau_x\lmk h_{m_1,\bb}\rmk G_N\rmk
\le \sum_{0\le x\le N-m_1 }d_1C_1 s_1^{\max\{N-(m_1+x),x\}}
\end{align*}
for all $N\ge \tilde N_0$.
Therefore, there exists $\tilde C>0$ such that
\begin{align*}
\Tr_{[0,N-1]}\lmk \lmk 1-G_{N,\bb}\rmk G_N\rmk
\le \tilde C
 s_1^{\frac N2},\quad N\ge \tilde N_0.
\end{align*}
This proves the claim.

The same kind of inequality with $G_N$ and $G_{N,\bb}$ interchanged holds, i.e.,
there exist a constant $C'>0$ and a natural number 
$N_0'\in\nan$ such that
\begin{align}
&\lv
{\Tr_{[0,N-1]}\lmk \lmk 1- G_{N} \rmk G_{N,\bb} \rmk }
\rv\le C' {s}_\bb^{\frac N2},\quad  N_{0}'\le N.
\end{align}
Here, $s_\bb$ is given in (iv) of Theorem 1.18 Part I.

Define $0<s<1$, $C>0$, and $N_0\in \nan$ by $s:=\max\{s_1^{\frac 14}, s_\bb^{\frac 14}\}$, 
$C:=\tilde C^{\frac 12}+{C'}^{\frac 12}$, and $N_0:=\tilde N_0+N_0'$.
Then we have
\begin{align*}
&\lV G_N-G_{N,\bb}\rV\le 
\lV G_N\lmk \unit -G_{N,\bb}\rmk\rV+\lV\lmk \unit- G_N\rmk G_{N,\bb}\rV\\
&\le\lmk \Tr_{[0,N-1]}\lmk G_N\lmk \unit -G_{N,\bb}\rmk\rmk\rmk^{\frac 12}
+\lmk \Tr_{[0,N-1]}\lmk G_{N,\bb}\lmk \unit -G_{N}\rmk\rmk\rmk^{\frac 12}
\le
\tilde C^{\frac 12} s_1^{\frac N 4}+{C'}^{\frac 12} s_\bb^{\frac N 4}\le C s^N,
\end{align*}
for all $N\ge N_0$.
\end{proofof}
\begin{proofof}[Corollary \ref{cor:cl}]
Let $\bb\in\ClassA$ and $m_1\in\nan$ given in Theorem \ref{thm:main}.
From [A2], Theorem 1.18 (ii) of Part I, and Theorem \ref{thm:main}, there exists
$\hat \gamma>0$, $\check N_0\in\nan$, $C>0$, and $0<s<1$ such that
\begin{align}\label{eq:cgap}
\hat \gamma\lmk 1-G_N\rmk
\le \lmk H\rmk_{[0,N-1]},\quad
\hat \gamma\lmk 1-G_{N,\bb}\rmk
\le \lmk H_{\Phi_{m_1,\bb}}\rmk_{[0,N-1]},
\quad \lV G_N-G_{N,\bb}\rV\le C s^N,
\end{align}
for all $\check N_0\le N\in\nan$.

We claim the following:
Let  $\check N_0\le N\in\nan$ and $t\in[0,1]$.
Assume that  $\lambda\in (0,\hat\gamma)$ is an eigenvalue of $(1-t)H_{[0,N-1]}+t\lmk H_{\Phi_{m_1,\bb}}\rmk_{[0,N-1]}$ with a
unit eigenvector $\xi\in\bigotimes_{i=0}^{N-1}$.
Then we have
\begin{align}\label{eq:gngnb}
\lV G_N\xi\rV\le \frac{C}{\lambda}\cdot Ns^N,\quad
\lV G_{N,\bb}\xi\rV\le \frac{C}{\lambda}\lV h\rV\cdot Ns^N.
\end{align}
Multiplying
\begin{align}\label{eq:ee}
\lmk(1-t)H_{[0,N-1]}+t\lmk H_{\Phi_{m_1,\bb}}\rmk_{[0,N-1]}\rmk\xi
=\lambda\xi,
\end{align}
by $G_N$,
we obtain
\[
tG_N\lmk H_{\Phi_{m_1,\bb}}\rmk_{[0,N-1]}\xi
=\lambda G_N\xi.
\]
From this, we obtain
\begin{align*}
&\lV G_N\xi\rV=\frac{t}{\lambda}\lV G_N\lmk H_{\Phi_{m_1,\bb}}\rmk_{[0,N-1]}\xi\rV
\le \frac{t}{\lambda}\lV\lmk G_N-G_{N,\bb}\rmk\lmk H_{\Phi_{m_1,\bb}}\rmk_{[0,N-1]}\xi\rV
+\frac{t}{\lambda}\lV G_{N,\bb}\lmk H_{\Phi_{m_1,\bb}}\rmk_{[0,N-1]}\xi\rV\\
&=\frac{t}{\lambda}\lV\lmk G_N-G_{N,\bb}\rmk \lmk H_{\Phi_{m_1,\bb}}\rmk_{[0,N-1]}\xi\rV
\le\frac{t}{\lambda}\sum_{x:0\le x\le N-m_1}\lV G_N-G_{N,\bb}\rV
\le \frac{tC}{\lambda}\cdot Ns^N\le \frac{C}{\lambda}\cdot Ns^N.
\end{align*}
This proves the first inequality of the claim. The proof of the second one is the same.

Let  $\check N_0\le N\in\nan$ and $t\in[0,1]$.
Assume that  $\lambda\in (0,\hat\gamma)$ is an eigenvalue of $(1-t)H_{[0,N-1]}+t\lmk H_{\Phi_{m_1,\bb}}\rmk_{[0,N-1]}$ with a
unit eigenvector $\xi\in\bigotimes_{i=0}^{N-1}$.
We have the following estimation on $\lambda$.
\begin{align}\label{eq:lb}
\hat\gamma\lmk 1-\frac{C}{\lambda}\lmk 1+\lV h\rV\rmk Ns^N\rmk\le
\lambda.
\end{align}
To see this, we use the bound (\ref{eq:cgap}) (\ref{eq:gngnb}) and obtain
\begin{align*}
&\lambda=\braket{\xi}{\lmk(1-t)H_{[0,N-1]}+t\lmk H_{\Phi_{m_1,\bb}}\rmk_{[0,N-1]}\rmk\xi}
\ge (1-t)\hat\gamma\braket{\xi}{(1-G_N)\xi}+t\hat\gamma\braket{\xi}{\lmk 1-G_{N,\bb}\rmk \xi}\\
&= \hat\gamma\lmk 1-(1-t)\braket{\xi}{G_N\xi}-t\braket{\xi}{G_{N,\bb}\xi}\rmk
\ge 
\hat\gamma\lmk
 1-(1-t)\frac{C}{\lambda}\cdot Ns^N-t\frac{C}{\lambda}\lV h\rV\cdot Ns^N
\rmk
\ge
\hat\gamma\lmk
1-\frac{C}{\lambda}\lmk 1+\lV h\rV\rmk \cdot Ns^N
\rmk.
\end{align*}

Set $c_1:=\frac{2C}{\hat\gamma}\lmk 1+\lV h\rV\rmk \lmk \sup_M Ms^{\frac M 4}\rmk$.
It is a positive constant.
We also set $s_1:=s^{\frac 12}$ and $s_2:=s^{\frac 14}$. They satisfy $0<s_1,s_2<1$.
We claim 
\[
\sigma\lmk
(1-t)H_{[0,N-1]}+t\lmk H_{\Phi_{m_1,\bb}}\rmk_{[0,N-1]}
\rmk\cap[\hat \gamma c_1 s_1^{N},\hat \gamma -\hat\gamma s_2^{N}]=\emptyset,\quad
t\in [0,1],\quad N\ge \hat N_0.
\]
We prove this by contradiction. Assume this is not true.
Then there exist $t\in [0,1]$, $N\ge \hat N_0$, and $\lambda\in [\hat \gamma c_1 s_1^{N},\hat \gamma -\hat\gamma s_2^{N}]$
such that $\lambda$ is an eigenvalue of $(1-t)H_{[0,N-1]}+t\lmk H_{\Phi_{m_1,\bb}}\rmk_{[0,N-1]}$ with a unit 
eigenvector $\xi$.
Then, by (\ref{eq:lb}), we have
\begin{align*}
\hat\gamma\lmk 1-\frac{C}{\lambda}\lmk 1+\lV h\rV\rmk Ns^N\rmk
\le \lambda\le \hat \gamma (1-s_2^{N})=
\hat \gamma(1-s^{\frac N4}).
\end{align*}
Compairing the left and the right hand side of this inequality, we obtain
\begin{align*}
\hat \gamma c_1\le \lambda s^{-\frac N2}\le
CNs^{\frac N4} \lmk 1+\lV h\rV\rmk \le \frac{\hat\gamma }{2} c_1,
\end{align*} 
which is a contradiction.
\end{proofof}
\noindent
{\bf Acknowledgment.}\\
{
This work was supported by JSPS KAKENHI Grant Number 25800057 and 16K05171. 
}
\appendix
\section{Notations}\label{sec:nota}
 In addition to the notations given in Subsection 1.1, 1.2, 1.3, and Appendix A of Part I, we use the following notations.
For $\sigma=L,R$,
\begin{align*}
&\tau^{(\sigma)}_x:=\left\{
\begin{gathered}
\tau_{x},\quad\text{if } \sigma=R\\
\tau_{-x},\quad\text{if } \sigma=L
\end{gathered}
\right.,\quad x\in\nan\cup\{0\},
&&\caS_{\sigma}(H)
:=\left\{\begin{gathered}
\caS_{[0,\infty)}(H),\quad\text{if } \sigma=R\\
\caS_{(-\infty,-1]}(H),\quad\text{if } \sigma=L
\end{gathered}
\right.,\\
&\caA_{\sigma}:=
\left\{
\begin{gathered}
\caA_{[0,\infty)},\quad\text{if } \sigma=R\\
\caA_{(-\infty,-1]},\quad\text{if } \sigma=L
\end{gathered}
\right.,
&&\caA_{\sigma,l}:=
\left\{
\begin{gathered}
\caA_{[0,l-1]},\quad\text{if } \sigma=R\\
\caA_{[-l,-1]},\quad\text{if } \sigma=L
\end{gathered}
\right.,\quad l\in\nan,\\
&\bbZ^{(\sigma)}:=
\left\{
\begin{gathered}
\{x\in\bbZ\mid 0\le x\},\quad\text{if } \sigma=R\\
\{x\in\bbZ \mid x\le -m\},\quad\text{if } \sigma=L
\end{gathered}
\right.,
&&\Gamma_\sigma:=
\left\{
\begin{gathered}
{[0,\infty)},\quad\text{if } \sigma=R\\
{(-\infty,-1]},\quad\text{if } \sigma=L
\end{gathered}
\right..
\end{align*}
Here $m\in\nan$ in the definition of $\bbZ^{(\sigma)}$
is the interaction length of our $h$.
Furthermore, for $\mu^{(l)}=(\mu_1,\ldots,\mu_l)\in\{1,\dots,n\}^{\times l}$, $l\in\nan$, we set
\begin{align*}
\mu^{(l,\sigma)}:=\left\{
\begin{gathered}
\mu^{(l)},\quad\text{if } \sigma=R\\
(\mu_l,\mu_{l-1},\ldots,\mu_1),\quad\text{if } \sigma=L
\end{gathered}
\right..
\end{align*}
For $n$-tuple of $d\times d$-matrices $\vv=(v_1,\ldots,v_n)\in\Mat_d^{\times n}$ and $l\in\nan$, we define
$\Gamma_{l,\vv}^{(\sigma)}:\Mat_d\to\bigotimes_{i=0}^{l-1}\cc^n$ by
\[
\Gamma_{l,\vv}^{(\sigma)}\lmk
X
\rmk:=
\sum_{\mu^{(l)}}\lmk
\Tr X\lmk \widehat{v_{\mu^{(l,\sigma)}}}\rmk^*
\rmk
\ws{l},\quad\quad
X\in\Mat_{d}.
\]

\section{Endomorphisms of $B(\caH)$}
Endomorphisms of $B(\caH)$ can be represented by a representation of Cuntz algebra.
\begin{lem}[\cite{arv}]\label{lem:arv}
Let $\caH$ be a separable infinite dimensional Hilbert space, and $n\in \nan$.
Let $\Phi : B(\caH)\to B(\caH)$ be a
unital endomorphism of $B(\caH)$
such that$
\lmk \Phi\lmk B(\caH)\rmk\rmk'
$
is isomorphic to $\Mat_n$.
Then there exist $S_i\in B(\caH)$, $i=1,\ldots,n$ such that
\begin{align}\label{eq:cuntz}
S_i^*S_j=\delta_{ij},\quad
\sum_{j=1}^n S_j x S_j^*=\Phi(x),\quad x\in B(\caH).
\end{align}
\end{lem}
From this Lemma, we obtain the following.
\begin{lem}\label{lem:srep}
Let $\mathfrak A$, $\mathfrak B$ be separable infinite dimensional simple unital $C^*$-algebras, and $n\in\nan$,
such that ${\mathfrak A}=\Mat_n\otimes {\mathfrak B}$.
Let $\omega$ be a pure state on $\mathfrak A$ with GNS triple
$(\caH,\pi,\Omega)$.
Let $\gamma$ be a unital endomorphism of $\mathfrak A$ such that $\gamma({\mathfrak A})=\unit\otimes {\mathfrak B}$.
Assume that $\omega$ and $\omega\circ\gamma$ are quasi-equivalent.
Then there exist $S_i\in B(\caH)$, $i=1,\ldots,n$ such that
\begin{align*}
S_i^*S_j=\delta_{ij},\quad S_iS_j^*=\pi\lmk e_{ij}^{(n)}\otimes\unit \rmk ,\quad
\sum_{j=1}^n S_j \pi\lmk A\rmk S_j^*=\pi\circ\gamma\lmk A\rmk,\quad A\in {\mathfrak A}.
\end{align*}
\end{lem}
\begin{proof}
As  $\mathfrak A$ is a separable infinite dimensional simple unital $C^*$-algebra,
$\caH$ is a separable infinite dimensional Hilbert space.
Note that $E:\pi\lmk {\mathfrak A}\rmk^{''}\to \pi\lmk\Mat_n\otimes \unit\rmk'\cap \pi\lmk {\mathfrak A}\rmk^{''}$,
\[
E(x):=\sum_{i=1}^n \pi(e_{i1}^{(n)}\otimes \unit )x \pi(e_{1i}^{(n)}\otimes \unit),
\]
defines a $\sigma w$-continuous projection.
By Lemma 2.6.8 of \cite{BR1}, we have $\pi\lmk\Mat_n\otimes \unit\rmk'=\pi\lmk\unit\otimes {\mathfrak B}\rmk^{''}$. This and the condition $\gamma({\mathfrak A})=\unit\otimes {\mathfrak B}$ implies that
$\lmk \pi\circ\gamma\lmk {\mathfrak A}\rmk\rmk^{''}$ is a factor.
As $\omega$ and $\omega\circ\gamma$ are quasi-equivalent, $\omega$ is pure,
and $\lmk \pi\circ\gamma\lmk {\mathfrak A}\rmk\rmk^{''}$ is a factor, there exists a $\sigma w$-continuous unital endomorphism
$\Phi:B(\caH)\to B(\caH)$
such that $\Phi\lmk\pi(A)\rmk=\pi\circ\gamma(A)$, for $A\in{\mathfrak A}$,
and 
$\Phi\lmk B(\caH)\rmk=\lmk\pi\circ\gamma({\mathfrak A})\rmk^{''}$.
We have $\lmk \Phi\lmk B(\caH)\rmk\rmk'=\lmk\pi\circ\gamma({\mathfrak A})\rmk^{'}=\pi(\unit\otimes {\mathfrak B})'=\pi\lmk \Mat_n\otimes \unit\rmk$.
Applying Lemma \ref{lem:arv} we obtain $S_i$ satisfying (\ref{eq:cuntz}).
As in the argument of Lemma 3.5 in \cite{Matsui1}, we can deform $S_i$s so that they satisfy
$S_iS_j^*=\pi\lmk e_{ij}^{(n)}\otimes\unit \rmk$.
\end{proof}

\section{Finitely correlated states}\label{app:c}
First we recall the definitions introduced in \cite{fnwpure}.
\begin{defn}
Let $n\in\nan$. 
The triple $({\mathfrak B},\bbE,\rho)$ given by
a finite dimensional $C^*$-algebra $\mathfrak B$, a unital CP map $\bbE:\Mat_n\otimes {\mathfrak B}\to {\mathfrak B}$, and
a faithful state $\rho$ on $\mathfrak B$  such that
$\rho\circ\bbE(\unit\otimes X) =\rho(X)$, $X\in{\mathfrak B}$ is called a standard triple.
For each $A\in\Mat_n$, we define a map $\bbE_A:{\mathfrak B}\to {\mathfrak B}$
by $\bbE_A(X)=\bbE\lmk A\otimes X\rmk$, $X\in{\mathfrak B}$.
A standard triple $({\mathfrak B},\bbE,\rho)$ is minimal if
$\mathfrak B$ has no proper sub $C^*$-algebra, which contains $\unit$ and is
$\bbE_A$-invariant for any $A\in\Mat_n$.
\end{defn}
\begin{defn}
Let $({\mathfrak B},\bbE,\rho)$ be a standard triple, and $\omega$ a state on $\caA_{\bbZ}$.
We say the standard triple $({\mathfrak B},\bbE,\rho)$ right (resp. left) generates $\omega$ if
\[
\omega\lmk\bigotimes_{i=a}^{a+l-1}A_i\rmk
=\rho\circ\bbE_{A_a}\circ\bbE_{A_{a+1}}\circ\cdots\circ \bbE_{A_{a+l-1}}\lmk\unit\rmk,
\]
(resp.
\[
\omega\lmk\bigotimes_{i=a}^{a+l-1}A_i\rmk
=\rho\circ\bbE_{A_{a+l-1}}\circ\bbE_{A_{a+l-2}}\circ\cdots\circ \bbE_{A_{a}}\lmk\unit\rmk,)
\]
for any $a\in\bbZ$, $l\in\nan$, $A_i\in\Mat_n$.
If  $\mathfrak B$ is a $*$-subalgebra of $\Mat_k$ containing unit $\unit$
and 
$\bbE$ is given by an $n$-tuple of matrices $\vv=(v_\mu)_{\mu=1}^n\subset\mk^{\times n}$ as
\[
\bbE \lmk e_{\mu\nu}^{(n)}\otimes X\rmk:=
\lmk v_{\mu}\rmk X\lmk  v_{\nu}\rmk^*,\quad X\in {\mathfrak B},
\] 
we also say that $({\mathfrak B}, \vv,\rho)$ right (resp. left) -generates $\omega$. 
In this case, with a bit of abuse of notation, we say $({\mathfrak B}, \vv,\rho)$ is minimal
if the corresponding
$({\mathfrak B},\bbE,\rho)$ is minimal.
\end{defn}
The formalism  introduced in \cite{Fannes:1992vq}\cite{fnwpure} is the right version, but left version can be defined analogously.

For the class of states we consider, the minimal standard triple is unique up to isomorphism.
\begin{thm}\label{thm:iso}\cite{fnwpure}
Let $n\in\nan$ and $\omega$ be a state on $\caA_{\bbZ}$.
For each $N\in \nan$, let $D_N$ be the density matrix of $\omega\vert_{\caA_{[0,N-1]}}$.
Assume that $\sup_N\rank D_N<\infty$.
Let $({\mathfrak B}_i,\bbE^{(i)},\rho_i)$ $i=1,2$ be standard minimal triples
right (resp. left) generating $\omega$.
Assume that the eigenspace of $1$ for $\bbE^{(i)}_\unit$ is $\cc\unit_{\mathfrak B_i}$ for each $i=1,2$.
Then there exists a $C^*$-isomorphism $\Theta:{\mathfrak B}_1\to{\mathfrak B}_2$ such that
$\bbE^{(2)}\circ\lmk id_{\Mat_n}\otimes\Theta\rmk=\Theta\circ\bbE^{(1)}$.
\end{thm}
From this, we can prove the following.
\begin{lem}\label{lem:ouo}
Let $n, k_1,k_2\in\nan$ and $\oo^{(i)}\in\Primz_u(n,k_i)$.
Let $\rho_i$ be the faithful $T_{\oo^{(i)}}$-invariant state. (See Lemma C.5 of Part I.)
Let $\varphi$ be a state on $\caA_{\bbZ}$.
For each $N\in \nan$, let $D_N$ be the density matrix of $\varphi\vert_{\caA_{[0,N-1]}}$, and assume $\sup_N\rank D_N<\infty$.
Suppose that both of  $(\Mat_{k_1},\oo^{(1)},\rho_1)$ and $(\Mat_{k_2}, \oo^{(2)},\rho_2)$ right-generate $\varphi$.
Then there exist a unitary $U:\cc^{k_1}\to\cc^{k_2}$ and $c\in\bbT$
such that
\[
U\omega_{\mu}^{(1)}=c\omega_{\mu}^{(2)}U,\quad
\mu=1,\ldots,n.
\]
\end{lem}
\begin{proof}
As $\oo^{(i)}\in\Primz_u(n,k_i)$,
$(\Mat_{k_i},\oo^{(i)},\rho_i)$ is a minimal standard triple, for each $i=1,2$.
Furthermore, the eigenspace of $1$ for $T_{\oo^{(i)}}$ is $\bbC\unit_{\Mat_{k_i}}$,
because of the primitivity of $\oo^{(i)}$.
Therefore, we may apply Theorem \ref{thm:iso}.
By Theorem \ref{thm:iso}, there is a $*$-isomorphism $\Theta:\Mat_{k_1}\to \Mat_{k_2}$
such that 
\begin{align}\label{eq:itt}
\omega_{\mu}^{(2)}\Theta\lmk X\rmk\lmk \omega_{\nu}^{(2)}\rmk^*
=\Theta \lmk \omega_{\mu}^{(1)} X\lmk \omega_{\nu}^{(1)}\rmk^*\rmk,\quad
\mu,\nu=1,\ldots,n,\quad X\in\Mat_{k_1}.
\end{align}
The rest of the proof is an easy case of Lemma \ref{lem:irpri}.
\end{proof}
\section{CP maps}
\begin{lem}\label{lem:irr}
Let $n\in\nan$ and $T:\Mat_n\to\Mat_n$ be 
the irreducible unital CP map.
Then the followings hold.
\begin{enumerate}
\item There exists $b\in\nan$ such that
$\sigma(T)\cap\bbT=\left\{ \exp\lmk \frac{2\pi i}{b}k\rmk\mid
k=0,\ldots,b-1\right\}$.
\item
For any $\lambda\in \sigma(T)\cap\bbT$, $\lambda$ is a nondegenerate
eigenvalue of $T$.
\item
There exists a unitary matrix  $U\in\Mat_n$ such that 
\[
T\lmk U^k\rmk=\exp\lmk \frac{2\pi i}{b}k\rmk U^k,\quad k=0,\ldots,b-1.
\]
\item The unitary matrix $U$ in {\it 3} has a spectral decomposition
\[
U=\sum_{k=0}^{b-1}\exp\lmk \frac{2\pi i}{b}k\rmk P_k,
\]
with spectral projections satisfying
\[
T\lmk P_k\rmk=P_{k-1},\quad\mod b.
\]
\item
The restriction $T^b\vert_{P_k \Mat_n P_k}$ of $T^b$ on 
$P_k \Mat_n P_k$ defines a primitive unital CP map on
$P_k \Mat_n P_k$.
\item
There exists a faithful $T$-invariant state $\varphi$.
\end{enumerate}
\end{lem}
\begin{proof}
See \cite{Wolf:2012aa}, for example.
\end{proof}

\section{A useful fact on ${}_lC_{k}$.}

\begin{lem}\label{lem:ccac}
Let $l\in\nan$, $m_1,m_2\in \nan\cup\{0\}$
such that $m_1+m_2\le l$.
Then we have
\[
{}_lC_{m_1}\cdot{}_lC_{m_2}=\sum_{k=0}^{m_1+m_2}
\alpha_{(m_1,m_2)}^{(k)}\cdot {}_lC_k,
\]
where
\begin{align*}
\alpha_{(m_1,m_2)}^{(k)}
=
\begin{cases}
{}_kC_{m_2}\cdot\sum_{j=0}^{m_2}\delta_{m_1,k-j}\cdot {}_{m_2}C_{j},&\text{ if } k\ge m_2\\
0 ,&\text{ if } k< m_2
\end{cases}.
\end{align*}
\end{lem}
\begin{proof}
This can be checked inductively.
\end{proof}

\end{document}